\documentclass{scrartcl}

\usepackage{booktabs} 

\usepackage{amsmath,amssymb,amsthm, mathtools}
\usepackage{hyperref}
\usepackage[nameinlink,capitalize,noabbrev]{cleveref}
\usepackage[square,numbers]{natbib}
\usepackage{booktabs}
\usepackage{algorithm}
\usepackage{algorithmicx}
\usepackage{algpseudocode}
\usepackage{bbm}
\usepackage{verbatim}
\usepackage{subcaption}
\usepackage{paralist}
\usepackage{pgfplotstable}
\usepackage{pgfplots}
\usepackage{pifont}
\usepackage{lscape}
\usepackage[]{todonotes}
\usepackage{scrtime}
\usepackage{eso-pic}
\usepackage{authblk}

\usepackage{tikz}
\usetikzlibrary{decorations.pathreplacing,backgrounds}
\usetikzlibrary{fit,decorations.pathreplacing,arrows,patterns,automata,calc,plotmarks}
\pgfplotstableset{
    create on use/errors/.style={
        create col/copy column from table={error.dat}{RunTimeBKD}
    }
}
\pgfmathdeclarefunction{lg10}{1}{%
    \pgfmathparse{ln(#1)/ln(10)}%
}
\newtheorem{theorem}{Theorem}

\newtheorem{lemma}{Lemma}

\newtheorem{proposition}{Proposition}

\theoremstyle{definition}
\newtheorem{definition}{Definition}

\algnewcommand{\LeftComment}[1]{\Statex \textcolor{blue}{\(\triangleright\) #1}}
\newcommand{\iset}{\mathcal{I}}
\newcommand{\jset}{\mathcal{J}}
\newcommand{\smi}{\smallsetminus} 
\newcommand{\non}[2]{\overline{N}^\Delta(#1,#2)}

\newcommand{\sqin}{%
  \mathrel{\vphantom{\sqsubset}\text{%
    \mathsurround=0pt
    \ooalign{$\sqsubset$\cr$-$\cr}%
  }}%
}

\newcommand{\update}{\textsc{Update}}
\newcommand{\updatepool}{\textsc{UpdatePool}}
\newcommand{\bkd}{HMNS}
\newcommand{\viard}{VML}
\newcommand{\dkp}{$\Delta$-$k$-plex}
\newcommand{\dkps}{$\Delta$-$k$-plexes}
\newcommand{\BKB}{$\Delta$-$k$-\textsc{Bron\-Kerbosch}}

\newcommand{\dnn}{$\Delta$-non-neighbor}
\newcommand{\dnns}{$\Delta$-non-neighbors}
\newcommand{\dnnh}{$\Delta$-non-neighborhood}

\newcommand{\dfr}{$\Delta$-frame}

\newcommand{\dsd}{$\Delta$-slice degeneracy}

\usepackage{etoolbox}

\newcommand{\appsymb}{$\star$}
\newcommand{\appref}[1]{{\hyperref[proof:#1]{\appsymb}}}
\newcommand{\toappendix}[1]{\gappto{\appendixProofText}{#1}}

\begin{document}
\title{Listing All Maximal $k$-Plexes in Temporal Graphs
	\footnote{An extended abstract of this work appeared in the proceedings of the 2018 IEEE/ACM International Conference on Advances in Social Networks Analysis and Mining (ASONAM '18)~\cite{BHMMNS16}. This version contains all proof details, extended experimental findings, and the analysis of a new version of our algorithm that fixed a small bug in the code (which has no large impact on the results).}}

\author{Matthias Bentert}
\author{Anne-Sophie Himmel\thanks{Supported by the DFG, projects DAMM (NI 369/13) and FPTinP (NI 369/16).}}
\author{Hendrik Molter\thanks{Supported by the DFG, project MATE (NI 369/17).}}
\author{Marco~Morik}
\author{Rolf Niedermeier}
\author{Ren\'e Saitenmacher}
\affil{\small Institut f\"ur Softwaretechnik und Theoretische Informatik, TU~Berlin, Berlin, Germany

\texttt{\{anne-sophie.himmel, matthias.bentert, h.molter, rolf.niedermeier\}@tu-berlin.de}

\texttt{\{marco.t.morik, r.saitenmacher\}@campus.tu-berlin.de}}
\date{}
\maketitle

\begin{abstract}
Many real-world networks evolve over time, that is, new contacts appear and old contacts may disappear. They can be modeled as temporal graphs where interactions between vertices (which represent people in the case of social networks) are represented by time-stamped edges. One of the most fundamental problems in (social) network analysis is community detection, and one of the most basic primitives to model a community is a clique.
Addressing the problem of finding communities in temporal networks, Viard et al.~[TCS~2016] introduced $\Delta$-cliques as a natural temporal version of cliques. Himmel et al.~[SNAM~2017] showed how to adapt the well-known Bron-Kerbosch algorithm to enumerate $\Delta$-cliques. We continue this work and improve and extend the algorithm of Himmel et al. to enumerate temporal $k$-plexes (notably, cliques are the special case $k=1$). 

We define a \dkp~as a set of vertices and a time interval, where during this time interval each vertex has in each consecutive~$\Delta + 1$ time steps at least one edge to all but at most $k-1$ vertices in the chosen set of vertices.
We develop a recursive algorithm for enumerating all maximal \dkps~and perform experiments on real-world social networks that demonstrate the practical feasibility of our approach. 
In particular, for the special case of $\Delta$-1-plexes (that is,~$\Delta$-cliques), we observe that our algorithm is on average significantly faster than the previous algorithms by Himmel et al.~[SNAM~2017] and Viard et al.~[IPL~2018] for enumerating $\Delta$-cliques.
\end{abstract}

\newpage

\section{Introduction}
Community detection in networks is a highly active research area. In the probably most basic version, a community is modeled as a clique, that is, every vertex is connected to every other vertex in the clique. The concept is not only used for detecting communities in social networks, but it has also applications in ad hoc wireless networks~\cite{chen2004clustering} or biochemistry and genomics~\cite{butenko2006clique}. Cliques as a mathematical model, however, are often too restrictive for real-world applications, where some edges in communities might not exist because of errors in measurements or application-specific reasons. To circumvent this fact, the clique concept has seen several relaxations. Our work focuses on a popular degree-based relaxation of cliques known as $k$-plexes~\cite{seidman1978graph,guo2010more,conte2017fast,xiao2017fast,moser2012exact}. In a $k$-plex, every vertex must be adjacent to all but at most $k-1$ vertices in the~$k$-plex (excluding itself). A $1$-plex is a clique and in a $2$-plex every vertex can have a missing edge to one other vertex in the $2$-plex.
One can use $k$-plexes also as a tool for link-prediction, as the missing edges are probably good candidates for missing links in social networks:
It has been observed that friends of friends tend to become friends themselves \cite{ma2016playing}.

Previous work on $k$-plexes uses \emph{static} graph models~\cite{seidman1978graph,balasundaram2011clique,guo2010more,wu2007parallel,mcclosky2012ostergard,moser2012exact,pattillo2013clique,berlowitz2015efficient,conte2017fast,xiao2017fast,conte2018d2k}. Nowadays, however, an increasing amount of real-world data sets are time-labeled.
For example, in communication networks, such as email networks, the data is frequently time-stamped. A static network cannot distinguish at which time an email was sent and whether there are several emails sent between two persons. Modeling with static graphs is therefore often too restrictive. 
For community detection in temporal graphs, the concept of $\Delta$-cliques~\cite{viard2016deltacliquesjournal} has been introduced and studied~\cite{himmel2017adapting,viard2016deltacliquesjournal,ViardML18}. In a $\Delta$-clique, during its lifetime, each vertex has contact to each other vertex of the $\Delta$-clique at least once every $\Delta+1$ consecutive time steps. We extend this concept by allowing to have up to~$k-1$ missing edges per vertex during each interval of $\Delta$ consecutive time steps. For a formal definition we refer to \cref{sec:prelim}.


\subsection{Related Work}
There has been extensive research on both clique enumeration in temporal graphs~\cite{himmel2017adapting,viard2016deltacliquesjournal,ViardML18} and maximum $k$-plex detection~\cite{mcclosky2012ostergard,moser2012exact,xiao2017fast,balasundaram2011clique} and $k$-plex enumeration~\cite{conte2018d2k,wu2007parallel,berlowitz2015efficient,conte2017fast} in static graphs. 
To the best of our knowledge, we are the first to investigate the enumeration of $k$-plexes in temporal graphs. 
We follow up on the work of Himmel et al.~\cite{himmel2017adapting}, where the famous \textsc{BronKerbosch} algorithm \cite{bron1973algorithm} to enumerate cliques in static graphs was lifted to the temporal setting. The problem of finding $\Delta$-cliques in temporal graphs was introduced and motivated by the study of Viard et al.~\cite{viard2016deltacliquesjournal} who enumerated contact patterns among high-school students. 

The concept of $k$-plexes is due to Seidman and Foster~\cite{seidman1978graph}. To find maximum $k$-plexes, there are several combinatorial branch-and-cut approaches~\cite{mcclosky2012ostergard,moser2012exact,xiao2017fast} as well as ILP-based algorithms~\cite{balasundaram2011clique}. 
The \textsc{BronKerbosch} algorithm \cite{bron1973algorithm} has been adapted to enumerate $k$-plexes in static graphs~\cite{wu2007parallel,conte2018d2k}, but there are also enumeration algorithms based on other approaches~\cite{berlowitz2015efficient,conte2017fast}.
To the best of our knowledge, the currently fastest algorithm for finding a maximum-cardinality~$k$-plex in a static graph is due to Xiao et al.~\cite{xiao2017fast} and the fastest algorithm for listing all maximal~$k$-plexes in a static graph is due to Berlowitz~et~al.~\cite{berlowitz2015efficient}.
There are several other clique relaxations.
Typically, the corresponding decision problems are NP-complete. For more details on different clique relaxations we refer to Patillo et al.~\cite{pattillo2013clique}.

Aside from community detection, many other problem areas also have been studied in the context of temporal graphs, including connectivity problems~\cite{kempe2000connectivity,Mertzios13,akridaTOCS17,FluschnikMNZ18,ZschocheFMN18}, graph exploration~\cite{erlebach15,ErlebachS18}, clustering~\cite{CMSS2018}, and covering problems~\cite{Akrida-ICALP18,MMZ2019,Temp-Matching-19-arxiv}. For an extended overview on research related to temporal graphs, we refer to the surveys of Holme and Saram\"aki~\cite{holme2012temporal}, Casteigts and Flocchini~\cite{flocchini1,flocchini2}, Michail~\cite{michail2016introduction}, and Latapy et al.~\cite{latapy2017stream}.

\subsection{Our Contributions}
Regarding theory, we formally define~\dkps{}, adapt and extend an existing recursive algorithm specialized to $\Delta$-cliques \cite{himmel2017adapting} to enumerate them, prove its correctness (\Cref{theorem:correct}), and present a worst-case running time analysis of our new algorithm (\Cref{thm:runtime}).
In particular, our running-time analysis shows that 
our algorithm has polynomial running time for constant $k$ if the input graph has constant \dsd{}, a measure for sparseness of temporal graphs~\cite{himmel2017adapting}.

Regarding practice, we present and evaluate two heuristic speed-up techniques.
\begin{itemize}
\item We propose a pivoting strategy to reduce the number of recursive calls of the algorithm.
\item We present a strategy that does not enumerate all maximal \dkps{} but only those which might be ``of interest''. This excludes for example~$\binom{n}{k}$ trivial solutions induced by any set of~$k$~vertices over the whole lifetime of the graph (where $n$ is the total number of vertices). 
\end{itemize}
The paper is organized as follows. 
In \cref{sec:prelim} we introduce some definitions and notations used throughout the paper. 
We further explain the original \textsc{BronKerbosch} algorithm which serves as a role model for the  $\Delta$-clique algorithm \cite{himmel2017adapting} and, consequently, for our algorithm. 
In \cref{sec:ourAlg}, we present our adaptation of the \textsc{BronKerbosch} algorithm for enumerating all maximal \dkps{} in temporal graphs. 
After a detailed description of the algorithm, we prove the correctness and analyze the running time of the algorithm. We further utilize the parameter $\Delta$-slice degeneracy~\cite{himmel2017adapting} of a temporal graph to upper-bound the running time. We continue with heuristic tricks to improve the running time of the algorithm.
In \cref{sec:experiment}, we conduct an experimental analysis of the algorithm on real-world data sets. 
We study the effectiveness of our heuristics and compare the running time of our algorithm to the running time of the algorithm by \citet{ViardML18} for $\Delta$-cliques ($\Delta$-$1$-plexes). We then study the number and characteristics of \dkps{} in our data sets. 
\Cref{sec:conclusion} presents our conclusion and \cref{sec:appendix} is an appendix containing additional diagrams and tables of our experimental findings.

\section{Preliminaries}
\label{sec:prelim}
We provide notation for (time) intervals and temporal graphs.
For static graphs, we refer to the book of Diestel \cite{diestel2000graph}. 
We also give a short description of the classic \textsc{BronKerbosch} algorithm~\cite{bron1973algorithm}.

\subsection{Intervals and Sets of Intervals}
\label{sub:intervals}
We refer to an interval as a contiguous ordered set of discrete time steps.
    Formally, an \emph{interval} is an ordered set $$I=[a,b]:=\{n \mid n \in \mathbb{N}
    \land a \leq n \leq b \},$$where~$a,b \in \mathbb{N}$. Further, let $[a]:=[1,a]$.

For a set $A\subseteq\mathbb{N}$, we say that an interval $I \subseteq A$ is \emph{maximal} with respect to $A$ if there is no larger interval $I' \subseteq A$ such that~$I\subset I'$.
If $A$ is not contiguous, then it contains multiple maximal intervals.

A \emph{set $\iset$ of intervals} is an ordered set of $n$ pairwise disjoint intervals, that is,~$\iset = \{I_i \mid i \in [n]\}$ with intervals $I_i = [a_i, b_i]$ where for all $i \in [n-1]$ it holds that $b_i < a_{i+1} - 1$.
We will use $I,J$ to refer to intervals, and $\iset, \jset$ to refer to sets of intervals.
Next, we define the operations \emph{union}, \emph{intersection}, and \emph{difference} for
sets of intervals as a straight-forward extension of the standard set operations on intervals.
We also introduce the notion of an interval respectively a set of intervals being \emph{covered} by a set of intervals.

\begin{definition}\label{def:def}
Given two sets $\iset$ and $\jset$ of intervals,
\begin{itemize}
\item an interval $I$ is \emph{covered} by $\iset$ (i.e.\ $I \sqin \iset$) if there exists an~$I' \in \iset$ such that
$I \subseteq I'$;
\item an interval-set $\iset$ is \emph{covered} by $\jset$ (i.e.\ $\iset \sqsubseteq \jset$) if for all $I \in \iset$ it holds that $I \sqin \jset$;
\item the \emph{union} of $\iset$ and $\jset$ is the set of intervals defined by\\
\[
\iset\sqcup\jset:= \big\{I' \mid I' \text{ is maximal w.r.t.\
}\big((\bigcup_{I\in\iset}I)\cup(\bigcup_{J\in\jset}J)\big)\big\};
\]
\item the \emph{intersection} of $\iset$ and $\jset$ is the set of intervals\\
\[
\iset\sqcap\jset:= \big\{I' \mid I' \text{ is maximal \ w.r.t.\ }
\big((\bigcup_{I\in\iset}I)\cap(\bigcup_{J\in\jset}J)\big)\big\};
\]
\item and $\iset$ \emph{minus} $\jset$ is the set of intervals\\
\[
\iset\smi\jset:= \big\{I' \mid I' \text{ is maximal\ w.r.t.\ }
\big((\bigcup_{I\in\iset}I)\setminus(\bigcup_{J\in\jset}J)\big)\big\}.
\]
\end{itemize}
\end{definition}
To make \Cref{def:def} more accessible, we give some examples:
\[[1,2]\sqin \{[1,4]\}; \ \{[1,2]\}\sqsubseteq \{[1,4]\};\]
\[\{[1,2],[5,8]\}\sqcup \{[1,4],[5,6]\}=\{[1,4],[5,8]\};\]
\[\{[1,2],[5,8]\}\sqcap \{[1,4],[5,6]\}=\{[1,2],[5,6]\};\]
\[\{[1,4],[5,8]\}\smi \{[1,2],[5,6]\}=\{[3,4],[7,8]\}.\]

We refer to a tuple $(v,I_v)$ with $v$ being a vertex and $I_v$ being an interval as a 
\emph{vertex-interval pair} and to a tuple $(v, \mathcal I)$ with a vertex $v
\in V$ and a set $\mathcal I$ of intervals as a \emph{vertex-interval-set pair}. 
For a set~$A$ of vertex-interval-set pairs, we define $V(A)$ to be the set of
all vertices which are contained in a vertex-interval-set pair in~$A$.
We further use the following notation. Let $X,Y$ be sets of vertex-interval-set pairs, let~$(v,\mathcal I_v)$ be a vertex-interval-set pair, let $I \subseteq
T$ be an interval, and let~$V' \subseteq V$ be a set of vertices:
\begin{itemize}
\item $X[\mathcal I]:= \{ (u, \mathcal I_u \sqcap \mathcal I) \mid (u,\mathcal
I_u) \in X\}$;
\item $X[V']:= \{ (u, \mathcal I) \mid (u,\mathcal I) \in X \wedge u \in V'\}$;
\item 
$ (v,\mathcal I_v) \sqcup Y :=
\begin{cases}
Y \cup \{(v,\mathcal I_v)\} &\text{ if } \not \exists (v,\mathcal I^Y_v )\in Y \\
(Y \setminus \{(v,\mathcal I^Y_v )\}) \cup \{(v, \mathcal I_v \sqcup \mathcal I^Y_v )\} &\text{ otherwise; }
\end{cases}$
\item $X \sqcap Y := \{ (u, \mathcal I^X_u \sqcap \mathcal I^Y_u ) \mid
(u,\mathcal I^X_u)  \in X \wedge (u,\mathcal I^Y_u ) \in Y\}$; and
\item $(v,\mathcal I_v) \sqcap Y := \{(v,\mathcal I_v)\} \sqcap Y$.
\end{itemize}

\subsection{Temporal Graphs}
A temporal graph~\cite{michail2016introduction}, also referred to as temporal network \cite{holme2012temporal} or link stream \cite{latapy2017stream}, is a graph whose edge set changes over time.
A temporal graph can be seen as a sequence of static graphs over a fixed set of vertices.
\begin{definition}
A \emph{temporal graph} $G=(V,E, \omega)$ is a triple consisting
of a set~$V$ of vertices, a lifetime~$\omega$, and a
set~$E\subseteq\binom{V}{2}\times [\omega]$  of time-stamped edges. 
\end{definition}
Throughout this paper, we will assume that each vertex-interval pair~$(v,I_v)$ satisfies~$v \in V$ and~$I_v \subseteq [\omega]$, and each vertex-interval-set pair~$(v,\iset)$ satisfies~$v \in V$ and~$I_v \subseteq [\omega]$ for each~$I_v \in \iset$.

\subsubsection{$\Delta$-Frame and $\Delta$-Non-Neighborhood}
A \emph{$\Delta$-frame} is an interval of (consecutive) time steps, that is, each~$\Delta$-frame~$\Delta_i$ with~$i \in [\omega-\Delta]$ corresponds to the interval~${[i,i+\Delta]}$ of time steps.
In order to properly define~\dkp{es}, we need to adjust the notion of neighborhood from static to temporal graphs.
Instead of just considering the incident edges of a vertex at one time step, we consider all incident edges within a $\Delta$-frame.
We say that two vertices $u$ and $v$ are \emph{neighbors} in~$\Delta_i$ if there is an edge~$(\{u,v\},t)\in E$ with~$t \in \Delta_i$. 
Accordingly, we say that $u$ and $v$ are \emph{non-neighbors} in~$\Delta_i$ 
if there exists no edge $(\{u,v\},t)\in E$ with~$t \in \Delta_i$.

The \emph{$\Delta$-non-neighborhood}~$\non{v}{i}$ of a vertex~$v \in V$ and a $\Delta$-frame $\Delta_i$
 is the set of all non-neighbors of $v$ in $\Delta_i$. 
 More formally, we arrive at the following definition.
 \begin{definition}
Let $G=(V,E,\omega)$ be a temporal graph, let~$v \in V$ be a
vertex, and let~$\Delta_i$, $i \in [\omega-\Delta]$, be a $\Delta$-frame in~$G$.
The \emph{$\Delta$-non-neighborhood} of $v$ in $\Delta_i$ is
  \begin{align*}
  \non{v}{i} := \{ w \in V \mid \forall t \in \Delta_i : (\{v,w\},t)
  \not \in E\}.
 \end{align*}
\end{definition}
Accordingly, we define the \emph{$\Delta$-non-neighborhood}~$\non{v}{\mathcal I_v}$ of a vertex-interval-set pair~$(v,\mathcal I_v)$
 as the set of (time-maximal) vertex-interval-set pairs~$(u,\mathcal I_u)$ such that there is no edge between~$u$ and~$v$ in
any~$\Delta$-frame~$\Delta_i, i \in I$, for all $I \in \mathcal I_u \sqcap \mathcal I_v$. See
\cref{subfig:Nona,subfig:Nonb,subfig:Nonc} for an example.
Formally, it is defined as follows.
\begin{definition}
Let $G=(V,E,\omega)$ be a temporal graph, let~$(v, \mathcal I_v)$ be a
vertex-interval-set pair of~$G$.
The \emph{$\Delta$-non-neighborhood} is defined as
$$ \non{v}{\mathcal I_v} := \{(u,\mathcal I_u) \mid u \in V \land \mathcal I_u = \{I \mid I \text{ is maximal w.r.t.\ } \bigcup_{i\in J, J\in \mathcal I_v} \{i \mid u \in \non{v}{i}\}\}\}.$$
\end{definition}
 

\subsubsection{$\Delta$-$k$-Plex}
We define a \emph{\dkp} as a straightforward relaxation of a $\Delta$-clique defined by Viard et al.~\cite{viard2016deltacliquesjournal}.
Given a temporal graph $G=(V,E,\omega)$, a $\Delta$-clique consists of a set~$C$ of vertices and a lifetime $I=[a,b]$.
Each two vertices~$u,v \in C$ are required to be neighbors within any~$\Delta$-frame~$\Delta_i, i\in I$.

Analogously to~$k$-plexes in static graphs, \dkps{} are defined so that each vertex in the vertex set~$C$ of the \dkp{} must have at least $|C|-k$ neighbors in each~$\Delta$-frame~$\Delta_i, i \in I$.
See \cref{subfig:Kplex} for an illustration of a \dkp{}.
\begin{definition}
  Given a temporal graph $G=(V,E,\omega)$, ${\Delta \in \mathbb{N}}$, a subset~${C\subseteq V}$ of vertices, and an interval $I=[a,b] \subseteq [\omega-\Delta]$, then
    $R = (C,I)$ is a \emph{$\Delta$-$k$-plex} if for all $v \in C$
    and all $\Delta_i,~i \in I$, it holds that~${|\non{v}{i} \cap
    C| \leq k}$.
\end{definition}

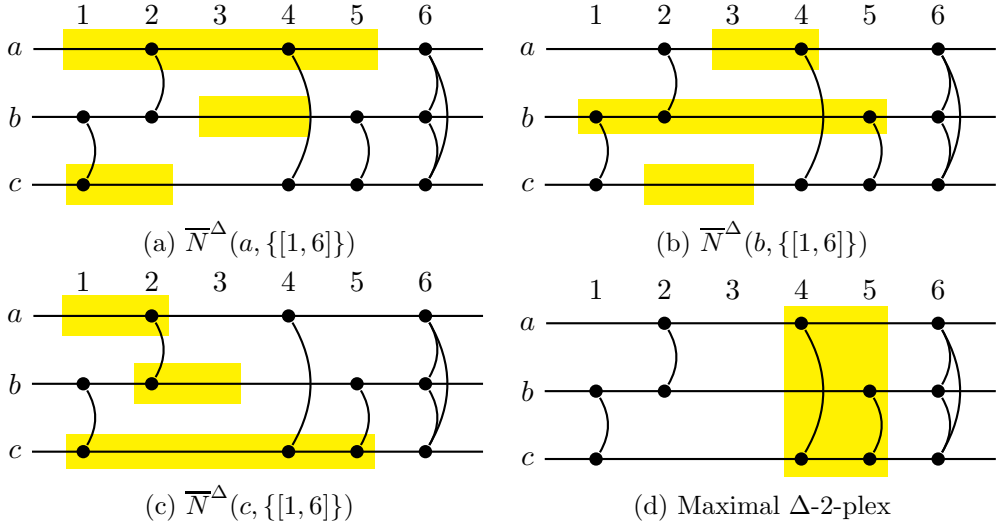
\begin{figure}[t]
    \centering
    \begin{subfigure}[b]{.45\linewidth}
    \centering
    \begin{tikzpicture}[thick,scale=.9]
        \tikzstyle{phase} = [fill,shape=circle,minimum size=5pt,inner sep=0pt];
        
        \node (time1) at (0,0.5) {1};
        \node (time2) at (1,0.5) {2};
        \node (time3) at (2,0.5) {3};
        \node (time4) at (3,0.5) {4};
        \node (time5) at (4,0.5) {5};
        \node (time6) at (5,0.5) {6};
        \node at (-1,0) (q1) {$a$};
        \node at (-1,-1) (q2) {$b$};
        \node at (-1,-2) (q3) {$c$};
        \node (node0a) at (0,0) {};
        \node[phase] (node0b) at (0,-1) {};
        \node[phase] (node0c) at (0,-2) {};
        \draw[-] (node0b) to[out=-60,in=60] (node0c);
        \node[phase] (node1a) at (1,0) {};
        \node[phase] (node1b) at (1,-1) {};
        \node (node1c) at (1,-2) {};
        \draw[-] (node1a) to[out=-60,in=60] (node1b);
        \node (node2a) at (2,0) {};
        \node(node2b) at (2,-1) {};
        \node (node2c) at (2,-2) {};

        \node[phase] (node3a) at (3,0) {};
        \node (node3b) at (3,-1) {};
        \node[phase] (node3c) at (3,-2) {};
        \draw[-] (node3a) to[out=-60,in=60] (node3c);
        \node (node4a) at (4,0) {};
        \node[phase] (node4b) at (4,-1) {};
        \node[phase] (node4c) at (4,-2) {};
        \draw[-] (node4b) to[out=-60,in=60] (node4c); 
        \node[phase] (node5a) at (5,0) {};
        \node[phase] (node5b) at (5,-1) {};
        \node[phase] (node5c) at (5,-2) {};
        \draw[-] (node5a) to[out=-60,in=60] (node5b);
        \draw[-] (node5a) to[out=-60,in=60] (node5c);
        \draw[-] (node5b) to[out=-60,in=60] (node5c);
        \node (end1) at (6,0) {} edge [-] (q1);
        \node (end2) at (6,-1) {} edge [-] (q2);
        \node (end3) at (6,-2) {} edge [-] (q3);

        \node (dummy) at (3.2, -2) {};
        \node (dummy2) at (5.2, -1) {};
        \node (dummy3) at (5.2, -2) {};
        \begin{pgfonlayer} {background}
        
        \node (background) [fill=yellow , fill opacity=1, fit = (node2b) (node3b)] {};
        \node (background) [fill=yellow , fill opacity=1, fit = (node0c) (node1c)] {};
        \node (background) [fill=yellow , fill opacity=1, fit = (node0a) (node4a)] {};
        \end{pgfonlayer}
    \end{tikzpicture}
    \caption{$\non{a}{\{[1,6]\}}$}
    \label{subfig:Nona}
    \end{subfigure}
    \begin{subfigure}[b]{.45\linewidth}
    \centering
    \begin{tikzpicture}[thick,scale=.9]
        \tikzstyle{phase} = [fill,shape=circle,minimum size=5pt,inner sep=0pt];
        
        \node (time1) at (0,0.5) {1};
        \node (time2) at (1,0.5) {2};
        \node (time3) at (2,0.5) {3};
        \node (time4) at (3,0.5) {4};
        \node (time5) at (4,0.5) {5};
        \node (time6) at (5,0.5) {6};

        \node at (-1,0) (q1) {$a$};
        \node at (-1,-1) (q2) {$b$};
        \node at (-1,-2) (q3) {$c$};
        \node (node0a) at (0,0) {};
        \node[phase] (node0b) at (0,-1) {};
        \node[phase] (node0c) at (0,-2) {};
        \draw[-] (node0b) to[out=-60,in=60] (node0c);
        \node[phase] (node1a) at (1,0) {};
        \node[phase] (node1b) at (1,-1) {};
        \node (node1c) at (1,-2) {};
        \draw[-] (node1a) to[out=-60,in=60] (node1b);
        \node (node2a) at (2,0) {};
        \node(node2b) at (2,-1) {};
        \node (node2c) at (2,-2) {};

        \node[phase] (node3a) at (3,0) {};
        \node (node3b) at (3,-1) {};
        \node[phase] (node3c) at (3,-2) {};
        \draw[-] (node3a) to[out=-60,in=60] (node3c);
        \node (node4a) at (4,0) {};
        \node[phase] (node4b) at (4,-1) {};
        \node[phase] (node4c) at (4,-2) {};
        \draw[-] (node4b) to[out=-60,in=60] (node4c); 
        \node[phase] (node5a) at (5,0) {};
        \node[phase] (node5b) at (5,-1) {};
        \node[phase] (node5c) at (5,-2) {};
        \draw[-] (node5a) to[out=-60,in=60] (node5b);
        \draw[-] (node5a) to[out=-60,in=60] (node5c);
        \draw[-] (node5b) to[out=-60,in=60] (node5c);
        \node (end1) at (6,0) {} edge [-] (q1);
        \node (end2) at (6,-1) {} edge [-] (q2);
        \node (end3) at (6,-2) {} edge [-] (q3);

        \node (dummy) at (3.2, -2) {};
        \node (dummy2) at (5.2, -1) {};
        \node (dummy3) at (5.2, -2) {};
        \begin{pgfonlayer} {background}
        
        \node (background) [fill=yellow , fill opacity=1, fit = (node2a) (node3a)] {};
        \node (background) [fill=yellow , fill opacity=1, fit = (node1c) (node2c)] {};
        \node (background) [fill=yellow , fill opacity=1, fit = (node0b) (node4b)] {};
        \end{pgfonlayer}
    \end{tikzpicture}
    \caption{$\non{b}{\{[1,6]\}}$}
    \label{subfig:Nonb}
    \end{subfigure}
    \begin{subfigure}[b]{.45\linewidth}
    \begin{tikzpicture}[thick,scale=.9]
        \tikzstyle{phase} = [fill,shape=circle,minimum size=5pt,inner sep=0pt];
        
        \node (time1) at (0,0.5) {1};
        \node (time2) at (1,0.5) {2};
        \node (time3) at (2,0.5) {3};
        \node (time4) at (3,0.5) {4};
        \node (time5) at (4,0.5) {5};
        \node (time6) at (5,0.5) {6};
        \node at (-1,0) (q1) {$a$};
        \node at (-1,-1) (q2) {$b$};
        \node at (-1,-2) (q3) {$c$};
        \node (node0a) at (0,0) {};
        \node[phase] (node0b) at (0,-1) {};
        \node[phase] (node0c) at (0,-2) {};
        \draw[-] (node0b) to[out=-60,in=60] (node0c);
        \node[phase] (node1a) at (1,0) {};
        \node[phase] (node1b) at (1,-1) {};
        \node (node1c) at (1,-2) {};
        \draw[-] (node1a) to[out=-60,in=60] (node1b);
        \node (node2a) at (2,0) {};
        \node(node2b) at (2,-1) {};
        \node (node2c) at (2,-2) {};

        \node[phase] (node3a) at (3,0) {};
        \node (node3b) at (3,-1) {};
        \node[phase] (node3c) at (3,-2) {};
        \draw[-] (node3a) to[out=-60,in=60] (node3c);
        \node (node4a) at (4,0) {};
        \node[phase] (node4b) at (4,-1) {};
        \node[phase] (node4c) at (4,-2) {};
        \draw[-] (node4b) to[out=-60,in=60] (node4c); 
        \node[phase] (node5a) at (5,0) {};
        \node[phase] (node5b) at (5,-1) {};
        \node[phase] (node5c) at (5,-2) {};
        \draw[-] (node5a) to[out=-60,in=60] (node5b);
        \draw[-] (node5a) to[out=-60,in=60] (node5c);
        \draw[-] (node5b) to[out=-60,in=60] (node5c);
        \node (end1) at (6,0) {} edge [-] (q1);
        \node (end2) at (6,-1) {} edge [-] (q2);
        \node (end3) at (6,-2) {} edge [-] (q3);

        \node (dummy) at (3.2, -2) {};
        \node (dummy2) at (5.2, -1) {};
        \node (dummy3) at (5.2, -2) {};
        \begin{pgfonlayer} {background}
        
        \node (background) [fill=yellow , fill opacity=1, fit = (node0a) (node1a)] {};
        \node (background) [fill=yellow , fill opacity=1, fit = (node1b) (node2b)] {};
        \node (background) [fill=yellow , fill opacity=1, fit = (node0c) (node4c)] {};
        \end{pgfonlayer}
    \end{tikzpicture}
    \caption{$\non{c}{\{[1,6]\}}$}
    \label{subfig:Nonc}
    \end{subfigure}
    \begin{subfigure}[b]{.45\linewidth}
    \begin{tikzpicture}[thick,scale=.9]
        \tikzstyle{phase} = [fill,shape=circle,minimum size=5pt,inner sep=0pt];
        
        \node (time1) at (0,0.5) {1};
        \node (time2) at (1,0.5) {2};
        \node (time3) at (2,0.5) {3};
        \node (time4) at (3,0.5) {4};
        \node (time5) at (4,0.5) {5};
        \node (time6) at (5,0.5) {6};
        \node at (-1,0) (q1) {$a$};
        \node at (-1,-1) (q2) {$b$};
        \node at (-1,-2) (q3) {$c$};
        \node (node0a) at (0,0) {};
        \node[phase] (node0b) at (0,-1) {};
        \node[phase] (node0c) at (0,-2) {};
        \draw[-] (node0b) to[out=-60,in=60] (node0c);
        \node[phase] (node1a) at (1,0) {};
        \node[phase] (node1b) at (1,-1) {};
        \node (node1c) at (1,-2) {};
        \draw[-] (node1a) to[out=-60,in=60] (node1b);
        \node (node2a) at (2,0) {};
        \node(node2b) at (2,-1) {};
        \node (node2c) at (2,-2) {};

        \node[phase] (node3a) at (3,0) {};
        \node (node3b) at (3,-1) {};
        \node[phase] (node3c) at (3,-2) {};
        \draw[-] (node3a) to[out=-60,in=60] (node3c);
        \node (node4a) at (4,0) {};
        \node[phase] (node4b) at (4,-1) {};
        \node[phase] (node4c) at (4,-2) {};
        \draw[-] (node4b) to[out=-60,in=60] (node4c); 
        \node[phase] (node5a) at (5,0) {};
        \node[phase] (node5b) at (5,-1) {};
        \node[phase] (node5c) at (5,-2) {};
        \draw[-] (node5a) to[out=-60,in=60] (node5b);
        \draw[-] (node5a) to[out=-60,in=60] (node5c);
        \draw[-] (node5b) to[out=-60,in=60] (node5c);
        \node (end1) at (6,0) {} edge [-] (q1);
        \node (end2) at (6,-1) {} edge [-] (q2);
        \node (end3) at (6,-2) {} edge [-] (q3);

        \node (dummy) at (3.2, -2) {};
        \node (dummy2) at (5.2, -1) {};
        \node (dummy3) at (5.2, -2) {};
        \begin{pgfonlayer} {background}
        
        \node (background) [fill=yellow , fill opacity=1, fit = (node3a) (node4c)] {};
        \end{pgfonlayer}
    \end{tikzpicture}
    \caption{Maximal $\Delta$-2-plex}
    \label{subfig:Kplex}
    \end{subfigure}
    \caption{Given a temporal graph with $\Delta=1$, \cref{subfig:Nona,subfig:Nonb,subfig:Nonc} show the \dnnh~of vertices $a$, $b$, and~$c$, respectively, shaded in yellow. In \cref{subfig:Nona}, we can see for example that $b$ within time interval $[3,4]$ is in the \dnnh~of $a$ because in the $\Delta$-frames $\Delta_3$ ($[3,4]$) and $\Delta_4$ ($[4,5]$) there are no time-stamped edges between $a$ and $b$. 
    \cref{subfig:Kplex} shows a maximal $\Delta$-2-plex shaded in yellow. 
    The tuple $(\{a,b,c\},[4,5])$ is a $\Delta$-2-plex because in~$\Delta_4$ and $\Delta_5$ each vertex has at most two non-neighbors (including itself). }
    \label{fig:Non-Neighborhood}
\end{figure}

We focus on finding maximal \dkps.
As already discussed for $\Delta$-cliques~\cite{viard2016deltacliquesjournal}, there is both \emph{vertex-maximality} and \emph{time-maximality}.
Given a temporal graph $G=(V,E,\omega)$, a \dkp~$R=(C,I)$ is \emph{vertex-maximal} if and only if there is no vertex~$v \in V\setminus C$ such that~$(C \cup \{v\}, I)$ is a \dkp.
Intuitively, a \dkp~is vertex-maximal if no other vertex can be added to it without decreasing its lifetime.
We say that a \dkp~$R=(C,I)$ is \emph{time-maximal} if and only if there is no $I\subset I' \subseteq [\omega]$ such that~$R'=(C,I')$ is a \dkp.
Intuitively, a \dkp~is time-maximal if we cannot increase its lifetime without removing a vertex from it.
We call a~\dkp~\emph{maximal} if it is both vertex-maximal and time-maximal.

\subsubsection{Degeneracy of Temporal Graphs}
 The \emph{degeneracy} of a static graph~$G$ is the smallest integer~$d$ such that
 every non-empty subgraph of~$G$ contains a vertex of degree at most~$d$. We use
 an analogue for the temporal setting as introduced by Himmel et al.~\cite{himmel2017adapting},
 also motivated by the fact that many real-world static graphs have small degeneracy~\cite{ELS13}. 

\begin{definition}\label{def:dslicedeg}
A temporal graph~$G=(V,E,\omega)$ has \emph{$\Delta$-slice degeneracy}~$d$ if for all~$t \in [\omega-\Delta]$ it holds that~$G_{\Delta_t}=(V, E_{\Delta_t})$, where~$E_{\Delta_t} = \{\{v, w\} \ | \ (\{v, w\}, t') \in E \land t' \in \Delta_t\}$, has degeneracy at most~$d$.
\end{definition}

\subsection{\textsc{BronKerbosch}}
The \textsc{BronKerbosch} algorithm is a classic algorithm that enumerates all
maximal cliques in a static graph \cite{bron1973algorithm}. It is a simple yet clever
backtracking algorithm, which can be made to perform very well on real-world networks~\cite{ELS13}. We give a short
description here since our adaptation inherits several ideas of this algorithm.

The \textsc{BronKerbosch} algorithm maintains three distinct sets of vertices.
The first set~$R$ contains the current clique. The other two sets $P$ and $X$ contain the vertices that can be added to~$R$ such that~$R$ is still a clique. The set $P$ contains all candidates which have not been considered in previous iterations, while the set $X$ contains vertices that have been considered before. In each recursive call, the algorithm first checks whether the current clique~$R$ is maximal, that is, whether~$P\cup X = \emptyset$. If so, then it adds $R$ to the solution, otherwise it iterates through all vertices $v \in P$, adds $v$ to~$R$, and recursively calls itself with updated sets $P'$ and~$X'$ where all vertices that are not adjacent to~$v$ are removed. Afterwards, it removes~$v$ from $P$ and adds it to $X$. The initial call is with~$P=V$ and~$R=X=\emptyset$.

\section{$\Delta$-$k$-\textsc{Bron\-Kerbosch}}
\label{sec:ourAlg}

\begin{algorithm}[t!]
\begin{algorithmic}[1] 
\Function{$\Delta$-$k$-BronKerbosch}{$P ,R=(C,\mathcal I) , X,B$} 
\LeftComment{$R=(C,\mathcal I):$ for every $I \in \mathcal I$, $(C,I)$ is a time-maximal \dkp.}
\LeftComment{$P \cup X:$ set of all vertex-interval-set pairs $(v,\mathcal I_v)$ such that for all $I_v \in \mathcal I_v$ it holds that $I_v \sqin \mathcal I$ and $(C \cup \{v\}, I_v)$ is a time-maximal \dkp~and where
 \begin{compactitem}
   	\item vertex-interval-set pairs in $P$ have not yet been considered as additions to $R$, and
   	\item vertex-interval-set pairs in $X$ have been considered in earlier steps.
 \end{compactitem}}
 \LeftComment{$B \colon V \times [\omega-\Delta] \mapsto \mathbb{N}$ with $B(w,t) = |\non{w}{t} \cap C|$, that is, function~$B$ displays for every vertex~$w$ and every $\Delta$-frame $\Delta_t$ the number of non-neighbors of $w$ in $C$ within~$\Delta_t$. }
\For{$I \in \mathcal{I}$}
\If{$\forall (w, \mathcal I_w) \in P \cup X$ and $\forall I_w \in \mathcal I_w 
\colon I_w \not = I$}\label{line:addcondition}
	\State{add~$(C,I)$ to the solution}
\EndIf
\EndFor

\For{$(v,\mathcal I_v) \in P$}
	\State{$C' \gets C \cup \{v\}$}
	\State{$\mathcal I' \gets \mathcal I_v $}
	\LeftComment{Adaption of the function $B$ and the sets~$P$ and $X$ to the new set of \dkps~$(C',\mathcal I')$. Crit contains all pairs $(w, \mathcal I_w)$ where~$w$ has exactly $k$ $\Delta$-non-neighbors in $C'$ during $I \in \mathcal I_w$.}
	\State{$B',\text{Crit} \gets \Call{UpdatePool}{B,C',(v, \mathcal I_v)} $}
	\State{$P' \gets \Call{Update}{P,C',\text{Crit}, (v,\mathcal I_v)} $}
	\State{$X' \gets \Call{Update}{X,C',\text{Crit},(v,\mathcal I_v)} $}

	\State{\Call{$\Delta$-$k$-BronKerbosch}{$P',R'=(C',\mathcal I'), X',B'$}}
	\State{$P \gets P \setminus \{(v,\mathcal I_v)\}$}
	\State{$X \gets X \cup \{(v,\mathcal I_v)\}$}
\EndFor
\EndFunction
\end{algorithmic}
\caption{Enumerating all maximal $\Delta$-$k$-plexes}
\label{alg:bronkerdelta}
\end{algorithm}
\begin{algorithm}[t!]
\begin{algorithmic}[1] 
\Function{UpdatePool}{$B,C,(v, \mathcal I_v) $}
 \LeftComment{$B \colon V \times [\omega-\Delta]  \rightarrow \mathbb{N}$ with $B(v,t) = | \non{v}{t} \cap C |$.} 
 \LeftComment{$C,v, \mathcal I_v \colon$ 
$v \in C$, and for every $I \in \mathcal I_v$ it holds that $(C,I)$ is a time-maximal \dkp.
}
\LeftComment{Update function $B$ for $v \in C$; store critical vertex-interval-set pairs.}
\State{$\text{Crit} \gets \{(c,\emptyset) \mid c \in C \cup V(P)\}$ }
\State{$B' \gets B$}
\For{$ (w,\mathcal I_w) \in  P \cup X \cup \{(c,\mathcal I_v) \mid c \in C\}$} 
\LeftComment{Function $B$ only changes if a vertex is in the non-neighborhood of $v$ within $\mathcal I_v$.}
\State{$ (w,\mathcal I_{\text{crit}}) \gets (w,\mathcal I_w) \sqcap \non{v}{\mathcal I_v}$}
\For{$ t \in I_{\text{crit}}, I_{\text{crit}} \in \mathcal I_{\text{crit}}$}
\State{$ B'(w,t) \gets B'(w,t) + 1$}
\LeftComment{If a vertex $v$ at time step $t$ already has $k$ non-neighbors in $C$, then this vertex is critical.}
\If{$B'(w,t) = k$ }
\State{$\text{Crit} \gets \text{Crit} \sqcup (w,[t,t])$}
\EndIf
\EndFor
\EndFor
\State{\Return{$B',\text{Crit}$}}
\EndFunction
\Statex{~}
\Function{Update}{$P , C, \text{Crit} , (v,\mathcal I_v)$} 
\LeftComment{Update $P$ such that for all $(w, \mathcal I_w)$ and all $I_w \in \mathcal I_w, I_w \sqsubseteq \mathcal I_v$ it holds that $(C \cup \{w\}, I_w)$ is a time-maximal \dkp.}
\State{$P_{\text{reduced}} \gets P[\mathcal I_v]$}
\State{$P' \gets \emptyset$}
\For{$(w,\mathcal I_w) \in P_{\text{reduced}}[V(P)\setminus
\{v\}]$}\label{line:add}
\LeftComment{If for $w$ there exists a non-neighbor in some~$\Delta_i \in I_w \in \mathcal I_w$ in $C$ that is critical in~$\Delta_i$, then we cannot add $w$ to $C$ in~$\Delta_i$.}
	\For{$(u,\mathcal I_u) \in  \text{Crit}[C \cup \{w\}] \sqcap  \non{w}{\mathcal I_w}$}
		\State{$\mathcal I_w \gets \mathcal I_w \smi \mathcal I_u$}
	\EndFor
	\State{$P' \gets P' \cup \{(w,\mathcal I_w)\} $}
\EndFor 
\State{\Return{$P'$}}
\EndFunction
\end{algorithmic}
\caption{Updating all Vertex-Interval-Sets}
\label{alg:update}
\end{algorithm}

The original \textsc{BronKerbosch} algorithm enumerates all maximal cliques in a
static graph. In previous work, \textsc{BronKerbosch} was adapted to enumerate maximal $\Delta$-cliques in temporal graphs~\cite{himmel2017adapting}. In the following, we describe how to further
modify and improve the algorithm to enumerate maximal~$k$-plexes in temporal graphs. We call
the new algorithm \BKB, see \cref{alg:bronkerdelta} and \cref{alg:update} for pseudocode. 
To adapt the $\Delta$-clique algorithm~\cite{himmel2017adapting} to enumerate maximal \dkps{}, \BKB{} additionally maintains a \emph{pool} for the current \dkps{}, which is a data structure that keeps track of the missing neighbors of each vertex of the current \dkps.

More formally, the input of
\BKB{} consists of two sets~$P$ and~$X$ of vertex-interval-set pairs, an
implicit set of current time-maximal~\dkps{}~$R=(C,\mathcal I)$,
and a \emph{pool}~$B$.
Herein, $C$ is the set of vertices of the~\dkps{} and~$\mathcal I$ is a set of intervals~$I$ on which the $(C,I)$ forms a time-maximal~\dkp.
A pool is an auxiliary data structure that stores the number of $\Delta$-non-neighbors of the vertices of the~\dkp{} in any $\Delta$-frame.
While in the original \textsc{BronKerbosch} algorithm the sets~$P$ and~$X$ contain the common neighborhood of all vertices in~$R$, our sets~$P$ and~$X$ contain all vertices~$v$ with interval sets $\mathcal I_v$ such that for all~$I_v \in \mathcal I_v$ it holds that $I_v \sqin \mathcal I$ and $(C \cup \{v\}, I_v)$ is a time-maximal \dkp.
These vertices cannot be contained in the~\dnnh{} of more than~$k-1$ other vertices of~$C$ in each $\Delta$-frame~$\Delta_i$, $i\in I_v$. They can neither be contained in the~\dnnh{} of a vertex $w\in C$ during its \emph{critical} intervals, 
that is, the intervals where $w$ has exactly $k$ $\Delta$-non-neighbors in~$C$ (including $w$ itself). To maintain these properties after expanding the current~\dkp{}, we update the pool~$B$ with the \updatepool{} procedure after adding a new vertex $v\in V(P)$ to~$C$ and then update the sets~$P$ and~$X$ with the \update{} procedure, see \cref{alg:update}. For each vertex in $V(P)\cup V(X) \cup C$, we save the number of~\dnns{} of vertices in $C$ for each~$\Delta$-frame in the pool~$B$. We iterate through all vertex-interval-set pairs~$(v,\mathcal I_v)\in P$, call the \updatepool{} and \update{} procedures, and then do a recursive call with the updated sets~$R'$,~$P'$, and~$X'$.

One improvement over the basic idea of the algorithm of Himmel et al.~\cite{himmel2017adapting} is that we maintain a vertex set together with a \emph{set} of time intervals where this vertex set induces a time-maximal~\dkp{} as opposed to only one time interval. In our experiments, this turned out to yield a significant speed-up for computing $\Delta$-cliques ($\Delta$-$1$-plexes) (see \Cref{sec:experiment}). 

For a given temporal graph $G=(V,E,\omega)$, the input for the initial call to
enumerate all maximal \dkps\ is~${P=\{(v,\{[\omega-\Delta]\})\mid v\in V\}}$, $X=\emptyset$,
$R=(\emptyset,\{[\omega-\Delta]\})$, and~$B(v,i)=0$ for all
$v\in V$ and $i\in [\omega-\Delta]$. In the remainder of this work, we always assume this
initial call of the algorithm.

\subsection{Correctness of $\Delta$-$k$-\textsc{Bron\-Kerbosch}}
\label{sec:correct}
In this section we prove the correctness of \BKB. 
We start by claiming that the pools are correctly maintained by the
\updatepool\ function (see~\cref{alg:update}), that is, the value of each pool on each relevant $\Delta$-frame is equal to the amount of non-neighbors in the current \dkp~$R$.

\begin{lemma}
\label{pool}
In each recursive call of \BKB$(P,R=(C,\mathcal I),X,B)$, for each vertex-interval-set-pair $(w,\mathcal I_w) \in P \cup X \cup \{(c, \mathcal I) \mid c \in C\}$ and each $t \in I_w, I_w \in \mathcal I_w$ we have $$B(w,t) = |\non{w}{t} \cap C |.$$ 
\end{lemma}
 \begin{proof}
The proof is by induction on the recursion depth, that is, the number $|C|$ of vertices of the \dkp{} in the current recursive call. 

Initially, the \BKB{} is called with $R=(\emptyset,\{[\omega-\Delta]\})$ and the pool~$B$ is
initialized for all~$w \in V$ and $t \in [\omega-\Delta]$ with $B(w,t) = 0 = |
\non{w}{t} \cap \emptyset|$.

Now let us assume that for a recursive call \BKB$(P,R=(C,\mathcal I),X,B)$ the condition holds.
Let $(v, \mathcal I_v) \in P$ be a vertex-interval-set-pair added to $R$, that is,~$R'= (C \cup \{v\},\mathcal I_v)$. 
Now for each $(w,\mathcal I_w) \in P \cup X \cup \{(c, \mathcal I_v) \mid c \in
C\}$ and~$t \in I, I\in \mathcal I_w \sqcap \mathcal I_v$ the pool-function
value~$B(w,t)$ is increased by~one if~$w$ and~$v$ are non-neighbors in the
$\Delta$-frame $\Delta_t$ in Line 4 to 7 in the \updatepool~procedure, that is,
\begin{align*}
   B'(w,t) &=
   \begin{cases}
     B(w,t) +1 & \text{if } v \in \non{w}{t}\\
     B(w,t)  & \text{else } 
   \end{cases}\\
   &= B(w,t) + |\non{w}{t} \cap \{v\}|&\\
   &= | \non{w}{t} \cap C | + |\non{w}{t}\cap \{v\}|&\\
   &= | \non{w}{t} \cap (C \cup \{v\}) |.& 
\end{align*}
The last equality holds because $v\notin C$. Next in the \update~procedure, the
sets~$P$ and~$X$ are updated according to the new \dkps{} in $R'$.
For each vertex-interval-set-pair $(w,\mathcal I'_w)$ in the new sets $P'$ and $X'$ it holds that $(w,\mathcal I_w) \in P \cup X$ with $\mathcal I'_w \sqsubseteq\mathcal I_w$, $\mathcal I'_w \sqsubseteq \mathcal I_v$, and consequently $\mathcal I'_w \sqsubseteq\mathcal I_w \sqcap \mathcal I_v$---see Lines 16 and 20 of the \update~procedure. 
Hence, the pool-function in the new recursive call \BKB$(P',R'=(C \cup \{v\},\mathcal
I_v),X',B')$ fulfills the claimed condition.
\end{proof}

Next, we show that $R$ contains time-maximal \dkp es. We further show that the
sets~$P$ and~$X$ contain all time-maximal vertex-interval-set pairs, which can be added to~$R$ such that the result still remains a time-maximal \dkp.
\begin{lemma}
\label{Lemma:PX}
In each recursive call of~\BKB$(P,R=(C,\mathcal I),X,B)$ the following holds:
\begin{enumerate}
\item for all $I \in \mathcal I$ it holds that $(C,I)$ is a time-maximal \dkp, \label{lemma:prop1}
\item for all~$(v,\mathcal I_v) \in P \cup X$ it holds that for all $I_v \in \mathcal I_v$, $(C \cup \{v\}, I_v)$ is a time-maximal \dkp, and\label{lemma:prop2}
\item all vertex-interval-set pairs $(v,\mathcal I_v)$ which satisfy 
the second property are contained in either $P$ or $X$.
\end{enumerate}
\end{lemma}
\begin{proof}
All properties can be proven by induction on the recursion depth, that is, the
number~$|C|$ of vertices of the \dkp{} in the current recursive call.

Initially, \BKB{} is called with $R=(C,\mathcal I)=(\emptyset,\{[\omega-\Delta]\})$, $P=\{(v,\{[\omega-\Delta]\}) \mid v \in V\}$, and~$X=\emptyset$. The interval set $\mathcal I$ contains the whole lifetime of the temporal graph and~$(\emptyset,[\omega])$ is a trivial time-maximal \dkp.
For all $(v,\mathcal I_v) \in P \cup X$ it holds that~$\mathcal I_v = \{[\omega-\Delta]\}$ and, thus, that~$(\emptyset \cup \{v\},[\omega-\Delta])$ is a trivial time-maximal \dkp. Obviously, $P \cup X$ contains all vertex-interval-set pairs that form time-maximal \dkps~with~$R$. 

Now let us assume that for a recursive call \BKB$(P,R=(C,\mathcal I),X,B)$ all properties hold.
Let $(v, \mathcal I_v) \in P$ be a vertex-interval-set-pair added to $R$, that is,~$R'= (C',\mathcal I_v)$ with $C' = C \cup \{v\}$. By induction hypothesis, for all $I \in \mathcal I_v$ it holds that $(C',I)$ is a time-maximal \dkp. It remains to show that $P$ and $X$ are suitably adapted for the new recursive call on~$R'$. 

We show that $P'$ and $X'$ satisfy the above properties after a call of the
\update\ procedure.
The \update{} procedure gets as input the set $P$ (or $X$) of vertex-interval-set pairs, the
vertex set $C'$ of the current \dkp{} set $R'$, a set~$\text{Crit}$ of critical vertex-interval-set pairs, and the newly added vertex-interval-set pair $(v,\mathcal I_v)$. 
The set $\text{Crit}$ contains all vertices $w$ and $\Delta$-frames $\Delta_i$ such that the vertex $w$ has $k$ non-neighbors in~$C'$ within~$\Delta_i$. More formally, it contains all vertex-interval-set pairs $(w, \mathcal I_{\text{Crit}})$ such that for all $i \in I_{\text{Crit}},I_{\text{Crit}} \in \mathcal I_{\text{Crit}}$, it holds that $ [i,i] \sqin \mathcal I_v$ and $|\non{w}{i} \cap  C' |=k$.
We now show that \update{} works as intended for $P$. The case for $X$ is
analogous. We create $P'$ as follows:

For each $(w,\mathcal I_w)\in P$ with $w \neq v$, the \update\ procedure
reduces the interval set~$\mathcal I_w$ to~$\mathcal I_v$, that is, $\mathcal I_w'= \mathcal I_w \sqcap \mathcal I_v$. 
Now, all $\Delta$-frames $\Delta_i$ with $i \in I, I \in \mathcal I_w'$, are removed in Lines $19$ to $20$ of the \update~procedure for which
\begin{enumerate}
\item vertex $w$ has already $k$ non-neighbors in $C'$, that is, $|\non{w}{i} \cap C' |=k$ (recall that~$w \in \non{w}{i}$), or 
\item vertex $w$ has a non-neighbor in $C$ that has already~$k$ non-neighbors in $C'$, that is, there exists a vertex~$c \in C'$ with $c \in \non{w}{i}$ and $|\non{c}{i}\cap~C'|=k$.  
\end{enumerate}
In both cases $w$ cannot be added to $C'$. 
In the end, maximal intervals of the remaining $\Delta$-frames are formed. 
By induction hypothesis all intervals in $\mathcal I_w$ were time-maximal with respect to $C$ and all intervals that form a time-maximal \dkp{} with $C$ were contained in $\mathcal I_w$. Adding a vertex $v$ to a \dkp~only lessens the set of possible $\Delta$-frames of an additional vertex $w$.
Hence, all vertex-interval-set pairs in~$P'$ (and~$X'$) are time-maximal and complete with respect to $C'$. 

Thus, the recursive call \BKB$(P',R'=(C',\mathcal I_v),X',B')$ fulfills the conditions stated in \cref{Lemma:PX}.
\end{proof}

We are now ready to prove the correctness of \BKB. 
\begin{theorem}
\label{theorem:correct}
For any given temporal graph $G=(V,E,\omega)$, \BKB{} computes all maximal~\dkps{} of~$G$.
\end{theorem}
\begin{proof}
Let $R^*=(C^*,I^*)$ be a maximal \dkp. We show that there will be a recursive
call adding $R^*$ to the solution. Since we are building \dkps~bottom up,
there will be a recursive call of \BKB{} on~$(P,R,X,B)$ with $R=(C,\mathcal I)$,
$C \subseteq C^*$, $I^* \sqin \mathcal I$ and~${|C|=|C^*|-\ell}$ for all $\ell=
0,1,\ldots,|C^*|$. Additionally, all vertices $v \in C^* \setminus C$ with $I^*
\sqin \mathcal I_v $, called \emph{candidates}, will be contained in $P$. We
show this by induction on~$|C|$. 

Clearly, in the initial call, $C=\emptyset \subseteq C^*$ and $I^* \sqin
\{[\omega-\Delta]\}$. Since $P=\{(v,\{[\omega-\Delta]\})\mid v\in V\}$, every vertex~$v\in~C^*$ is
contained in~$P$. Now assume that there is a recursive call with $(P,R,X,B)$, where~$R=(C,\mathcal I)$,
$C \subseteq C^*$, $I^* \sqin \mathcal I$, and all candidates are contained in $P$.
Consider the first candidate~$(v,\mathcal I_v)$ in the for-loop of that
recursive call.
After adding $v$ to $C$, since $R^*$ is a \dkp, according to \cref{Lemma:PX} all
other candidates are still contained in $P'$ after a call of \update.
Since~$(v,\mathcal I_v)$ was a candidate, it holds for the new \dkp{} set $R'=(C'=C\cup
\{v\},\mathcal I_v)$ that $C' \subseteq C^*$ and~$I^*\sqin \mathcal I_v$.
Hence, by induction, there is a recursive call with $R=(C^*,\mathcal I^*)$
with $I^*\in \mathcal I^*$ and, since $R^*$ is maximal, there is no
vertex-interval-set pair $ (v,\mathcal I) \in P\cup X$ with $I^* \sqin \mathcal
I$.
Thus,~$(C^*,I^*)$ is enumerated. 

Now assume that some pair $(C,I)$ is added to the solution.
We show that~$(C,I)$ is a maximal \dkp{}.
By \cref{Lemma:PX} we
know that $(C,I)$ is a time-maximal \dkp\ and we know that all
vertex-interval-set pairs~$(v,\mathcal I_v)$, where~$R'=(C\cup
\{v\},\mathcal I_v)$ is a set of \dkps{}, are contained in~$P\cup X$. Since we check whether $\forall (w,
\mathcal I_w) \in P \cup X$ and~$\forall I_w \in \mathcal I_w 
\colon I_w \not = I$ in Line~\ref{line:addcondition} of \cref{alg:bronkerdelta}, it follows that there
is no vertex in~$P$ or~$X$ which can be added without decreasing the
interval $I$, hence, $(C,I)$ is also vertex-maximal. Thus, $(C,I)$ is a maximal
\dkp.
\end{proof}

\subsection{Running Time of $\Delta$-$k$-\textsc{Bron\-Kerbosch}}
\label{sec:time}
We have shown that \BKB~enumerates all maximal \dkps\ in a temporal
graph.
In this section, we analyze its running time in four steps.
First, we determine the running time of precomputing the \dnnh s. Then, we
analyze the running time of \updatepool\ and \update.
In a third step, we prove an upper bound on the number of time-maximal \dkps{} that depends on the \dsd{} of the temporal graph
and show that there is at most one recursive call for each of them. 
Finally, we combine our findings to obtain an upper bound on the running time of \BKB.

The running time of computing the \dnnh{} has a big influence on the overall running time since it is accessed
multiple times in each recursive call. However, it can be
precomputed once before the initial call of \BKB. In the following lemma, we
show an upper bound on the running time of this computation assuming that
the edges are sorted by their time stamps. 

\begin{lemma}
\label{lemma-building-nnh}
If the edges are sorted by their time stamps, then the~\dnnh\ for all vertices over the
whole lifetime, that is, $\non{v}{\{[\omega-\Delta]\}}$ for all $v \in V$, can be computed in $\mathcal{O}(|V|^2 + |E|)$ time.
\end{lemma}
\begin{proof}
First, for each pair of vertices $v,w$ we initially set their \dnnh\ to the whole lifetime of the temporal graph. 
The initialization can be done in $\mathcal O(|V|^2)$ time.   
Then, for each time-stamped edge $(\{v,w\},t)$ the $\Delta$-neighborhood interval $[t-\Delta,t]$ is cut out of the \dnnh\ of $v,w$. Due to the sorting of the edges by time stamps, this can be done in~$\mathcal O(|E|)$ time. We end up with a sorted list of \dnnh\ intervals for each vertex pair with at most~$\mathcal O(\min\{|E|,\omega\})$ many non-overlapping time intervals. 
In the last step, the \dnnh\ $\non{v}{\{[\omega-\Delta]\}}$ for each vertex $v$ is computed in $O(|V|^2 + |E|)$ time using the sorted lists of \dnnh\ intervals of all vertex pairs containing $v$.   
\end{proof}

Next, we determine the running time of \updatepool{} and \update. 
Before starting the proof, we first briefly discuss the structure of the pool function.
In the pool function~$B$, we store for each vertex $v$ and each \dfr\ $\Delta_t$ the number of \dnn s in the current \dkp. This information can be maintained for each vertex by storing time intervals of the same function values. For a vertex $v$, each change in the number of \dnn s in the current \dkp\ is induced by a time-stamped edge between~$v$ and a vertex in the current \dkp. Hence, the number of time intervals for each vertex~$v$ in~$B$ is bounded by $\mathcal O(\min\{|\{(\{v,w\},t)\in E\}|,\omega\})$. 

\begin{lemma}\label{theorem:addupdate}
The procedures \updatepool{} and \update{} take~$\mathcal{O}(\min\{|E|,\omega\} \cdot|V|^2)$ time.
\end{lemma}
\begin{proof} 
First, let us briefly discuss the structure of the pool function.
In the pool function~$B$, we store for each vertex $v$ and each \dfr\ $\Delta_t$ the number of \dnn s in the current \dkp. This information can be maintained for each vertex by storing time intervals of the same function values. For a vertex $v$, each change in the number of \dnn s in the current \dkp\ is induced by a time-stamped edge between~$v$ and a vertex in the current \dkp. Hence, the number of time intervals for each vertex~$v$ in~$B$ is bounded by $\mathcal O(\min\{|\{(\{v,w\},t)\in E\}|,\omega\})$. 

Also note that in the beginning of the algorithm, all interval sets are sorted by the start time of the intervals and that all intervals in such a set are pairwise non-overlapping. These two properties of interval sets are preserved during all of our interval set operations. Furthermore, each interval in an interval set is induced by a different time-stamped edge. Hence, the size of each set of intervals can be bounded by $\mathcal O(\min\{|E|,\omega\})$. Thus, each interval set cut operation can be done $\mathcal O(\min\{|E|,\omega\})$ time as shown in~\citet[Lemma 4]{himmel2017adapting}.

Now, let us consider the running time of the \updatepool\ procedure.
Initializing~$\text{Crit}$ and~$B'$ takes~$\mathcal O(|V|)$ and $\mathcal O( \min\{|E|,|V|\omega\})$ time, respectively.  
Next, for each~$(w,\mathcal I_w) \in  P \cup X \cup \{(c,\mathcal I_v) \mid c \in C\}$ we compute the cut with the \dnnh~$\non{v}{ \mathcal I_v}$ in $\mathcal O( \min\{|E|,\omega\})$ time. Updating the pool function also takes $\mathcal O( \min\{|E|,\omega\})$ time. Filtering the critical time intervals can be done during the update with no extra time consumption. There are $|V|$ vertex-interval-set pairs in $P \cup X \cup \{(c,\mathcal I_v) \mid c \in C\}$. Hence, the overall time is in $\mathcal O(|V| \min\{|E|,\omega\}))$.  

Next, consider the \update\ procedure.  
Reducing $P$ to the interval set $\mathcal I_V$ takes $\mathcal O(|V| \min\{|E|,\omega\})$ time.   
There are at most $|V|$ elements in the set $P_{\text{reduced}}$. For each $(w,\mathcal I_w) \in P_{\text{reduced}}[V(P) \setminus \{v\}]$, we compute the cut $\text{Crit}[C \cup \{w\}] \sqcap  \non{w}{\mathcal I_w}$ in~$\mathcal O(|V| \min\{|E|,\omega\})$ time. In this cut, there are at most~$\mathcal O(|V| \min\{|E|,\omega\}$ vertex-interval-set pairs. For each of these elements $(u,\mathcal I_u)$ we can compute~$I_w \setminus I_u$ again in~$\mathcal O(\min\{|E|,\omega\}$ time.

Altogether, the running time of this whole procedure is in $\mathcal O(|V|^2 \min\{|E|,\omega\})$.
\end{proof}

Moving forward, we now upper-bound the number of recursive calls of \BKB.
\begin{lemma}
\label{lem:numcall}
For each time-maximal \dkp~$(C,I)$ of a temporal graph $G=(V,E,\omega)$, there is at
most one recursive call of \BKB~with $R=(C,\mathcal I)$ with $I\in \mathcal I$
as input.
\end{lemma}
\begin{proof}
Assume that there are two recursive calls $A$ and~$B$ with the same $R=(C,\mathcal I)$ as part of the input. 
Let $R'=(C',\mathcal I')$, with $C' \subset C$ and for all $I\in \mathcal I$, $I \sqin \mathcal I'$, be in the input of the least common ancestor of $A$ and $B$ in the tree of recursive calls. Let~$P'$ be the candidate set of that recursion call. There must be two vertex-interval-set pairs~$(u,\mathcal I_u),(w,\mathcal I_w)\in P'$ that lead to the recursive calls $A$ and $B$, respectively.

Clearly, for all $I\in \mathcal I$ it holds that $I \sqin \mathcal I_u$, ${I \sqin \mathcal I_w}$, and~$\{u,w\} \subseteq C$. Since for a fixed vertex all intervals in a candidate set are distinct, as shown in \cref{Lemma:PX}, it follows that $u \neq w$. Without loss of generality, assume that $(u,\mathcal I_u)$ is considered first in the for-loop over all candidates in~$P$. When $(w,\mathcal I_w)$ is considered, then the vertex-interval-set pair $(u,\mathcal I_u)$ is in $X'$ and not in $P'$. Hence, all following recursion calls do not consider $(u,\mathcal I_u)$. Recall that we assumed that the recursive call~$B$ outputs $R=(C,\mathcal I)$. Since for all~$I\in \mathcal I$ it holds that $I \sqin \mathcal I_u$ and~$u \in C$, it follows that a vertex-interval-set pair $(u,\mathcal I_u')$ with the property that 
\begin{itemize}
\item[a)] for all $I\in \mathcal I$, it holds that $I \sqin \mathcal I'_u$ and 
\item[b)] for all $I'\in \mathcal I'_u$, it holds that $I' \sqin \mathcal I_u$ 
\end{itemize}
needs to be considered in a future call. 
This contradicts the fact that we do not consider vertex-interval-set pairs that are contained in $X$. 
Thus, there cannot be two recursive calls of~\BKB~with the same $R$.
\end{proof}
Finally, we upper-bound the number of time-maximal \dkps~in a temporal graph $G=(V,E,\omega)$ using the $\Delta$-slice degeneracy value $d$ (\Cref{def:dslicedeg}) of $G$.
This upper bound improves the
theoretical upper bound on the number of time-maximal $\Delta$-cliques shown by Himmel et
al.~\cite{himmel2017adapting} from~$|V| \cdot 3^{d/3} \cdot 2^{d+1} \cdot \omega$ to~$|V| \cdot 2^{d+1} \cdot \min\{|E|, \omega\}$.
\begin{proposition}
\label{lem:numkplex}
Let~$G=(V,E,\omega)$ be a temporal graph with~$\Delta$-slice degeneracy~$d$. Then,
the number of time-maximal $\Delta$-$k$-plexes in~$G$ is at most~$|V| \cdot \binom{|V|}{k} \cdot  2^{d + k } \cdot \min\{|E|, \omega\}$.
\end{proposition}
\begin{proof}
The statement can be shown by a simple counting argument. For each $\Delta$-frame, we count how many time-maximal~$\Delta$-$k$-plexes have a lifetime that contains this $\Delta$-frame.
For a given $\Delta$-frame~$\Delta_i$, there exists a \emph{degeneracy ordering} of~$G_{\Delta_i}=(V, E_{\Delta_i})$, where~$E_{\Delta_i}=\{e\mid (e,t)\in E\land t\in \Delta_i\}$. The degeneracy ordering of a graph is a linear ordering of its vertices with the property that for each vertex~$v$ at most~$d$ of its neighbors occur at a later position. 

Now for each $\Delta$-frame~$\Delta_i$ and each vertex $v$ (that is, $\min\{|E|, \omega\} \cdot |V|$ possibilities) in the degeneracy ordering of graph $G_{\Delta_i}=(V, E_{\Delta_i})$, we count the number of $k$-plexes of~$G_{\Delta_i}$ which only contain~$v$ and vertices that appear at a later position in the ordering. By definition, $v$ has at most~$d$ neighbors that appear later in the ordering and~$v$ can be connected to $k-1$ other vertices. For the latter, we consider all vertices, yielding the factor $\binom{|V|}{k}$ in the upper bound. Each subset of these~$k+d$ vertices in $G_{\Delta_i}$ can, together with~$v$, potentially be the vertex set of several~$\Delta$-$k$-plexes. This yields our upper bound since each of these vertex sets can potentially form at most one time-maximal $\Delta$-$k$-plex with a lifetime that contains~$\Delta$-frame~$\Delta_i$. The number of such subsets is at most~$2^{d + k }$. 
\end{proof}

We now combine the previous results to upper-bound the running time of \BKB.
For comparison,~the $\Delta$-clique algorithm of \citet{himmel2017adapting} takes~$O(2^d \cdot \min\{|E|^2, \omega^2\}  \cdot |V|)$~time.
\begin{theorem}
\label{thm:runtime}
\BKB~runs in $\mathcal{O}(\binom{|V|}{k} \cdot  2^{d + k
}\cdot\min\{|E|^2,\omega^2\} \cdot|V|^3)$ time, where $d$ is the $\Delta$-slice
degeneracy of the input graph.
\end{theorem}
\begin{proof}
First, recall that by \cref{lemma-building-nnh} the \dnnh~can be
precomputed once at the start of the algorithm in $\mathcal{O}(|V|^2 + |E|)$ time assuming
that the edges are sorted. 
If they are unsorted, then we can sort them in $\mathcal{O}(|E|\cdot\log |E|)$ time. 

We first give an upper bound on the number of recursive calls in an execution of \BKB. By \Cref{Lemma:PX}(\ref{lemma:prop1}) we know that in each recursive call of~\BKB{}, we have that~$(C,I)$ is a time-maximal \dkp\ for all $I \in \mathcal I$. \Cref{lem:numcall} tells us that the time-maximal \dkps\ in all recursive calls are distinct. Finally, \Cref{lem:numkplex} gives us an upper bound on the number of distinct time-maximal \dkps. We can conclude that for an
execution of \BKB~the number of recursive calls is bounded by~$\mathcal{O}(\binom{|V|}{k} \cdot  2^{d + k
}\cdot\min\{|E|,\omega\} \cdot |V|)$. 

Now, we analyze the running time
of each recursive call. For this, notice that there is exactly one call to \updatepool{} and two calls to
\update{} in each recursive call of~\BKB{}. Hence, the
running time of the for-loop is dominated by the complexity of
\updatepool~and~\update~which run in $\mathcal{O}(\min\{|E|,\omega\} \cdot|V|^2)$
time by~\Cref{theorem:addupdate}. Concluding the proof, there are
$\mathcal{O}(\binom{|V|}{k} \cdot  2^{d + k
}\cdot\min\{|E|,\omega\} \cdot |V|)$ recursive calls and each of
these recursive calls runs in~$\mathcal{O}(\min\{|E|,\omega\} \cdot|V|^2)$ time. This yields a total running time of~$\mathcal{O}(\binom{|V|}{k} \cdot  2^{d + k
}\cdot\min\{|E|^2,\omega^2\} \cdot|V|^3)$ for \BKB.
\end{proof}

\subsection{Heuristic Improvements}
We propose two heuristics to improve the running time of \BKB{}. In~\cref{sec:experiment}, we demonstrate the effectiveness of both heuristics on real-world data sets. 
\subsubsection{Pivoting}
A particular feature to improve  the running time of the classic 
\textsc{BronKerbosch} algorithm is the use of \emph{pivoting} \cite{ELS13}, a procedure to reduce the number of its recursive calls. Himmel et al.~\cite{himmel2017adapting} showed how to 
transfer this to the temporal setting. We use pivoting as a ``black box'' here and refer to Himmel et al.~\cite{himmel2017adapting} for more details on pivoting. They tested several alternatives and we chose to implement the method they found most effective (based on empirical results).

Let $R=(C, \mathcal I)$ be the currently considered \dkp~and let~$P$ be the set of candidates.
The rough idea of our pivoting method is as follows:
\begin{enumerate}
\item Select a pivot element  that
\begin{enumerate}
	\item is connected to each~$c\in C$ over the whole lifetime of the pivot element, and
	\item out of all of these has the largest set~$W \subseteq P$ of candidates to which it is connected over the candidates lifetime.
\end{enumerate}
\item Ignore all candidates in~$W$ in the current recursive call (but not in future calls).
\end{enumerate} 
(The same pivoting method was called ``1G'' in previous work by Himmel et al.~\cite{himmel2017adapting} and it was shown to be most effective for $\Delta$-cliques in their experiments.) 
Ignoring all candidates in~$W$ does not influence the correctness of the algorithm as each \dkp{} which is obtained by only adding candidates in~$W$ to~$R$ cannot be maximal as one can always add the pivot element to it without needing to reduce the lifetime.

\subsubsection{Connectedness Criterion}

A big obstacle when efficiently enumerating all maximal \dkps~is the solution size. Each
combination of $k$ vertices induces a \dkp~over the whole lifetime. Hence, there
are $\binom{|V|}{k}$ trivial solutions. These solutions contain no information
but for large~$k$ increase the computational cost immensely. We therefore want to
exclude these trivial solutions and terminate in the recursion before
computing those. A typical condition on $k$-plexes in enumeration algorithms for static graphs is that the $k$-plexes should be connected~\cite{conte2017fast,berlowitz2015efficient} or even have small diameter~\cite{conte2018d2k}. These constraints can indirectly be satisfied by additionally requiring the $k$-plexes to have a certain minimum size.

To this end, we modify \BKB{} to only enumerate~\dkp
es with a minimum number of~$2k+1$ vertices. These~\dkp es are clearly connected
in each $\Delta$-frame since every subgraph of order $x$, where every vertex has degree at least~$x/2$, is connected. This allows us to use a simple heuristic. During
the for-loop in \cref{alg:bronkerdelta} over all candidates in~$P$, we select only those
vertex-interval-set pairs in the candidate set~$P$ which have a connection to
the current~\dkp~$R$ in at least one $\Delta$-frame during the lifetime of $R$.
If there is no such vertex-interval-set pair in~$P$, then we can terminate this
recursive call early since it will not yield a connected \dkp~of order~$2k+1$.
To show that using this heuristic does not let us miss any connected \dkp{} is straightforward to verify, so we omit a formal proof.

\section{Experimental Analysis}
\label{sec:experiment}
In this section, we analyze the running time of \BKB~on several real-world temporal graphs and investigate the effect of different values for $k$ and~$\Delta$ on the running time.
We study the impact of the connectedness criterion and/or pivoting on the practical performance of our new algorithm and compare it to a state-of-the-art algorithm by \citet{ViardML18} and a previous version~\cite{himmel2017adapting}, both for enumerating temporal cliques.

\subsection{Setup and Statistics}
We implemented\footnote{The code of our implementation is freely available under GNU general public license version 3 at http://fpt.akt.tu-berlin.de/temporalkplex/}
\BKB{} in Python~2.7.12 and carried out experiments on an Intel Xeon
E5-1620 computer with four cores clocked at 3.6\,GHz and with 64\,GB RAM. We
did not utilize the parallel-processing capabilities.
The operating system was Debian GNU/Linux~6.0. 
We compare \BKB{} with the algorithms by \citet{himmel2017adapting} and \citet{ViardML18} 
which were implemented in Python~2.7.11. 
\begin{table}[t!]
\footnotesize
  \centering
  \caption{Statistics for the data sets used in our experiments. The lifetime $\omega$ of a graph is the difference between the largest and smallest time stamp on an edge in the graph. The resolution $r$ indicates how often edges were measured. As an example, a resolution of~$86400$ seconds (one day) indicates that the graph contains edges in only one out of~$86400$ consecutive time steps.}
  \begin{filecontents*}{GraphData.csv}
Data,Vertices,Edges,Lifetime (s),Resolution (s),Classical Degeneracy,Delta0,Degeneracy0,Delta1,Degeneracy1,Delta2,Degeneracy2,Delta3,Degeneracy3,Delta4,Degeneracy4,Delta5,Degeneracy5,Delta6,Degeneracy6,Delta7,Degeneracy7,,,,,,,,
flights,299,555515,2590200,1800,27,0,22,0,22,0,22,0,22,0,22,0,22,2,23,13,26,67,27,336,27,,,,
london,270,526371,1439,120,2,0,2,0,2,0,2,0,2,1,2,8,2,42,2,213,2,1067,2,,,,,,
paris,302,581292,172680,120,2,0,2,0,2,0,2,0,2,0,2,1,2,7,2,38,2,193,2,967,2,,,,
ny,417,658418,172680,120,2,0,2,0,2,0,2,0,2,0,2,1,2,6,2,34,2,170,2,853,2,,,,
highschool-2011,126,28560,272330,20,21,0,4,9,4,47,4,238,7,1191,11,5959,11,29798,15,148990,19,,,,,,,,
highschool-2012,180,45047,729500,20,18,0,4,16,4,80,4,404,5,2024,6,10121,6,50606,7,253034,12,,,,,,,,
highschool-2013,327,188508,363560,20,24,0,4,1,4,9,4,48,5,241,6,1205,9,6026,10,30134,14,150673,17,,,,,,
primaryschool,242,125773,116900,20,47,0,4,0,4,4,4,23,4,116,5,580,10,2904,18,14522,31,72613,34,,,,,,
hospital-ward,75,32424,347500,20,22,0,4,10,4,53,4,267,6,1339,7,6698,11,33491,14,167458,18,,,,,,,,
infectious,10972,415912,6946340,20,18,0,4,16,4,83,5,417,9,2087,18,10438,18,52192,18,260960,18,1304802,18,6524010,18,,,,
hypertext,113,20818,212340,20,28,0,6,10,6,50,6,254,7,1274,7,6374,8,31874,13,159372,22,,,,,,,,
dnc,1866,37421,84869755,1,17,0,1,2267,7,11339,7,56699,9,283496,13,1417482,16,7087410,17,35437051,17,,,,,,,,
karlsruhe,1870,461661,123837267,1,33,0,2,268,4,1341,6,6706,6,33530,6,167651,9,838258,12,4191294,17,20956473,23,104782367,32,,,,
facebook-like,1899,59835,16736181,1,20,0,1,279,2,1398,3,6992,3,34963,4,174815,6,874079,11,4370399,19,,,,,,,,
as-733,7579,11385666,67737600,86400,24,0,13,5,13,29,13,148,13,743,13,3718,13,18591,13,92958,14,464794,14,2323974,15,11619873,18,58099368,24
\end{filecontents*}
\pgfplotstabletypeset[
   col sep=comma,
   columns={Data,Vertices, Edges,Resolution (s), Lifetime (s)},
   columns/Data/.style={column type=l,string type},
   columns/Vertices/.style={column type=r, column name={\# Vertices $n$}},
   columns/Edges/.style={column type=r,int detect, column name={\# Edges $m$}},
   columns/Lifetime (s)/.style={column type=r,int detect, column name={Lifetime $\omega$ (in s)}},
   columns/Resolution (s)/.style={column type=r,int detect, column name={Resolution $r$ (in s)}},
   every head row/.style={before row=\toprule , after row=\midrule},
   every last row/.style={after row=\bottomrule},
   ]{GraphData.csv}
  \label{tab:stats1}
\end{table}
For the sake of comparability we tested our implementation on the same freely available data sets as \citet{himmel2017adapting} and \citet{ViardML18} as well as four additional traffic network data sets~\cite{williams2016spatio,dryad_3p27r}:
\begin{compactitem}
\item a US domestic flights network (``flights''~\cite{williams2016spatio,dryad_3p27r}),
\item networks of the London Underground, Paris Metro, and New York Subway (``london'', ``paris'', ``ny''~\cite{williams2016spatio,dryad_3p27r}),
\item physical-proximity networks\footnote{Available at http://www.sociopatterns.org/datasets/ .} between 
\begin{compactitem}
\item high school students (``highschool-2011'', ``highschool-2012'', ``highschool-2013''~\cite{gemmetto2014mitigation,stehle2011high,fournet2014contact}), 
\item children and teachers in a primary school (``primaryschooll''~\cite{stehle2011high}),
\item patients and health-care workers (``hospital-ward''~\cite{vanhems2013estimating}), 
\item attendees of the Infectious SocioPatterns event (``infectious''~\cite{isella2011s}),
\item conference attendees of ACM Hypertext 2009 (``hypertext''~\cite{isella2011s}), 
\end{compactitem}
\item an email communication network of the 2016 Democratic National Committee email leak (``dnc''~\cite{KONECT17}),
\item an email communication network (``karlsruhe''~\cite{EMT2011}), 
\item a social-network communication network (``facebook-like''~\cite{opsahl2009clustering}), and  
\item an internet router communication network (``as-733''~\cite{leskovec2005graphs}).
\end{compactitem}
We summarize some important statistics about the different data sets in \cref{tab:stats1,tab:stats2}.
\begin{table}[t]
\footnotesize
  \centering
  \caption{Static degeneracy and \dsd{} ($\Delta$-values are scaled by~$\omega /5m$).}
  \pgfplotstabletypeset[
   col sep=comma,
   columns={Data,Classical Degeneracy,Degeneracy0,Degeneracy3,Degeneracy5,Degeneracy7},
   columns/Data/.style={column type=l,string type},
   columns/Classical Degeneracy/.style={column name={Static}},
   columns/Degeneracy0/.style={column name={$\Delta=0$}},
   columns/Degeneracy3/.style={column name={$\sim 5^3$}},
   columns/Degeneracy5/.style={column name={$\sim 5^5$}},
   columns/Degeneracy7/.style={column name={$\sim 5^7$}},
   every head row/.style={before row=\toprule, after row=\midrule},
   every last row/.style={after row=\bottomrule},
   ]{GraphData.csv}
  \label{tab:stats2}
\end{table}
We used the same~$\Delta$-values\footnote{In order to limit the influence of time scales in the data and to make running times comparable between instances, the $\Delta$-values are chosen by fixing some exponentially increasing reference values and then scaling the values by~$\omega /5m$~\cite[Section~5.1]{himmel2017adapting}.} and the same time limit of one hour as \citet{himmel2017adapting}.
We present solutions found by our implementation for~$k \leq 3$ as for higher values of~$k$ the time limit of one hour was reached in all instances.
For the case of connected~$\Delta$-$k$-plexes we were able to solve instances up to~${k = 5}$ within the time limit of one hour.
\subsection{Pivoting and Connectedness Criterion}
In this part we analyze the effectiveness of pivoting and the use of our connectedness criterion.
We start with pivoting.
\Cref{fig:pivot} illustrates the use of pivoting on all considered data sets.
As one can see, the use of pivoting only improves the running time for the case of large~$\Delta$-values.
This intuitively makes sense as with large~$\Delta$-values one has more elements to choose from and hence carefully picking the best one improves the performance whereas if there are only few alternatives, then the overhead for computing a good pivot element outweighs its gain.
Additionally, there are two temporal graphs (``as-733'' and ``flights'', they correspond to the outliers in \Cref{fig:pivot}) for which pivoting seems to be especially effective.
\begin{figure}[t!]
	\centering
	\def\maxValueRun{3600}
	\def\minValueRun{1}
	\begin{tikzpicture}[scale=.9]
		\begin{loglogaxis}[
					width=0.9\textwidth,
					height=0.5\textwidth,
					ylabel={with \textbf{Pivoting}},
					xlabel={running time of \BKB{}},
					legend cell align=left,
					legend pos=north west
					,xmax=\maxValueRun
					,xmin=\minValueRun
					,ymax=\maxValueRun
					,ymin=\minValueRun,
					colorbar,
					colorbar style={
                		scale,
                		ylabel=$\Delta$-value,
                		ylabel style={
							yshift=-7em
						},
						ytick={1,2,...,11},
						yticklabel={$5^{\pgfmathprintnumber{\tick}}$},
					}
			]
 			\addplot[scatter src=explicit,scatter,only marks,mark=o, thick] table[col sep=comma,y={runtime_classic_pivo2}, x={runtime_classic2}, meta={delta_number}] {resultsAll.csv};
			\addplot[color=black,domain=\minValueRun:\maxValueRun,samples=4] {x};
			\addplot[dashed,color=black!75,domain=\minValueRun:\maxValueRun,samples=4] {5*x};
			\addplot[dashed,color=black!75,domain=\minValueRun:\maxValueRun,samples=4] {0.2*x};
			\addplot[dotted,color=black,domain=\minValueRun:\maxValueRun,samples=4] {25*x};
			\addplot[dotted,color=black,domain=\minValueRun:\maxValueRun,samples=4] {0.04*x};
		\end{loglogaxis}
	\end{tikzpicture}
	\caption{Comparison of the running times (in seconds) of \BKB{} with and without pivoting for different~$k$- and~$\Delta$-values (scaled by~$\omega/5m$) on all considered data sets.
	The~$x$-coordinate of a data point represents the running time without pivoting, its~$y$-coordinate represents the running time with pivoting, and its color encodes the~$\Delta$-value that was used. Data points on the upper or right border of the plot signal that the time limit of one hour was exceeded by \BKB{} with or without pivoting, respectively.
	We omitted the~$k$-value in this diagram as we could not measure any effect of the~$k$-value on the running time difference between the two runs.
	\Cref{fig:pivotk} (Appendix) illustrates the correlation between~$k$ and the running-time difference.}
	\label{fig:pivot}
~\\
	\centering
	\def\maxValueRun{5000000}
	\def\minValueRun{1}
	\begin{tikzpicture}[scale=.9]
		\begin{loglogaxis}[
 					width=0.9\textwidth,
 					height=0.5\textwidth,
					xlabel={\# \textbf{all} $\Delta$-$k$-plexes},
					ylabel={\# \textbf{connected} $\Delta$-$k$-plexes},
					legend cell align=left,
					legend pos=north west,
					xmax=\maxValueRun,
					xmin=5000,
					ymax=\maxValueRun,
					ymin=\minValueRun,
					ymode=log,
					colorbar,
					colorbar style={
                		scale,
                		ylabel={$k$},
                		ylabel style={
							yshift=-7em
						},
						ytick={1,2,...,6},
						yticklabel={$\pgfmathprintnumber{\tick}$},
				}
			]
 			\addplot[scatter src=explicit,scatter,only marks,mark=o, thick] table[col sep=comma,y={num_kplexes_conn}, x={num_kplexes_classic}, meta={k}] {resultsAll.csv};
			\addplot[color=black,domain=\minValueRun:\maxValueRun,samples=4] {x};
			\addplot[dashed,color=black!75,domain=\minValueRun:\maxValueRun,samples=4] {5*x};
			\addplot[dashed,color=black!75,domain=\minValueRun:\maxValueRun,samples=4] {0.2*x};
			\addplot[dotted,color=black,domain=\minValueRun:\maxValueRun,samples=4] {25*x};
			\addplot[dotted,color=black,domain=\minValueRun:\maxValueRun,samples=4] {0.04*x};
			\addplot[loosely dotted,color=black,domain=\minValueRun:\maxValueRun,samples=4] {125*x};
			\addplot[loosely dotted,color=black,domain=\minValueRun:\maxValueRun,samples=4] {0.008*x};
		\end{loglogaxis}
	\end{tikzpicture}%
	\caption{Comparison of the number of~maximal $\Delta$-$k$-plexes and connected~maximal $\Delta$-$k$-plexes on all considered data sets for different values of~$k$.
	We only included instances where both numbers could be computed within one hour.}
	\label{fig:numberkplex}
\end{figure}
\toappendix{
	\begin{figure}[!hp]
		\centering
		\def\maxValueRun{3600}
		\def\minValueRun{1}
		\begin{tikzpicture}[scale=1]
			\begin{loglogaxis}[
						width=0.9\textwidth,
						height=0.5\textwidth,
						ylabel={with \textbf{Pivoting}},
						xlabel={running time of \BKB{}},
						legend cell align=left,
						legend pos=north west
						,xmax=\maxValueRun
						,xmin=\minValueRun
						,ymax=\maxValueRun
						,ymin=\minValueRun,
						colorbar,
						colorbar style={
	                		scale,
	                		ylabel=$k$,
	                		ylabel style={
								yshift=-7em
							},
							ytick={0,1,2,3,4},
							yticklabel={$\pgfmathprintnumber{\tick}$},
						}
				]
	 			\addplot[scatter src=explicit,scatter,only marks,mark=o, thick] table[col sep=comma,y={runtime_classic_pivo2}, x={runtime_classic2}, meta={k}] {resultsAll.csv};
				\addplot[color=black,domain=\minValueRun:\maxValueRun,samples=4] {x};
				\addplot[dashed,color=black!75,domain=\minValueRun:\maxValueRun,samples=4] {5*x};
				\addplot[dashed,color=black!75,domain=\minValueRun:\maxValueRun,samples=4] {0.2*x};
				\addplot[dotted,color=black,domain=\minValueRun:\maxValueRun,samples=4] {25*x};
				\addplot[dotted,color=black,domain=\minValueRun:\maxValueRun,samples=4] {0.04*x};
			\end{loglogaxis}
		\end{tikzpicture}
		\caption{Comparison of the running times (in seconds) of \BKB{} with and without pivoting for different~$k$- and~$\Delta$-values (scaled by~$\omega/5m$) on all considered data sets.
		Here we include the~$k$-value to show that it plays no significant rule for pivoting.}
		\label{fig:pivotk}
	\end{figure}
}

We next investigate the effects of using a connectedness criterion.
\Cref{fig:numberkplex} shows the difference in number of~$\Delta$-$k$-plexes and connected~$\Delta$-$k$-plexes (of order at least~${2k+1}$).
A similar plot correlating the difference in running time with the~$\Delta$-value can be found in the Appendix (\cref{fig:kplexDelta}).
\toappendix{
\begin{figure}[!hp]
	\centering
	\def\maxValueRun{5000000}
	\def\minValueRun{1}
	\begin{tikzpicture}[scale=1]
		\begin{loglogaxis}[
 					width=0.9\textwidth,
 					height=0.5\textwidth,
					xlabel={\# \textbf{all} $\Delta$-$k$-plexes},
					ylabel={\# \textbf{connected} $\Delta$-$k$-plexes},
					legend cell align=left,
					legend pos=north west,
					xmax=\maxValueRun,
					xmin=5000,
					ymax=\maxValueRun,
					ymin=\minValueRun,
					ymode=log,
					colorbar,
					colorbar style={
                		scale,
                		ylabel=$\Delta$-value,
                		ylabel style={
							yshift=-7em
						},
						ytick={1,2,...,8},
						yticklabel={$5^{\pgfmathprintnumber{\tick}}$},
				}
			]
 			\addplot[scatter src=explicit,scatter,only marks,mark=o, thick] table[col sep=comma,y={num_kplexes_conn}, x={num_kplexes_classic}, meta={delta_number}] {resultsAll.csv};
			\addplot[color=black,domain=\minValueRun:\maxValueRun,samples=4] {x};
			\addplot[dashed,color=black!75,domain=\minValueRun:\maxValueRun,samples=4] {5*x};
			\addplot[dashed,color=black!75,domain=\minValueRun:\maxValueRun,samples=4] {0.2*x};
			\addplot[dotted,color=black,domain=\minValueRun:\maxValueRun,samples=4] {25*x};
			\addplot[dotted,color=black,domain=\minValueRun:\maxValueRun,samples=4] {0.04*x};
			\addplot[loosely dotted,color=black,domain=\minValueRun:\maxValueRun,samples=4] {125*x};
			\addplot[loosely dotted,color=black,domain=\minValueRun:\maxValueRun,samples=4] {0.008*x};
		\end{loglogaxis}
	\end{tikzpicture}%
	\caption{Comparison of the number of~maximal $\Delta$-$k$-plexes and connected maximal~$\Delta$-$k$-plexes on all considered data sets for different values of~$\Delta$ (scaled by~$\omega / 5m$).}
	\label{fig:kplexDelta}
\end{figure}
}%
As expected, the number of connected~\dkps{} is significantly smaller, especially for larger $k$-values as this excludes many trivial \dkps{}, while the~$\Delta$-value seems to play a negligible role.

\Cref{fig:runtimeconnected} shows the running-time difference of \BKB{} with and without connectedness criterion.
\begin{figure}[!t]
	\centering
	\def\maxValueRun{3600}
	\def\minValueRun{1}
	\begin{tikzpicture}[scale=.9]
		\begin{loglogaxis}[
					width=0.9\textwidth,
					height=0.5\textwidth,
					xlabel={running time of \BKB{}},
					ylabel={with \textbf{Connectedness Criterion}},
					legend cell align=left,
					legend pos=north west
					,xmax=\maxValueRun
					,xmin=1
					,ymax=\maxValueRun
					,ymin=\minValueRun,
					colorbar,
					colorbar style={
                		scale,
                		ylabel=$k$,
                		ylabel style={
							yshift=-7em
						},
						ytick={1,2,...,11},
						yticklabel={$\pgfmathprintnumber{\tick}$},
					}
			]
 			\addplot[scatter src=explicit,scatter,only marks,mark=o, thick] table[col sep=comma,y={runtime_conn2}, x={runtime_classic2}, meta={k}] {resultsAll.csv};
			\addplot[color=black,domain=\minValueRun:\maxValueRun,samples=4] {x};
			\addplot[dashed,color=black!75,domain=\minValueRun:\maxValueRun,samples=4] {5*x};
			\addplot[dashed,color=black!75,domain=\minValueRun:\maxValueRun,samples=4] {0.2*x};
			\addplot[dotted,color=black,domain=\minValueRun:\maxValueRun,samples=4] {25*x};
			\addplot[dotted,color=black,domain=\minValueRun:\maxValueRun,samples=4] {0.04*x};
		\end{loglogaxis}
	\end{tikzpicture}
		\caption{Comparison of the running times (in seconds) of \BKB{} with and without the connectedness criterion for different~$k$-values.
		We omitted the~$\Delta$-value here as we could not measure any correlation between the difference in running time and the~$\Delta$-value.
		An according plot can be found in the Appendix (\cref{fig:connectedDelta}) together with similar plots and correlations to different parameters (\cref{fig:connectedVertices,fig:connectedDSD}).}
		\label{fig:runtimeconnected}
~\\
~\\
	\centering
	\def\maxValueRun{3600}
	\def\minValueRun{1}
	\begin{tikzpicture}[scale=.9]
		\begin{loglogaxis}[
					width=0.9\textwidth,
					height=0.5\textwidth,
					ylabel={with \textbf{Pivoting}},
					xlabel={running time of \BKB{} with \textbf{Connectedness Criterion}},
					legend cell align=left,
					legend pos=north west
					,xmax=\maxValueRun
					,xmin=\minValueRun
					,ymax=\maxValueRun
					,ymin=\minValueRun,
					colorbar,
					colorbar style={
                		scale,
                		ylabel=$\Delta$-value,
                		ylabel style={
							yshift=-7em
						},
						ytick={1,2,...,11},
						yticklabel={$5^{\pgfmathprintnumber{\tick}}$},
					}
			]
 			\addplot[scatter src=explicit,scatter,only marks,mark=o, thick] table[col sep=comma,y={runtime_conn_pivo2}, x={runtime_conn2}, meta={delta_number}] {resultsAll.csv};
			\addplot[color=black,domain=\minValueRun:\maxValueRun,samples=4] {x};
			\addplot[dashed,color=black!75,domain=\minValueRun:\maxValueRun,samples=4] {5*x};
			\addplot[dashed,color=black!75,domain=\minValueRun:\maxValueRun,samples=4] {0.2*x};
			\addplot[dotted,color=black,domain=\minValueRun:\maxValueRun,samples=4] {25*x};
			\addplot[dotted,color=black,domain=\minValueRun:\maxValueRun,samples=4] {0.04*x};
		\end{loglogaxis}
	\end{tikzpicture}
	\caption{Comparison of the running times (in seconds) of \BKB{} with our connectedness criterion with and without pivoting for different~$k$- and~$\Delta$-values (scaled by~$\omega /5m$) on all considered data sets.
	We omitted the~$k$-value in this illustration, \cref{fig:pivotk} (Appendix) illustrates the correlation of the running-time difference with~$k$.}
	\label{fig:connectedPivot}
\end{figure}
One can observe that similar to the number of~\dkps{} the connectedness criterion improves the running time especially for large~$k$-values while all other measured parameters seem to be mostly irrelevant.
\toappendix{
	\begin{figure}[!hp]
		\centering
		\def\maxValueRun{3600}
		\def\minValueRun{1}
		\begin{tikzpicture}[scale=1]
			\begin{loglogaxis}[
						width=0.9\textwidth,
						height=0.5\textwidth,
						xlabel={running time of \BKB{}},
						ylabel={with \textbf{Connectedness Criterion }},
						legend cell align=left,
						legend pos=north west
						,xmax=\maxValueRun
						,xmin=1
						,ymax=\maxValueRun
						,ymin=\minValueRun,
						colorbar,
						colorbar style={
	                		scale,
	                		ylabel=$\Delta$-value,
	                		ylabel style={
								yshift=-7em
							},
							ytick={1,2,...,11},
							yticklabel={$5^{\pgfmathprintnumber{\tick}}$},
						}
				]
	 			\addplot[scatter src=explicit,scatter,only marks,mark=o, thick] table[col sep=comma,y={runtime_conn2}, x={runtime_classic2}, meta={delta_number}] {resultsAll.csv};
				\addplot[color=black,domain=\minValueRun:\maxValueRun,samples=4] {x};
				\addplot[dashed,color=black!75,domain=\minValueRun:\maxValueRun,samples=4] {5*x};
				\addplot[dashed,color=black!75,domain=\minValueRun:\maxValueRun,samples=4] {0.2*x};
				\addplot[dotted,color=black,domain=\minValueRun:\maxValueRun,samples=4] {25*x};
				\addplot[dotted,color=black,domain=\minValueRun:\maxValueRun,samples=4] {0.04*x};
			\end{loglogaxis}
		\end{tikzpicture}
			\caption{Comparison of the running times (in seconds) of \BKB{} with and without the connectedness criterion for different~$\Delta$-values (scaled by $\omega /5m$).}
			\label{fig:connectedDelta}
	\end{figure}
	
	\begin{figure}[!hp]
	\centering
	\def\maxValueRun{3600}
	\def\minValueRun{1}
	\begin{tikzpicture}[scale=1]
		\begin{loglogaxis}[
					width=0.9\textwidth,
					height=0.5\textwidth,
					ylabel={with \textbf{Connectedness Criterion}},
					xlabel={running time of \BKB{}},
					legend cell align=left,
					legend pos=north west
					,xmax=\maxValueRun
					,xmin=1
					,ymax=\maxValueRun
					,ymin=\minValueRun,
					colorbar,
					colorbar style={
						scaled ticks= true,
                		ylabel=\# vertices,
                		ylabel style={
							yshift=-7em
						},
						yticklabel={$10^{\pgfmathprintnumber{\tick}}$},
					}
			]
 			\addplot[scatter src=explicit,scatter,only marks,mark=o, thick] table[col sep=comma,y={runtime_conn2}, x={runtime_classic2}, meta expr=lg10(\thisrow{vertices}) ] {resultsAll.csv};
			\addplot[color=black,domain=\minValueRun:\maxValueRun,samples=4] {x};
			\addplot[dashed,color=black!75,domain=\minValueRun:\maxValueRun,samples=4] {5*x};
			\addplot[dashed,color=black!75,domain=\minValueRun:\maxValueRun,samples=4] {0.2*x};
			\addplot[dotted,color=black,domain=\minValueRun:\maxValueRun,samples=4] {25*x};
			\addplot[dotted,color=black,domain=\minValueRun:\maxValueRun,samples=4] {0.04*x};
		\end{loglogaxis}
	\end{tikzpicture}
	\caption{Comparison of the running times (in seconds) of \BKB{} with and without the connectedness criterion in correlation with the number of vertices.}
	\label{fig:connectedVertices}
	\end{figure}
	\begin{figure}
	\centering
	\def\maxValueRun{3600}
	\def\minValueRun{1}
	\begin{tikzpicture}[scale=1]
		\begin{loglogaxis}[
					width=0.9\textwidth,
					height=0.5\textwidth,
					ylabel={with \textbf{Connectedness Criterion}},
					xlabel={running time of \BKB{}},
					legend cell align=left,
					legend pos=north west
					,xmax=\maxValueRun
					,xmin=1
					,ymax=\maxValueRun
					,ymin=\minValueRun,
					colorbar,
					colorbar style={
						scaled ticks= true,
                		ylabel=$\Delta$-slice degeneracy,
                		ylabel style={
							yshift=-7em
						},
						yticklabel={${\pgfmathprintnumber{\tick}}$},
					}
			]
 			\addplot[scatter src=explicit,scatter,only marks,mark=o, thick] table[col sep=comma,y={runtime_conn2}, x={runtime_classic2}, meta={degeneracy} ] {resultsAll.csv};
			\addplot[color=black,domain=\minValueRun:\maxValueRun,samples=4] {x};
			\addplot[dashed,color=black!75,domain=\minValueRun:\maxValueRun,samples=4] {5*x};
			\addplot[dashed,color=black!75,domain=\minValueRun:\maxValueRun,samples=4] {0.2*x};
			\addplot[dotted,color=black,domain=\minValueRun:\maxValueRun,samples=4] {25*x};
			\addplot[dotted,color=black,domain=\minValueRun:\maxValueRun,samples=4] {0.04*x};
		\end{loglogaxis}
	\end{tikzpicture}
	\caption{Comparison of the running times (in seconds) of \BKB{} with and without the connectedness criterion in correlation with the $\Delta$-slice degeneracy.}
	\label{fig:connectedDSD}
\end{figure}
}

It remains to analyze how the two approaches work together.
\Cref{fig:connectedPivot} shows the running time difference of \BKB{} with the connectedness criterion with and without pivoting.
\toappendix{
	\begin{figure}[!hp]
		\centering
		\def\maxValueRun{3600}
		\def\minValueRun{1}
		\begin{tikzpicture}[scale=1]
			\begin{loglogaxis}[
						width=0.9\textwidth,
						height=0.5\textwidth,
						ylabel={with \textbf{Pivoting}},
						xlabel={running time of \BKB{} with \textbf{Connectedness Criterion}},
						legend cell align=left,
						legend pos=north west
						,xmax=\maxValueRun
						,xmin=\minValueRun
						,ymax=\maxValueRun
						,ymin=\minValueRun,
						colorbar,
						colorbar style={
	                		scale,
	                		ylabel=k,
	                		ylabel style={
								yshift=-7em
							},
							ytick={1,2,...,11},
							yticklabel={$\pgfmathprintnumber{\tick}$},
						}
				]
	 			\addplot[scatter src=explicit,scatter,only marks,mark=o, thick] table[col sep=comma,y={runtime_conn_pivo2}, x={runtime_conn2}, meta={k}] {resultsAll.csv};
				\addplot[color=black,domain=\minValueRun:\maxValueRun,samples=4] {x};
				\addplot[dashed,color=black!75,domain=\minValueRun:\maxValueRun,samples=4] {5*x};
				\addplot[dashed,color=black!75,domain=\minValueRun:\maxValueRun,samples=4] {0.2*x};
				\addplot[dotted,color=black,domain=\minValueRun:\maxValueRun,samples=4] {25*x};
				\addplot[dotted,color=black,domain=\minValueRun:\maxValueRun,samples=4] {0.04*x};
			\end{loglogaxis}
		\end{tikzpicture}
		\caption{Comparison of the running times (in seconds) of \BKB{} with our connectedness criterion with and without pivoting for different~$k$-values.}
		\label{fig:connectedPivotk}
	\end{figure}
}%
We measured that pivoting improves the (relative) running time a little further when combined with the connectedness criterion in comparison to the classic setting (without connectedness criterion).

\subsection{Running Time}
In this subsection we compare the running time of \BKB{} with the algorithm of \citet{ViardML18} which we call~\viard.
Note that \viard{} enumerates~$\Delta$-cliques ($\Delta$-$1$-plexes) and hence we use~$k=1$ and no connectedness criterion for the comparison.
\cref{fig:VML} shows a comparison of the two algorithms and \cref{data_runningTime_Clique} shows the running times of the two algorithms and the algorithm by Himmel et al.~\cite{himmel2017adapting} for each data set.
	\begin{figure}[!t]
		\centering
		\def\maxValueRun{3600}
		\def\minValueRun{1}
		\begin{tikzpicture}[scale=1]
			\begin{loglogaxis}[
						width=0.9\textwidth,
						height=0.5\textwidth,
						ylabel={running time of \textbf{\viard{}}},
						xlabel={running time of \BKB{}},
						legend cell align=left,
						legend pos=north west
						,xmax=\maxValueRun
						,xmin=\minValueRun
						,ymax=\maxValueRun
						,ymin=\minValueRun,
						colorbar,
						colorbar style={
	                		scale,
	                		ylabel=\# vertices,
	                		ylabel style={
								yshift=-7em
							},
							yticklabel={$10^{\pgfmathprintnumber{\tick}}$},
						}
				]
	 			\addplot[scatter src=explicit,scatter,only marks,mark=o, thick] table[col sep=comma,y={running time_vlm22}, x={runtime_kplex2}, meta expr=lg10(\thisrow{vertices}] {resultsClique2.csv};
				\addplot[color=black,domain=\minValueRun:\maxValueRun,samples=4] {x};
				\addplot[dashed,color=black!75,domain=\minValueRun:\maxValueRun,samples=4] {5*x};
				\addplot[dashed,color=black!75,domain=\minValueRun:\maxValueRun,samples=4] {0.2*x};
				\addplot[dotted,color=black,domain=\minValueRun:\maxValueRun,samples=4] {25*x};
				\addplot[dotted,color=black,domain=\minValueRun:\maxValueRun,samples=4] {0.04*x};
			\end{loglogaxis}
		\end{tikzpicture}
		\caption{Comparison of the running times (in seconds) of \BKB{} and \viard{} for different~$\Delta$-values.}
		\label{fig:VML}
	\end{figure} 
	 \begin{table*}[t]
\footnotesize
  \caption{Running time comparison for finding all maximal $\Delta$-cliques in a temporal graph. The abbreviation~1Plex denotes the running time of \BKB, \viard{} refers to the running time of the algorithm by \citet{ViardML18}, and \bkd{} denotes the running time of the algorithm by \citet{himmel2017adapting}. Empty cells represent that the time limit of one hour was exceeded. The algorithms \viard{} and \bkd{} add $\Delta$ empty time steps at the beginning and at the end. Hence, the lifetime of the temporal graphs had to be adapted to ensure comparability to the other algorithms.}%
  \centering
  \begin{filecontents*}{resultsClique.csv}
name,delta0,runtime_kplex0,num_calls_kplex0,num_kplexes_kplex0,running time_bkd0,num_calls_bkd0,num_cliques_bkd0,running time_vlm20,num_calls_vlm20,num_cliques_vlm20,delta1,runtime_kplex1,num_calls_kplex1,num_kplexes_kplex1,running time_bkd1,num_calls_bkd1,num_cliques_bkd1,running time_vlm21,num_calls_vlm21,num_cliques_vlm21,delta2,runtime_kplex2,num_calls_kplex2,num_kplexes_kplex2,running time_bkd2,num_calls_bkd2,num_cliques_bkd2,running time_vlm22,num_calls_vlm22,num_cliques_vlm22,delta3,runtime_kplex3,num_calls_kplex3,num_kplexes_kplex3,running time_bkd3,num_calls_bkd3,num_cliques_bkd3,running time_vlm23,num_calls_vlm23,num_cliques_vlm23,delta4,runtime_kplex4,num_calls_kplex4,num_kplexes_kplex4,running time_bkd4,num_calls_bkd4,num_cliques_bkd4,running time_vlm24,num_calls_vlm24,num_cliques_vlm24,delta5,runtime_kplex5,num_calls_kplex5,num_kplexes_kplex5,running time_bkd5,num_calls_bkd5,num_cliques_bkd5,running time_vlm25,num_calls_vlm25,num_cliques_vlm25,delta6,runtime_kplex6,num_calls_kplex6,num_kplexes_kplex6,running time_bkd6,num_calls_bkd6,num_cliques_bkd6,running time_vlm26,num_calls_vlm26,num_cliques_vlm26,,,,,,,,,,,,,,,,,,,,,,,,,,,,,,,,,,,,,,,,,,,,,,,,,,
flights,0,, , ,, , ,, , ,0,, , ,, , ,, , ,0,, , ,, , ,, , ,0,, , ,, , ,, , ,0,, , ,, , ,, , ,0,, , ,, , ,, , ,2,, , ,, , ,, , ,13,, , ,, , ,, , ,67,, , ,, , ,, , ,336,, , ,, , ,, , ,,,,,,,,,,,,,,,,,,,,
london,0,0.995564937591553,602,28500,, , ,7.61694312095642,0,294269,0,0.995564937591553,602,28500,, , ,7.61694312095642,0,294269,0,0.995564937591553,602,28500,, , ,7.61694312095642,0,294269,0,0.995564937591553,602,28500,, , ,7.61694312095642,0,294269,0,0.995564937591553,602,28500,, , ,7.61694312095642,0,294269,1,0.677000045776367,602,10632,11.2941839694977,11171,10632,7.2393159866333,0,294269,8,0.509991884231567,602,1192,0.730263948440552,1218,1192,0.026844024658203,0,969,42,0.523868083953857,602,1048,0.696253061294556,1071,1048,0.022828102111817,0,778,213,0.497265100479126,602,583,0.691809177398682,611,583,0.008951902389526,0,326,1067,0.489720821380615,602,577,0.62024998664856,602,577,0.008249998092651,0,320,,,,,,,,,,,,,,,,,,,,
paris,0,1.58146905899048,671,50702,, , ,8.02345895767212,0,342540,0,1.58146905899048,671,50702,, , ,8.02345895767212,0,342540,0,1.58146905899048,671,50702,, , ,8.02345895767212,0,342540,0,1.58146905899048,671,50702,, , ,8.02345895767212,0,342540,0,1.58146905899048,671,50702,, , ,8.02345895767212,0,342540,1,0.830420970916748,671,13087,9.76262497901917,13343,13087,8.15845608711243,0,342540,7,0.618497133255005,671,1403,0.86834979057312,1424,1403,0.032827138900757,0,1164,38,0.639312028884888,671,1395,0.853507041931152,1416,1395,0.032464027404785,0,1093,193,0.616575956344605,671,642,0.785978078842163,680,642,0.00963306427002,0,367,967,0.599287986755371,671,636,0.77082085609436,671,636,0.009322166442871,0,361,,,,,,,,,,,,,,,,,,,,
ny,0,3.30335307121277,991,115718,, , ,11.7525489330292,0,442341,0,3.30335307121277,991,115718,, , ,11.7525489330292,0,442341,0,3.30335307121277,991,115718,, , ,11.7525489330292,0,442341,0,3.30335307121277,991,115718,, , ,11.7525489330292,0,442341,0,3.30335307121277,991,115718,, , ,11.7525489330292,0,442341,1,2.17711591720581,992,58541,106.318140029907,59937,58541,11.2493319511414,0,442341,6,1.19080495834351,992,6347,3.21372699737549,6642,6347,0.276515960693359,0,8419,34,1.10847210884094,992,1087,1.40935611724854,1190,1087,0.023652076721192,0,670,170,1.09510898590088,992,1055,1.46682000160217,1152,1055,0.022188901901245,0,638,853,1.09456086158752,992,957,1.45548486709595,992,957,0.015679836273193,0,540,,,,,,,,,,,,,,,,,,,,
highschool-2011,0,2.18383193016052,2443,26510,128.660057783127,30071,26510,2.63053107261658,0,26384,9,2.2648868560791,2443,26504,131.682803869247,30061,26504,5.06976389884949,0,26384,47,1.157310962677,3045,9601,11.2152349948883,10718,9601,1.90854597091675,0,9475,238,1.14942407608032,4619,7394,5.35186982154846,9185,7394,1.93180012702942,0,7268,1191,1.796550989151,9732,7337,3.95813608169556,12715,7337,3.06031608581543,0,7211,5959,2.61050510406494,15200,7760,4.18444299697876,17703,7760,5.1487729549408,0,7643,29798,9.50847482681274,53672,12499,9.38425207138062,56157,12499,22.4922590255737,0,12373,148990,66.4621698856354,287558,21791,47.8061759471893,292596,21791,172.914883852005,0,21665,,,,,,,,,,,,,,,,,,,,,,,,,,,,,,,,,,,,,,,,
highschool-2012,0,3.26321291923523,3012,42285,253.41979598999,47079,42285,4.05473899841309,0,42105,16,3.27020883560181,3012,42285,245.818442821503,47079,42285,9.12778496742249,0,42105,80,1.55708193778992,3543,13705,20.0087261199951,14821,13705,2.3634340763092,0,13525,404,1.30721497535706,4108,9501,8.2542610168457,10657,9501,1.76675891876221,0,9321,2024,1.24275994300842,4715,7688,4.83829116821289,9128,7688,1.4942638874054,0,7508,10121,1.34871315956116,6074,7536,3.58974814414978,9579,7536,1.69816994667053,0,7356,50606,1.7050518989563,9128,8357,3.3893039226532,12143,8357,2.55718302726746,0,8177,253034,6.3624861240387,38986,14864,7.55158615112305,44644,14864,14.35178399086,0,14684,,,,,,,,,,,,,,,,,,,,,,,,,,,,,,,,,,,,,,,,
highschool-2013,0,17.1229259967804,8761,172362,1953.91768097878,199249,172362,25.6885418891907,0,172035,1,16.9824020862579,8761,172362,1969.38795900345,199249,172362,23.6957919597626,0,172035,9,16.596832036972,8761,172362,1931.48483896256,199249,172362,47.3336679935455,0,172035,48,7.05205011367798,10772,57121,147.32005405426,63785,57121,14.3340220451355,0,56794,241,5.72824001312256,15068,35941,45.6185369491577,43049,35941,9.87557911872864,0,35614,1205,6.15944719314575,23573,29752,22.9237589836121,42092,29752,11.1370449066162,0,29425,6026,7.2617678642273,32876,30069,16.6736631393433,49255,30069,14.4197459220886,0,29742,30134,18.8399589061737,99517,41947,22.9484479427338,117196,41947,47.2222208976746,0,41620,150673,148.544738054276,676786,85565,125.584292888641,752925,85565,432.892069101334,0,85238,,,,,,,,,,,,,,,,,,,,,,,,,,,,,,
primaryschool,0,20.9873540401459,14089,107121,990.225807905197,137796,107121,32.1029977798462,0,106879,0,20.9873540401459,14089,107121,990.225807905197,137796,107121,32.1029977798462,0,106879,4,20.9722619056702,14089,107121,986.752471923828,137796,107121,61.6921539306641,0,106879,23,15.9503891468048,17434,80648,378.742170095444,95956,80648,44.4291698932648,0,80406,116,14.8239159584045,26194,68705,177.539680957794,85443,68705,38.8747820854187,0,68463,580,24.5186841487885,84409,83314,83.4711651802063,146925,83314,75.8143620491028,0,83072,2904,222.250084161758,899607,186411,214.104776144028,1165631,186411,848.749577999115,0,186169,14522,, , ,, , ,, , ,72613,, , ,, , ,, , ,,,,,,,,,,,,,,,,,,,,,,,,,,,,,,
hospital-ward,0,3.44594407081604,1988,27910,381.526314973831,35197,27910,4.41420888900757,0,27835,10,3.46505308151245,1988,27910,379.883021831512,35197,27910,9.4495542049408,0,27835,53,1.6851749420166,2502,12560,44.8973729610443,14819,12560,4.10285305976868,0,12485,267,1.4797089099884,3739,8694,14.8437669277191,11491,8694,3.25266098976135,0,8619,1339,2.10651397705078,7843,8777,7.49431085586548,15638,8777,4.55797290802002,0,8728,6698,8.32586598396301,39545,10869,9.70006394386292,53657,10869,20.9704658985138,0,10794,33491,25.2751970291138,120289,11696,20.4893610477448,127121,11696,59.2587268352509,0,11621,167458,258.753175973892,974107,20519,185.237835168839,1055918,20519,670.018796920776,0,20444,,,,,,,,,,,,,,,,,,,,,,,,,,,,,,,,,,,,,,,,
infectious,0,783.718647956848,73219,349787,1448.52262091637,478305,349787,19.1187219619751,0,338815,16,759.043252944946,73219,349787,1457.70384788513,478305,349787,38.9064722061157,0,338815,83,793.054938793182,96662,161947,1110.83557987213,210784,161947,19.794429063797,0,150975,417,751.468516111374,170587,134787,1116.80169200897,240861,134787,33.4897019863129,0,123815,2087,1050.77435183525,1244253,157700,1297.21747303009,1329925,157700,474.627804040909,0,146728,10438,1137.30422091484,1815707,163162,1382.95061707497,1815707,163162,758.523902893066,0,152190,52192,1178.16027498245,1815707,163162,1353.26832413673,1815707,163162,707.565863847733,0,152190,260960,1150.05605697632,1815707,163162,1365.42335891724,1815707,163162,702.368739843369,0,152190,1304802,1155.77298498154,1815707,163162,1360.4157781601,1815707,163162,744.05947804451,0,152190,6524010,1150.63681793213,1815707,163162,1353.61903500557,1815707,163162,698.512861013413,0,152190,,,,,,,,,,,,,,,,,,,,
hypertext,0,2.44483184814453,2788,19150,86.8043849468231,22137,19150,2.73946595191956,0,19037,10,2.44200587272644,2788,19150,88.7530901432037,22137,19150,5.32926392555237,0,19037,50,1.50026988983154,3109,8763,14.0169930458069,9690,8763,2.27083206176758,0,8650,254,1.3846230506897,3523,6345,6.69680404663086,7228,6345,1.66895079612732,0,6232,1274,1.40481400489807,4440,6042,4.73697900772095,7122,6042,1.72912907600403,0,5929,6374,1.81547594070435,7783,7459,3.87480998039246,9974,7459,2.7738208770752,0,7346,31874,7.24077486991882,41604,16543,8.22752499580383,44037,16543,20.524533033371,0,16430,159372,151.220736026764,638610,65945,118.847393989563,681774,65945,422.806932926178,0,65832,,,,,,,,,,,,,,,,,,,,,,,,,,,,,,,,,,,,,,,,
dnc,0,23.3990859985352,6251,33568,223.157368898392,33569,33568,8.25224089622498,0,31725,2267,24.510281085968,9882,28433,104.563902139664,32059,28433,16.3859221935272,0,26576,11339,24.5278840065002,13573,25548,67.8904089927673,35077,25548,20.227667093277,0,23681,56699,25.500027179718,20442,24519,52.969386100769,44948,24519,31.8832159042358,0,22654,283496,35.0691170692444,73821,19169,46.4959402084351,113573,19169,110.194423913956,0,17303,1417482,106.619767904282,370853,16368,85.5617690086365,381006,16368,432.513975143433,0,14502,7087410,155.873016119003,528803,16972,110.00202703476,528881,16972,612.867604017258,0,15106,35437051,154.049250125885,535270,17397,110.347667217255,535270,17397,673.510179996491,0,15531,,,,,,,,,,,,,,,,,,,,,,,,,,,,,,,,,,,,,,,,
karlsruhe,0,79.7601788043976,18983,459230,, , ,125.554006099701,0,460070,268,72.7428071498871,20031,386233,, , ,215.279684066772,0,384497,1341,68.0343871116638,22158,337230,, , ,193.616269111633,0,335392,6706,65.9781260490418,25826,297293,, , ,171.311750888824,0,295426,33530,65.9510860443115,31640,272374,2621.10602998733,287793,272374,163.370349884033,0,270504,167651,62.9985601902008,46653,235684,1252.63596200943,277826,235684,163.809178829193,0,233814,838258,75.4411370754242,117637,203359,382.940582036972,352990,203359,235.798758983612,0,201489,4191294,258.958190917969,827580,203603,312.09579706192,1289170,203603,1212.33320808411,0,201733,20956473,, , ,, , ,, , ,104782367,, , ,, , ,, , ,,,,,,,,,,,,,,,,,,,,
facebook-like,0,29.7054281234741,15738,61648,, , ,10.0137670040131,0,59795,279,28.1867980957031,15773,45562,, , ,14.7424259185791,0,43722,1398,28.0057978630066,15871,38384,, , ,11.4417119026184,0,36487,6992,28.1804451942444,16108,33876,, , ,10.5296330451965,0,31978,34963,29.2418839931488,16531,28412,, , ,8.09879422187805,0,26513,174815,27.7758860588074,18184,23247,, , ,7.15765500068665,0,21349,874079,28.4540050029755,23613,24509,, , ,9.92086315155029,0,22610,4370399,30.4327211380005,32765,29397,, , ,17.4021039009094,0,27498,,,,,,,,,,,,,,,,,,,,,,,,,,,,,,,,,,,,,,,,
as-733,0,, , ,, , ,, , ,5,, , ,, , ,, , ,29,, , ,, , ,, , ,148,, , ,, , ,, , ,743,, , ,, , ,, , ,3718,, , ,, , ,, , ,18591,, , ,, , ,, , ,92958,653.953816890717,216037,348147,, , ,, , ,464794,465.542490959168,236033,93279,1004.58517003059,647804,93279,, , ,2323974,466.10139298439,318802,53103,719.970436811447,508303,53103,, , ,11619873,552.23959684372,737362,52319,789.50771188736,905897,52319,, , ,58099368,2074.4678170681,5347281,65598,1687.0066139698,5705838,65598,, , 
\end{filecontents*}
\pgfplotstabletypeset[
   col sep=comma,
   columns={name, runtime_kplex0,running time_vlm20,running time_bkd0,runtime_kplex3,running time_vlm23,running time_bkd3,runtime_kplex5,running time_vlm25,running time_bkd5},
   columns/name/.style={column type=l,string type,column name={Instance}},
   columns/runtime_kplex0/.style={column name={kPlex},column type=r,precision=0, relative*={0}},
   columns/running time_bkd0/.style={column name={\bkd},column type=r,precision=0, relative*={0}},
   columns/running time_vlm20/.style={column name={VML},column type=r,precision=0, relative*={0}},
   columns/runtime_kplex3/.style={column name={kPlex},column type=r,precision=0, relative*={0}},
   columns/running time_bkd3/.style={column name={\bkd},column type=r,precision=0, relative*={0}},
   columns/running time_vlm23/.style={column name={VML},column type=r,precision=0, relative*={0}},
   columns/runtime_kplex5/.style={column name={kPlex},column type=r,precision=0, relative*={0}},
   columns/running time_bkd5/.style={column name={\bkd},column type=r,precision=0, relative*={0}},
   columns/running time_vlm25/.style={column name={VML},column type=r,precision=0, relative*={0}},   
   every head row/.style={ before row=\toprule & \multicolumn{3}{c}{$\Delta = 0$} & \multicolumn{3}{c}{$\Delta \sim 5^3$} & \multicolumn{3}{c}{$\Delta \sim 5^5$}\\ \cmidrule(r){2-4} \cmidrule(r){5-7} \cmidrule(r){8-10},after row=\midrule},
   every last row/.style={after row=\bottomrule},
   after row={},
   ]{resultsClique.csv}
   \label{data_runningTime_Clique}
 \end{table*}
As one can see, \BKB{} outperforms the state-of-the-art algorithm~\viard{} on all but two graphs (by up to two orders of magnitude).
Notably, \viard{} performs significantly better on the ``infectious'' network for small $\Delta$-values. 
We could not find out which properties of the ``infectious'' network cause this behavior.

\subsection{Results on the Number and Characteristics of $\Delta$-$k$-Plexes}
\toappendix{
\begin{table*}[t]%
\footnotesize
  \caption{$\Delta$-$k$-plex statistics and running times for~$\Delta = 0$ in the classic setting: $\# R$~denotes the number of maximal $\Delta$-$k$-plexes, $|C|_{\max}$ 	denotes the maximum size of a $\Delta$-$k$-plex with respect to the number of vertices, $t$ and~$t_{\text{pivo}}$ denote the running times (in seconds) of \BKB~  without and with pivoting, respectively. The variables $c$ and $c_{\text{pivo}}$ denote the number of recursive calls of \BKB{} without and with pivoting, respectively. Empty cells represent that the time limit of one hour was exceeded. The instances were computed for $k=1,2,3$. 
  	Non-listed data sets signal that the instance could not be solved in the classic setting within the time limit of one hour.  }
  	\centering
  \begin{filecontents*}{results_d0_arxiv.csv}
name, k ,runtime_classic,num_calls_classic,num_kplexes_classic,max_kplex_size_classic,avg_kplex_size_classic,max_kplex_length_classic,avg_kplex_length_classic,runtime_classic_pivo,num_calls_classic_pivo,num_kplexes_classic_pivo,max_kplex_size_classic_pivo,avg_kplex_size_classic_pivo,max_kplex_length_classic_pivo,avg_kplex_length_classic_pivo,runtime_conn,num_calls_conn,num_kplexes_conn,max_kplex_size_conn,avg_kplex_size_conn,max_kplex_length_conn,avg_kplex_length_conn,runtime_conn_pivo,num_calls_conn_pivo,num_kplexes_conn_pivo,max_kplex_size_conn_pivo,avg_kplex_size_conn_pivo,max_kplex_length_conn_pivo,avg_kplex_length_conn_pivo,,,,,,,
\midrule $\boldsymbol{k=1}$,,,,,,,,,,,,,,,,,,,,,,,,,,,,,,,,,,,,
london,1,0.987957954406738,602,28500,3,2.036,1439,23.7933333333333,1.154620885849,602,28500,3,2.036,1439,23.7933333333333,0.850924015045166,122,1296,3,3,526,6.3070987654321,1.05120515823364,122,1296,3,3,526,6.3070987654321,,,,,,,
paris,1,1.5463600159,671,50702,3,2.0211825963,1439,14.5620093882,1.8481550217,671,50702,3,2.0211825963,1439,14.5620093882,1.3275740146637,120,1376,3,3,378,3.14752906976744,1.50421595573425,120,1376,3,3,378,3.14752906976744,,,,,,,
ny,1,3.3815159798,991,115718,3,2.0320952661,1439,8.2934893448,3.9642798901,991,115718,3,2.0320952661,1439,8.2934893448,2.62375617027283,221,4131,3,3,328,2.51246671508109,3.19930386543274,221,4131,3,3,328,2.51246671508109,,,,,,,
highschool-2011,1,2.0545630455,2443,26510,5,2.0386269332,272330,1294.3636363636,2.6926810741,2382,26510,5,2.0386269332,272330,1294.3636363636,1.84502696990967,1084,1069,5,3.07577174929841,0,0,2.4710328578949,1050,1069,5,3.07577174929841,0,0,,,,,,,
highschool-2012,1,2.8359839916,3012,42285,5,2.0319971621,729500,3105.3565094005,4.0445418358,2934,42285,5,2.0319971621,729500,3105.3565094005,2.71820998191833,1151,1434,5,3.06903765690377,0,0,3.6588990688324,1107,1434,5,3.06903765690377,0,0,,,,,,,
highschool-2013,1,14.5932919979,8761,172362,5,2.0482763022,363560,689.7350924218,20.7033081055,8525,172362,5,2.0482763022,363560,689.7350924218,14.1341631412506,4304,8058,5,3.07321916108215,0,0,19.8121280670166,4162,8058,5,3.07321916108215,0,0,,,,,,,
primaryschool,1,18.6299231052399,14089,107121,5,2.09333370674284,116900,264.092008102986,34.8048100471497,13748,107121,5,2.09333370674284,116900,264.092008102986,17.9004309177399,8177,9745,5,3.05079527963058,0,0,25.5673608779907,7924,9745,5,3.05079527963058,0,0,,,,,,,
hospital-ward,1,2.97551798820496,1988,27910,5,2.08448584736653,347500,933.805087782157,4.05345797538757,1947,27910,5,2.08448584736653,347500,933.805087782157,3.09471678733826,1133,2348,5,3.03620102214651,0,0,4.00938606262207,1105,2348,5,3.03620102214651,0,0,,,,,,,
infectious,1,780.119927883148,73219,349787,5,2.08347079794275,6946340,217890.437551996,859.857194900513,69012,349787,5,2.08347079794275,6946340,217890.437551996,716.520969152451,35310,36965,5,3.08667658595969,0,0,827.773962020874,33307,36965,5,3.08667658595969,0,0,,,,,,,
hypertext,1,2.20874309539795,2788,19150,6,2.04506527415144,212340,1252.97232375979,3.00041198730469,2655,19150,6,2.04506527415144,212340,1252.97232375979,2.10710787773132,840,917,6,3.06434023991276,0,0,2.76368808746338,736,917,6,3.06434023991276,0,0,,,,,,,
dnc,1,23.764121055603,6251,33568,2,1.94441134413727,84869755,4717795.60453408,24.2248330116272,6251,33568,2,1.94441134413727,84869755,4717795.60453408,29.6903610229492,6251,33568,2,1.94441134413727,84869755,4717795.60453408,23.3867359161377,362,0,0,0,0,0,29.30814909935,362,0,0,0,0,0
karlsruhe,1,69.6526448726654,18983,459230,3,1.99595191951745,123837267,504269.520719465,103.009526968002,18983,459230,3,1.99595191951745,123837267,504269.520719465,66.8472969532013,1007,11,3,3,0,0,95.996680021286,1007,11,3,3,0,0,,,,,,,
facebook-like,1,28.7566769123077,15738,61648,2,1.96919608097586,16736181,515539.965043473,34.2384910583496,15738,61648,2,1.96919608097586,16736181,515539.965043473,26.8328869342804,1081,0,0,0,0,0,32.6977460384369,1081,0,0,0,0,0,,,,,,,
as-733,1,3099.08181500435,205821,4245340,12,2.51045452189931,67737600,120928.658340675,, , , , , , ,3250.77057886124,186918,1290050,12,3.68569590325956,0,0,, , , , , , ,,,,,,,
\midrule $\boldsymbol{k=2}$,,,,,,,,,,,,,,,,,,,,,,,,,,,,,,,,,,,,
london,2,117.675229072571,37141,81164,4,2.56393228524962,1439,649.297582672121,116.51580286026,37141,81164,4,2.56393228524962,1439,649.297582672121,2.2125129699707,341,0,0,0,0,0,2.39055991172791,341,0,0,0,0,0,,,,,,,
paris,2,200.9043660164,46428,144413,4,2.697859611,1439,456.2263785116,205.2696530819,46428,144413,4,2.697859611,1439,456.2263785116,3.9685230255127,390,0,0,0,0,0,4.17670512199402,390,0,0,0,0,0,,,,,,,
ny,2,583.3912379742,88327,317195,4,2.7405255442,1439,395.2820000315,594.3624989986,88327,317195,4,2.7405255442,1439,395.2820000315,9.65313196182251,579,0,0,0,0,0,9.96362400054932,579,0,0,0,0,0,,,,,,,
highschool-2011,2,38.8545958996,11400,14967,5,2.5027059531,272330,143288.484666266,38.6924951077,11377,14967,5,2.5027059531,272330,143288.484666266,11.291218996048,632,15,5,5,0,0,11.2148518562317,632,15,5,5,0,0,,,,,,,
highschool-2012,2,87.2681910992,19016,24107,5,2.3472020575,729500,487503.422242502,87.971394062,18990,24107,5,2.3472020575,729500,487503.422242502,17.251855134964,545,15,5,5,0,0,18.2304141521454,545,15,5,5,0,0,,,,,,,
highschool-2013,2,631.3113069534,66612,109757,5,2.5360387037,363560,176554.675874887,621.9873850346,66529,109757,5,2.5360387037,363560,176554.675874887,96.0247809886932,2194,81,5,5,0,0,98.1391940116882,2195,81,5,5,0,0,,,,,,,
primaryschool,2,340.781122922897,58346,96690,5,2.73613610507808,116900,35256.188850967,331.757730007172,58166,96690,5,2.73613610507808,116900,35256.188850967,130.170913934708,3863,77,5,5,0,0,127.51246213913,3866,77,5,5,0,0,,,,,,,
hospital-ward,2,25.5156719684601,5358,13211,5,2.83195821663765,347500,72993.1496480206,26.3894579410553,5345,13211,5,2.83195821663765,347500,72993.1496480206,16.421511888504,416,13,5,5,0,0,17.1749410629272,416,13,5,5,0,0,,,,,,,
hypertext,2,27.6703398227692,9239,11565,8,2.47747514051016,212340,116185.691309987,29.1255309581757,9048,11565,8,2.47747514051016,212340,116185.691309987,13.0203990936279,735,35,8,5.48571428571429,0,0,13.6767330169678,638,35,8,5.48571428571429,0,0,,,,,,,
\midrule $\boldsymbol{k=3}$,,,,,,,,,,,,,,,,,,,,,,,,,,,,,,,,,,,,
highschool-2011,3,2279.4216780663,375004,427686,6,3.2395940012,272330,207262.840027497,2381.370513916,375001,427686,6,3.2395940012,272330,207262.840027497,329.255962133408,3889,0,0,0,0,0,326.998162984848,3889,0,0,0,0,0,,,,,,,
hospital-ward,3,901.527784824371,85685,117253,6,3.42573750778232,347500,200122.278321237,941.435976028442,85682,117253,6,3.42573750778232,347500,200122.278321237,541.424566030502,2470,0,0,0,0,0,527.99120092392,2470,0,0,0,0,0,,,,,,,
hypertext,3,1411.93648910522,268821,293253,9,3.20236110116521,212340,169534.286912666,1501.7961139679,268626,293253,9,3.20236110116521,212340,169534.286912666,557.8901450634,3656,7,9,7.28571428571429,0,0,565.334822177887,3626,7,9,7.28571428571429,0,0,,,,,,,
\end{filecontents*}
\pgfplotstabletypeset[
   col sep=comma,
   columns={name,num_kplexes_classic,max_kplex_size_classic, runtime_classic,num_calls_classic,runtime_classic_pivo,num_calls_classic_pivo},
   columns/name/.style={column type=l,string type,column name={Instance}},
   columns/num_kplexes_classic/.style={column name={$\# R$},column type=r,precision=0},
   columns/max_kplex_size_classic/.style={column name={$|C|_{\max}$},column type=r,precision=1},
   columns/runtime_classic/.style={column name={$t$},column type=r,precision=0, relative*={0}},
   columns/num_calls_classic/.style={column name={$c$},column type=r,int detect},
   columns/runtime_classic_pivo/.style={column name={$t_{\text{pivo}}$},column type=r,precision=1, relative*={0}},
   columns/num_calls_classic_pivo/.style={column name={$c_{\text{pivo}}$},column type=r,int detect},  
   every head row/.style={ before row=\toprule & & \multicolumn{5}{c}{Classical setting for $\Delta = 0 $ } \\ \cmidrule(r){2-7} ,after row=\midrule},
   every last row/.style={after row=\bottomrule},
   after row={},
   ]{results_d0_arxiv.csv}
   \label{data_delta01}
 \end{table*}%
\begin{table*}[t]%
\footnotesize
  \caption{$\Delta$-$k$-plex statistics and running times for~$\Delta \sim 5^3$ in the classic setting: $\# R$~denotes the number of maximal $\Delta$-$k$-plexes, $|C|_{\max}$ 	denotes the maximum size of a $\Delta$-$k$-plex with respect to the number of vertices, $t$ and~$t_{\text{pivo}}$ denote the running times (in seconds) of \BKB~  without and with pivoting, respectively. The variables $c$ and $c_{\text{pivo}}$ denote the number of recursive calls of \BKB{} without and with pivoting, respectively. Empty cells represent that the time limit of one hour was exceeded. The instances were computed for $k=1,2,3$. 
  	Non-listed data sets signal that the instance could not be solved in the classic setting within the time limit of one hour.}
  	\centering
  \begin{filecontents*}{results_d3_arxiv.csv}
name, k ,runtime_classic,num_calls_classic,num_kplexes_classic,max_kplex_size_classic,avg_kplex_size_classic,max_kplex_length_classic,avg_kplex_length_classic,runtime_classic_pivo,num_calls_classic_pivo,num_kplexes_classic_pivo,max_kplex_size_classic_pivo,avg_kplex_size_classic_pivo,max_kplex_length_classic_pivo,avg_kplex_length_classic_pivo,runtime_conn,num_calls_conn,num_kplexes_conn,max_kplex_size_conn,avg_kplex_size_conn,max_kplex_length_conn,avg_kplex_length_conn,runtime_conn_pivo,num_calls_conn_pivo,num_kplexes_conn_pivo,max_kplex_size_conn_pivo,avg_kplex_size_conn_pivo,max_kplex_length_conn_pivo,avg_kplex_length_conn_pivo
\midrule $\boldsymbol{k=1}$,,,,,,,,,,,,,,,,,,,,,,,,,,,,,
london,1,0.987957954406738,602,28500,3,2.036,1439,23.7933333333333,1.154620885849,602,28500,3,2.036,1439,23.7933333333333,0.850924015045166,122,1296,3,3,526,6.3070987654321,1.05120515823364,122,1296,3,3,526,6.3070987654321
paris,1,1.5463600159,671,50702,3,2.0211825963,1439,14.5620093882,1.8481550217,671,50702,3,2.0211825963,1439,14.5620093882,1.3275740146637,120,1376,3,3,378,3.14752906976744,1.50421595573425,120,1376,3,3,378,3.14752906976744
ny,1,3.3815159798,991,115718,3,2.0320952661,1439,8.2934893448,3.9642798901,991,115718,3,2.0320952661,1439,8.2934893448,2.62375617027283,221,4131,3,3,328,2.51246671508109,3.19930386543274,221,4131,3,3,328,2.51246671508109
highschool-2011,1,1.0881419182,4619,7389,7,2.4291514413,272330,5194.6293138449,1.3569729328,4196,7389,7,2.4291514413,272330,5194.6293138449,1.01570081710815,3672,2373,7,3.38938053097345,1276,423.316477033291,1.26203298568726,3296,2373,7,3.38938053097345,1276,423.316477033291
highschool-2012,1,1.2032730579,4108,9501,6,2.2039785286,729500,14784.7207662351,1.5418128967,3897,9501,6,2.2039785286,729500,14784.7207662351,1.150141954422,2607,1799,6,3.17732073374097,2008,765.587548638132,1.43614888191223,2440,1799,6,3.17732073374097,2008,765.587548638132
highschool-2013,1,6.4543859959,10772,57121,6,2.1605539119,363560,2227.5051207087,8.4801259041,10398,57121,6,2.1605539119,363560,2227.5051207087,6.12934994697571,6949,8442,6,3.12508884150675,1476,97.1210613598673,7.88178300857544,6667,8442,6,3.12508884150675,1476,97.1210613598673
primaryschool,1,14.3173310756683,17434,80648,5,2.20046374367622,116900,406.826232516615,19.447881937027,16939,80648,5,2.20046374367622,116900,406.826232516615,13.7847089767456,12272,14988,5,3.09480918067788,726,39.6854149986656,18.947114944458,11865,14988,5,3.09480918067788,726,39.6854149986656
hospital-ward,1,1.38535308837891,3739,8693,7,2.51880823651214,347500,3625.42919590475,1.72427105903626,3545,8693,7,2.51880823651214,347500,3625.42919590475,1.55375409126282,3075,3504,7,3.30850456621005,2034,503.69549086758,1.65914297103882,2897,3504,7,3.30850456621005,2034,503.69549086758
infectious,1,754.379296064377,170587,134787,9,2.99591948778443,6946340,566287.237062922,858.196667909622,149545,134787,9,2.99591948778443,6946340,566287.237062922,705.679383039475,143961,82125,9,3.76814611872146,5814,758.240401826484,821.392778873444,126520,82125,9,3.76814611872146,5814,758.240401826484
hypertext,1,1.29437208175659,3523,6344,7,2.1968789407314,212340,4362.49527112232,1.5932080745697,3273,6344,7,2.1968789407314,212340,4362.49527112232,1.19041299819946,1841,1166,7,3.16809605488851,2668,485.295025728988,1.47540092468262,1631,1166,7,3.16809605488851,2668,485.295025728988
dnc,1,25.273638010025,20442,24519,9,2.82715445164974,84869755,6563498.36918308,31.4676148891449,17977,24519,9,2.82715445164974,84869755,6563498.36918308,25.2900409698486,15913,11715,9,3.89048228766539,472458,102623.592744345,30.3505928516388,13557,11715,9,3.89048228766539,472458,102623.592744345
karlsruhe,1,58.8537931442261,25826,297293,7,2.07441480290488,123837267,792825.543497492,81.2019050121307,25603,297293,7,2.07441480290488,123837267,792825.543497492,57.986093044281,10419,21311,7,3.12585049974192,25139,10976.3512270658,78.2543909549713,10205,21311,7,3.12585049974192,25139,10976.3512270658
facebook-like,1,27.2274379730225,16108,33876,4,1.95563230605739,16736181,952040.42265911,32.6472518444061,16102,33876,4,1.95563230605739,16736181,952040.42265911,25.6492819786072,1765,393,4,3.00763358778626,17760,11176.9083969466,30.4217200279236,1765,393,4,3.00763358778626,17760,11176.9083969466
as-733,1,3130.89629292488,205821,4245340,12,2.51045452189931,67737600,121223.81029364,, , , , , , ,3318.01907801628,186918,1290050,12,3.68569590325956,296,295.693227394287,, , , , , , 
\midrule $\boldsymbol{k=2}$,,,,,,,,,,,,,,,,,,,,,,,,,,,,,
london,2,117.675229072571,37141,81164,4,2.56393228524962,1439,649.297582672121,116.51580286026,37141,81164,4,2.56393228524962,1439,649.297582672121,2.2125129699707,341,0,0,0,0,0,2.39055991172791,341,0,0,0,0,0
paris,2,200.9043660164,46428,144413,4,2.697859611,1439,456.2263785116,205.2696530819,46428,144413,4,2.697859611,1439,456.2263785116,3.9685230255127,390,0,0,0,0,0,4.17670512199402,390,0,0,0,0,0
ny,2,583.3912379742,88327,317195,4,2.7405255442,1439,395.2820000315,594.3624989986,88327,317195,4,2.7405255442,1439,395.2820000315,9.65313196182251,579,0,0,0,0,0,9.96362400054932,579,0,0,0,0,0
highschool-2011,2,17.0829899311,25953,27072,8,3.0345744681,272330,79515.4337322695,17.7492539883,25301,27072,8,3.0345744681,272330,79515.4337322695,6.9481189250946,6785,2109,8,5.26315789473684,756,338.105263157895,7.07563495635986,6531,2109,8,5.26315789473684,756,338.105263157895
highschool-2012,2,34.7958538532,27829,32399,7,2.5702336492,729500,363100.818543782,35.8042051792,27661,32399,7,2.5702336492,729500,363100.818543782,7.48498201370239,2179,185,7,5.0972972972973,1588,641.513513513514,7.52777695655823,2151,185,7,5.0972972972973,1588,641.513513513514
highschool-2013,2,269.6017990112,79651,111760,6,2.5752952756,363560,173440.526342162,270.1872179508,79481,111760,6,2.5752952756,363560,173440.526342162,42.3021740913391,4338,359,6,5.01114206128134,196,71.041782729805,42.0250210762024,4343,359,6,5.01114206128134,196,71.041782729805
primaryschool,2,240.670147895813,79053,142281,6,2.86926574876477,116900,23989.911808323,264.938748121262,78769,142281,6,2.86926574876477,116900,23989.911808323,97.8004808425903,6909,435,6,5.01149425287356,86,29.5402298850575,95.1880059242249,6914,435,6,5.01149425287356,86,29.5402298850575
hospital-ward,2,10.3721370697021,13603,23999,7,3.14146422767615,347500,40602.9737905746,10.4716579914093,13447,23999,7,3.14146422767615,347500,40602.9737905746,7.46811699867249,3050,960,7,5.171875,814,391,7.28465604782105,3009,960,7,5.171875,814,391
hypertext,2,13.9916679859161,17426,18910,9,2.73569539925965,212340,71349.965203596,14.1665599346161,17022,18910,9,2.73569539925965,212340,71349.965203596,6.8086941242218,2045,104,9,5.43269230769231,568,433.769230769231,7.03512001037598,1861,104,9,5.43269230769231,568,433.769230769231
\midrule $\boldsymbol{k=3}$,,,,,,,,,,,,,,,,,,,,,,,,,,,,,
highschool-2011,3,765.8150551319,534032,577447,9,3.5060083436,272330,153688.322678964,783.3533349037,532758,577447,9,3.5060083436,272330,153688.322678964,128.569464921951,23421,2494,9,7.14194065757819,496,300.121892542101,130.523743867874,23305,2494,9,7.14194065757819,496,300.121892542101
highschool-2012,3,2428.4441468716,1215394,1293685,7,3.2645350298,729500,539186.561504539,2427.5814189911,1215271,1293685,7,3.2645350298,729500,539186.561504539,167.057809114456,13082,8,7,7,648,550.5,162.758733034134,13082,8,7,7,648,550.5
hospital-ward,3,244.910546064377,121892,169876,8,3.67616967670536,347500,138411.276248558,248.785387039185,121816,169876,8,3.67616967670536,347500,138411.276248558,143.175406932831,6368,127,8,7.00787401574803,634,345.811023622047,140.798375844955,6373,127,8,7.00787401574803,634,345.811023622047
hypertext,3,603.744918107986,372289,397058,9,3.41723627278634,212340,125389.802668628,602.712821006775,371931,397058,9,3.41723627278634,212340,125389.802668628,229.61323094368,13437,48,9,7.27083333333333,548,487.166666666667,230.988229990005,13403,48,9,7.27083333333333,548,487.166666666667
\end{filecontents*}
\pgfplotstabletypeset[
   col sep=comma,
   columns={name,num_kplexes_classic,max_kplex_size_classic, runtime_classic,num_calls_classic,runtime_classic_pivo,num_calls_classic_pivo},
   columns/name/.style={column type=l,string type,column name={Instance}},
   columns/num_kplexes_classic/.style={column name={$\# R$},column type=r,precision=0},
   columns/max_kplex_size_classic/.style={column name={$|C|_{\max}$},column type=r,precision=1},
   columns/runtime_classic/.style={column name={$t$},column type=r,precision=0, relative*={0}},
   columns/num_calls_classic/.style={column name={$c$},column type=r,int detect},
   columns/runtime_classic_pivo/.style={column name={$t_{\text{pivo}}$},column type=r,precision=1, relative*={0}},
   columns/num_calls_classic_pivo/.style={column name={$c_{\text{pivo}}$},column type=r,int detect},  
   every head row/.style={ before row=\toprule & & \multicolumn{5}{c}{Classical setting for $\Delta \sim 5^3$ } \\ \cmidrule(r){2-7} ,after row=\midrule},
   every last row/.style={after row=\bottomrule},
   after row={},
   ]{results_d3_arxiv.csv}
   \label{data_delta31}
 \end{table*}
\begin{table*}[t]%
\footnotesize
  \caption{$\Delta$-$k$-plex statistics and running times for~$\Delta \sim 5^5$ in the classic setting: $\# R$~denotes the number of maximal $\Delta$-$k$-plexes, $|C|_{\max}$ 	denotes the maximum size of a $\Delta$-$k$-plex with respect to the number of vertices, $t$ and~$t_{\text{pivo}}$ denote the running times (in seconds) of \BKB~  without and with pivoting, respectively. The variables $c$ and $c_{\text{pivo}}$ denote the number of recursive calls of \BKB{} without and with pivoting, respectively. Empty cells represent that the time limit of one hour was exceeded. The instances were computed for $k=1,2,3$. 
  	Non-listed data sets signal that the instance could not be solved in the classic setting within the time limit of one hour. }
  	\centering
  \begin{filecontents*}{results_d5_arxiv.csv}
name, k ,runtime_classic,num_calls_classic,num_kplexes_classic,max_kplex_size_classic,avg_kplex_size_classic,max_kplex_length_classic,avg_kplex_length_classic,runtime_classic_pivo,num_calls_classic_pivo,num_kplexes_classic_pivo,max_kplex_size_classic_pivo,avg_kplex_size_classic_pivo,max_kplex_length_classic_pivo,avg_kplex_length_classic_pivo,runtime_conn,num_calls_conn,num_kplexes_conn,max_kplex_size_conn,avg_kplex_size_conn,max_kplex_length_conn,avg_kplex_length_conn,runtime_conn_pivo,num_calls_conn_pivo,num_kplexes_conn_pivo,max_kplex_size_conn_pivo,avg_kplex_size_conn_pivo,max_kplex_length_conn_pivo,avg_kplex_length_conn_pivo
\midrule $\boldsymbol{k=1}$,,,,,,,,,,,,,,,,,,,,,,,,,,,,,
london,1,0.666768074035645,602,10631,3,2.0467500705484,1439,68.8998212773963,0.782767057418823,602,10631,3,2.0467500705484,1439,68.8998212773963,0.629827976226807,122,767,3,3,554,14.6166883963494,0.71788215637207,122,767,3,3,554,14.6166883963494
paris,1,0.8793518543,671,13087,3,2.0046611141,1439,64.1807136853,1.0103099346,671,13087,3,2.0046611141,1439,64.1807136853,0.762784004211426,120,363,3,3,540,21.3911845730028,0.871726989746094,120,363,3,3,540,21.3911845730028
ny,1,2.2506468296,992,58540,3,2.03559959,1439,20.3413905022,2.9403460026,991,58540,3,2.03559959,1439,20.3413905022,1.91850996017456,223,2501,3,3,522,8.32467013194722,2.1763379573822,223,2501,3,3,522,8.32467013194722
highschool-2011,1,2.4610850811,15200,6978,10,3.3059615936,272330,16094.0725136142,2.5045781136,10903,6978,10,3.3059615936,272330,16094.0725136142,2.33603096008301,14513,4339,10,4.12929246370131,21738,10611.3422447569,2.32069301605225,10276,4339,10,4.12929246370131,21738,10611.3422447569
highschool-2012,1,1.2479259968,6074,7356,7,2.5730016313,729500,37605.3463839043,1.5333800316,5586,7356,7,2.5730016313,729500,37605.3463839043,1.17177200317383,4909,3345,7,3.31390134529148,45402,17599.1662182362,1.39259505271912,4499,3345,7,3.31390134529148,45402,17599.1662182362
highschool-2013,1,5.8596410751,23573,29483,8,2.7114608418,363560,6667.9152392904,7.3013279438,20751,29483,8,2.7114608418,363560,6667.9152392904,5.56113505363464,20880,14422,8,3.4771182914991,9930,2208.65829981972,6.65949988365173,18193,14422,8,3.4771182914991,9930,2208.65829981972
primaryschool,1,22.778491973877,84409,83157,9,3.24821722765372,116900,1428.09480861503,27.3395898342133,74584,83157,9,3.24821722765372,116900,1428.09480861503,22.3514239788055,81993,60021,9,3.73339331234068,5580,950.111127771947,26.8596129417419,72305,60021,9,3.73339331234068,5580,950.111127771947
hospital-ward,1,8.04682803153992,39545,10560,12,4.55956439393939,347500,14686.9875,8.47243618965149,29852,10560,12,4.55956439393939,347500,14686.9875,7.38489890098572,39163,9180,12,4.95250544662309,39356,11593.7867102397,8.02195000648499,29505,9180,12,4.95250544662309,39356,11593.7867102397
infectious,1,1145.8453540802,1815707,162220,16,4.55570213290593,6946340,489008.698631488,1067.01967716217,833076,162220,16,4.55570213290593,6946340,489008.698631488,1048.11416101456,1793901,125093,16,5.40193296187636,29516,20327.6443765838,1014.92827510834,815927,125093,16,5.40193296187636,29516,20327.6443765838
hypertext,1,1.74976587295532,7783,7313,7,2.81197866812526,212340,14867.3179269794,2.01671600341797,6996,7313,7,2.81197866812526,212340,14867.3179269794,1.59912300109863,6727,4243,7,3.42611359886873,25108,10307.4060806033,1.87058997154236,5990,4243,7,3.42611359886873,25108,10307.4060806033
dnc,1,106.262460947037,370853,12811,16,4.49270158457576,84869755,14514633.580907,48.0642821788788,80337,12811,16,4.49270158457576,84869755,14514633.580907,105.007340908051,366891,8304,16,6.07032755298651,4276955,2417924.18340559,49.2588050365448,76746,8304,16,6.07032755298651,4276955,2417924.18340559
karlsruhe,1,58.0550000667572,46653,235614,9,2.39735754242108,123837267,1326490.97180134,78.1328570842743,44624,235614,9,2.39735754242108,123837267,1326490.97180134,56.7113828659058,34031,73300,9,3.30276944065484,1306977,285928.605457026,73.9032800197601,32068,73300,9,3.30276944065484,1306977,285928.605457026
facebook-like,1,26.8673341274261,18184,23247,5,2.03148793392696,16736181,1709217.58463458,32.1139361858368,18104,23247,5,2.03148793392696,16736181,1709217.58463458,25.8974671363831,5111,2488,5,3.05747588424437,561961,277973.730305466,30.1288900375366,5066,2488,5,3.05747588424437,561961,277973.730305466
as-733,1,3103.52206683159,205821,4245340,12,2.51045452189931,67737600,128343.354023942,, , , , , , ,3193.38904500008,186918,1290050,12,3.68569590325956,7436,7428.2933746754,, , , , , , 
\midrule $\boldsymbol{k=2}$,,,,,,,,,,,,,,,,,,,,,,,,,,,,,
london,2,77.7697401046753,37141,55653,4,2.35509316658581,1439,948.794584299139,81.7953860759735,37141,55653,4,2.35509316658581,1439,948.794584299139,1.58037090301514,341,0,0,0,0,0,1.67599010467529,341,0,0,0,0,0
paris,2,117.2602350712,46428,76169,4,2.4127138337,1439,868.1294489884,116.2135088444,46428,76169,4,2.4127138337,1439,868.1294489884,2.0717499256134,390,0,0,0,0,0,2.24799108505249,390,0,0,0,0,0
ny,2,405.1012890339,88333,216582,4,2.6142431042,1439,581.2829828887,415.607929945,88333,216582,4,2.6142431042,1439,581.2829828887,6.08727502822876,579,0,0,0,0,0,6.64371085166931,579,0,0,0,0,0
highschool-2011,2,31.939437151,103916,48691,11,3.8604670268,272330,52871.072785525,31.232711792,87938,48691,11,3.8604670268,272330,52871.072785525,22.1002559661865,73577,12989,11,6.05535453075679,14618,9722.73385172069,21.7870590686798,59818,12989,11,6.05535453075679,14618,9722.73385172069
highschool-2012,2,33.3653450012,46342,51012,7,2.9306633733,729500,242292.701815259,35.2973759174,45696,51012,7,2.9306633733,729500,242292.701815259,10.4763629436493,7769,1869,7,5.11503477795613,27862,14737.8592830391,10.4806609153748,7633,1869,7,5.11503477795613,27862,14737.8592830391
highschool-2013,2,199.4550311565,177197,184404,10,3.0937181406,363560,106635.361543134,204.653826952,171866,184404,10,3.0937181406,363560,106635.361543134,49.7233240604401,52936,15637,10,5.38453667583296,5870,1825.08089787044,51.507150888443,49762,15637,10,5.38453667583296,5870,1825.08089787044
primaryschool,2,320.340541124344,581238,581212,11,3.81313186926629,116900,6755.73646105036,320.307415008545,561507,581212,11,3.81313186926629,116900,6755.73646105036,234.386074066162,302582,122076,11,5.44562403748485,3960,779.718372161604,234.407918930054,288722,122076,11,5.44562403748485,3960,779.718372161604
hospital-ward,2,55.8627550601959,194253,71854,13,4.91575973501823,347500,24500.520806079,57.7325620651245,173849,71854,13,4.91575973501823,347500,24500.520806079,52.4996798038483,168798,38529,13,6.23613382127748,36636,10251.9730073451,52.6835367679596,150078,38529,13,6.23613382127748,36636,10251.9730073451
hypertext,2,27.4098010063171,70856,66139,9,3.43703412510017,212340,29687.9488955079,27.7218241691589,68791,66139,9,3.43703412510017,212340,29687.9488955079,20.6437408924103,21228,6726,9,5.19134701159679,16548,8637.94617900684,21.3303859233856,20367,6726,9,5.19134701159679,16548,8637.94617900684
\midrule $\boldsymbol{k=3}$,,,,,,,,,,,,,,,,,,,,,,,,,,,,,
highschool-2011,3,836.8415000439,1288366,1059245,13,4.0599917866,272330,90780.1064003134,850.4650108814,1233562,1059245,13,4.0599917866,272330,90780.1064003134,287.100867986679,355247,54717,13,7.7561087047901,13878,9149.67956942084,289.949823856354,321818,54717,13,7.7561087047901,13878,9149.67956942084
highschool-2012,3,2149.5472280979,1722247,1950811,8,3.5414112387,729500,366041.390903065,2168.4767839909,1721351,1950811,8,3.5414112387,729500,366041.390903065,191.913043022156,36195,430,8,7.02558139534884,20042,13332.2395348837,190.394711971283,36204,430,8,7.02558139534884,20042,13332.2395348837
hospital-ward,3,450.933631896973,1023768,619778,14,5.24908434955742,347500,47623.328646709,478.005537986755,976275,619778,14,5.24908434955742,347500,47623.328646709,364.789851903915,547333,128059,14,7.75387907136554,28856,9327.85000663757,373.403419017792,516450,128059,14,7.75387907136554,28856,9327.85000663757
hypertext,3,805.626487016678,1211479,1315313,9,4.06902159409966,212340,46090.4464549503,842.882575035095,1205841,1315313,9,4.06902159409966,212340,46090.4464549503,488.321811914444,97808,5891,9,7.1120353080971,12628,8034.61245968426,481.175151109695,97711,5891,9,7.1120353080971,12628,8034.61245968426
\end{filecontents*}
\pgfplotstabletypeset[
   col sep=comma,
   columns={name,num_kplexes_classic, max_kplex_size_classic,runtime_classic,num_calls_classic,runtime_classic_pivo,num_calls_classic_pivo},
   columns/name/.style={column type=l,string type,column name={Instance}},
   columns/num_kplexes_classic/.style={column name={$\# R$},column type=r,precision=0},
   columns/max_kplex_size_classic/.style={column name={$|C|_{\max}$},column type=r,precision=1},
   columns/runtime_classic/.style={column name={$t$},column type=r,precision=0, relative*={0}},
   columns/num_calls_classic/.style={column name={$c$},column type=r,int detect},
   columns/runtime_classic_pivo/.style={column name={$t_{\text{pivo}}$},column type=r,precision=1, relative*={0}},
   columns/num_calls_classic_pivo/.style={column name={$c_{\text{pivo}}$},column type=r,int detect},  
   every head row/.style={ before row=\toprule & & \multicolumn{5}{c}{Classical setting for $\Delta \sim 5^5$ } \\ \cmidrule(r){2-7} ,after row=\midrule},
   every last row/.style={after row=\bottomrule},
   after row={},
   ]{results_d5_arxiv.csv}
   \label{data_delta51}
 \end{table*}
 }
\toappendix{ 
\begin{table*}[!hp]%
\footnotesize
  \caption{$\Delta$-$k$-plex statistics and running times for~$\Delta =0$ with the connectedness criterion: $ \# R$ denotes the number of maximal $\Delta$-$k$-plexes, $|C|_{\max}$ denotes the maximum size of a $\Delta$-$k$-plex with respect to the number of vertices, $t$ and~$t_{\text{pivo}}$ denote the running times (in seconds) of \BKB~  without and with pivoting, respectively. The variables $c$ and $c_{\text{pivo}}$ denote the number of recursive calls of \BKB{} without and with pivoting, respectively. Empty cells represent that the time limit of one hour was exceeded. The instances were computed for $k=1,2,3$. 
  	Non-listed data sets signal that the instance could  not be solved with the connectedness criterion within the time limit of one hour. }  
  	\centering
  \begin{filecontents*}{results_d0.csv}
name, k ,runtime_classic,num_calls_classic,num_kplexes_classic,max_kplex_size_classic,avg_kplex_size_classic,max_kplex_length_classic,avg_kplex_length_classic,runtime_classic_pivo,num_calls_classic_pivo,num_kplexes_classic_pivo,max_kplex_size_classic_pivo,avg_kplex_size_classic_pivo,max_kplex_length_classic_pivo,avg_kplex_length_classic_pivo,runtime_conn,num_calls_conn,num_kplexes_conn,max_kplex_size_conn,avg_kplex_size_conn,max_kplex_length_conn,avg_kplex_length_conn,runtime_conn_pivo,num_calls_conn_pivo,num_kplexes_conn_pivo,max_kplex_size_conn_pivo,avg_kplex_size_conn_pivo,max_kplex_length_conn_pivo,avg_kplex_length_conn_pivo,,,,,,,
\midrule $\boldsymbol{k=1}$,,,,,,,,,,,,,,,,,,,,,,,,,,,,,,,,,,,,
london,1,0.987957954406738,602,28500,3,2.036,1439,23.7933333333333,1.154620885849,602,28500,3,2.036,1439,23.7933333333333,0.850924015045166,122,1296,3,3,526,6.3070987654321,1.05120515823364,122,1296,3,3,526,6.3070987654321,,,,,,,
paris,1,1.5463600159,671,50702,3,2.0211825963,1439,14.5620093882,1.8481550217,671,50702,3,2.0211825963,1439,14.5620093882,1.3275740146637,120,1376,3,3,378,3.14752906976744,1.50421595573425,120,1376,3,3,378,3.14752906976744,,,,,,,
ny,1,3.3815159798,991,115718,3,2.0320952661,1439,8.2934893448,3.9642798901,991,115718,3,2.0320952661,1439,8.2934893448,2.62375617027283,221,4131,3,3,328,2.51246671508109,3.19930386543274,221,4131,3,3,328,2.51246671508109,,,,,,,
highschool-2011,1,2.0545630455,2443,26510,5,2.0386269332,272330,1294.3636363636,2.6926810741,2382,26510,5,2.0386269332,272330,1294.3636363636,1.84502696990967,1084,1069,5,3.07577174929841,0,0,2.4710328578949,1050,1069,5,3.07577174929841,0,0,,,,,,,
highschool-2012,1,2.8359839916,3012,42285,5,2.0319971621,729500,3105.3565094005,4.0445418358,2934,42285,5,2.0319971621,729500,3105.3565094005,2.71820998191833,1151,1434,5,3.06903765690377,0,0,3.6588990688324,1107,1434,5,3.06903765690377,0,0,,,,,,,
highschool-2013,1,14.5932919979,8761,172362,5,2.0482763022,363560,689.7350924218,20.7033081055,8525,172362,5,2.0482763022,363560,689.7350924218,14.1341631412506,4304,8058,5,3.07321916108215,0,0,19.8121280670166,4162,8058,5,3.07321916108215,0,0,,,,,,,
primaryschool,1,18.6299231052399,14089,107121,5,2.09333370674284,116900,264.092008102986,34.8048100471497,13748,107121,5,2.09333370674284,116900,264.092008102986,17.9004309177399,8177,9745,5,3.05079527963058,0,0,25.5673608779907,7924,9745,5,3.05079527963058,0,0,,,,,,,
hospital-ward,1,2.97551798820496,1988,27910,5,2.08448584736653,347500,933.805087782157,4.05345797538757,1947,27910,5,2.08448584736653,347500,933.805087782157,3.09471678733826,1133,2348,5,3.03620102214651,0,0,4.00938606262207,1105,2348,5,3.03620102214651,0,0,,,,,,,
infectious,1,780.119927883148,73219,349787,5,2.08347079794275,6946340,217890.437551996,859.857194900513,69012,349787,5,2.08347079794275,6946340,217890.437551996,716.520969152451,35310,36965,5,3.08667658595969,0,0,827.773962020874,33307,36965,5,3.08667658595969,0,0,,,,,,,
hypertext,1,2.20874309539795,2788,19150,6,2.04506527415144,212340,1252.97232375979,3.00041198730469,2655,19150,6,2.04506527415144,212340,1252.97232375979,2.10710787773132,840,917,6,3.06434023991276,0,0,2.76368808746338,736,917,6,3.06434023991276,0,0,,,,,,,
dnc,1,23.764121055603,6251,33568,2,1.94441134413727,84869755,4717795.60453408,24.2248330116272,6251,33568,2,1.94441134413727,84869755,4717795.60453408,29.6903610229492,6251,33568,2,1.94441134413727,84869755,4717795.60453408,23.3867359161377,362,0,0,0,0,0,29.30814909935,362,0,0,0,0,0
karlsruhe,1,69.6526448726654,18983,459230,3,1.99595191951745,123837267,504269.520719465,103.009526968002,18983,459230,3,1.99595191951745,123837267,504269.520719465,66.8472969532013,1007,11,3,3,0,0,95.996680021286,1007,11,3,3,0,0,,,,,,,
facebook-like,1,28.7566769123077,15738,61648,2,1.96919608097586,16736181,515539.965043473,34.2384910583496,15738,61648,2,1.96919608097586,16736181,515539.965043473,26.8328869342804,1081,0,0,0,0,0,32.6977460384369,1081,0,0,0,0,0,,,,,,,
as-733,1,3099.08181500435,205821,4245340,12,2.51045452189931,67737600,120928.658340675,, , , , , , ,3250.77057886124,186918,1290050,12,3.68569590325956,0,0,, , , , , , ,,,,,,,
\midrule $\boldsymbol{k=2}$,,,,,,,,,,,,,,,,,,,,,,,,,,,,,,,,,,,,
london,2,117.675229072571,37141,81164,4,2.56393228524962,1439,649.297582672121,116.51580286026,37141,81164,4,2.56393228524962,1439,649.297582672121,2.2125129699707,341,0,0,0,0,0,2.39055991172791,341,0,0,0,0,0,,,,,,,
paris,2,200.9043660164,46428,144413,4,2.697859611,1439,456.2263785116,205.2696530819,46428,144413,4,2.697859611,1439,456.2263785116,3.9685230255127,390,0,0,0,0,0,4.17670512199402,390,0,0,0,0,0,,,,,,,
ny,2,583.3912379742,88327,317195,4,2.7405255442,1439,395.2820000315,594.3624989986,88327,317195,4,2.7405255442,1439,395.2820000315,9.65313196182251,579,0,0,0,0,0,9.96362400054932,579,0,0,0,0,0,,,,,,,
highschool-2011,2,38.8545958996,11400,14967,5,2.5027059531,272330,143288.484666266,38.6924951077,11377,14967,5,2.5027059531,272330,143288.484666266,11.291218996048,632,15,5,5,0,0,11.2148518562317,632,15,5,5,0,0,,,,,,,
highschool-2012,2,87.2681910992,19016,24107,5,2.3472020575,729500,487503.422242502,87.971394062,18990,24107,5,2.3472020575,729500,487503.422242502,17.251855134964,545,15,5,5,0,0,18.2304141521454,545,15,5,5,0,0,,,,,,,
highschool-2013,2,631.3113069534,66612,109757,5,2.5360387037,363560,176554.675874887,621.9873850346,66529,109757,5,2.5360387037,363560,176554.675874887,96.0247809886932,2194,81,5,5,0,0,98.1391940116882,2195,81,5,5,0,0,,,,,,,
primaryschool,2,340.781122922897,58346,96690,5,2.73613610507808,116900,35256.188850967,331.757730007172,58166,96690,5,2.73613610507808,116900,35256.188850967,130.170913934708,3863,77,5,5,0,0,127.51246213913,3866,77,5,5,0,0,,,,,,,
hospital-ward,2,25.5156719684601,5358,13211,5,2.83195821663765,347500,72993.1496480206,26.3894579410553,5345,13211,5,2.83195821663765,347500,72993.1496480206,16.421511888504,416,13,5,5,0,0,17.1749410629272,416,13,5,5,0,0,,,,,,,
hypertext,2,27.6703398227692,9239,11565,8,2.47747514051016,212340,116185.691309987,29.1255309581757,9048,11565,8,2.47747514051016,212340,116185.691309987,13.0203990936279,735,35,8,5.48571428571429,0,0,13.6767330169678,638,35,8,5.48571428571429,0,0,,,,,,,
dnc,2,, , , , , , ,, , , , , , ,120.081828832626,3117,0,0,0,0,0,124.783806800842,3117,0,0,0,0,0,,,,,,,
karlsruhe,2,, , , , , , ,, , , , , , ,634.00106716156,3988,0,0,0,0,0,633.318753957748,3988,0,0,0,0,0,,,,,,,
facebook-like,2,, , , , , , ,, , , , , , ,218.83242893219,1993,0,0,0,0,0,225.569973945618,1993,0,0,0,0,0,,,,,,,
\midrule $\boldsymbol{k=3}$,,,,,,,,,,,,,,,,,,,,,,,,,,,,,,,,,,,,
london,3,, , , , , , ,, , , , , , ,5.48647093772888,693,0,0,0,0,0,5.68560004234314,693,0,0,0,0,0,,,,,,,
paris,3,, , , , , , ,, , , , , , ,10.8641550540924,799,0,0,0,0,0,11.2240259647369,799,0,0,0,0,0,,,,,,,
ny,3,, , , , , , ,, , , , , , ,36.624125957489,1318,0,0,0,0,0,37.141979932785,1318,0,0,0,0,0,,,,,,,
highschool-2011,3,2279.4216780663,375004,427686,6,3.2395940012,272330,207262.840027497,2381.370513916,375001,427686,6,3.2395940012,272330,207262.840027497,329.255962133408,3889,0,0,0,0,0,326.998162984848,3889,0,0,0,0,0,,,,,,,
highschool-2012,3,, , , , , , ,, , , , , , ,504.073801994324,3569,0,0,0,0,0,506.083670854568,3569,0,0,0,0,0,,,,,,,
hospital-ward,3,901.527784824371,85685,117253,6,3.42573750778232,347500,200122.278321237,941.435976028442,85682,117253,6,3.42573750778232,347500,200122.278321237,541.424566030502,2470,0,0,0,0,0,527.99120092392,2470,0,0,0,0,0,,,,,,,
hypertext,3,1411.93648910522,268821,293253,9,3.20236110116521,212340,169534.286912666,1501.7961139679,268626,293253,9,3.20236110116521,212340,169534.286912666,557.8901450634,3656,7,9,7.28571428571429,0,0,565.334822177887,3626,7,9,7.28571428571429,0,0,,,,,,,
\end{filecontents*}
\pgfplotstabletypeset[
   col sep=comma,
   columns={name,num_kplexes_conn,max_kplex_size_conn,runtime_conn,num_calls_conn,runtime_conn_pivo,num_calls_conn_pivo},
   columns/name/.style={column type=l,string type,column name={Instance}},
   columns/num_kplexes_conn/.style={column name={$ \# R$},column type=r,precision=1},
   columns/max_kplex_size_conn/.style={column name={$|C|_{\max}$},column type=r,precision=1},
   columns/runtime_conn/.style={column name={$t$},column type=r,precision=0, relative*={0}},
   columns/num_calls_conn/.style={column name={$c$},column type=r,precision=0},
   columns/runtime_conn_pivo/.style={column name={$t_{\text{pivo}}$},precision=0, relative*={0}},
   columns/num_calls_conn_pivo/.style={column name={$c_{\text{pivo}}$},column type=r,precision=0},   
   every head row/.style={ before row=\toprule & \multicolumn{6}{c}{Connectedness Criterion for $\Delta = 0 $ } \\ \cmidrule(r){2-7} ,after row=\midrule},
   every last row/.style={after row=\bottomrule},
   after row={},
   ]{results_d0.csv}
   \label{data_delta02}
 \end{table*}%
\begin{table*}[!hp]%
\footnotesize
  \caption{$\Delta$-$k$-plex statistics and running times for~$\Delta \sim 5^3$ with the connectedness criterion: $\# R$ denotes the number of maximal $\Delta$-$k$-plexes, $|C|_{\max}$ 	denotes the maximum size of a $\Delta$-$k$-plex with respect to the number of vertices, $t$ and~$t_{\text{pivo}}$ denote the running times (in seconds) of \BKB~  without and with pivoting, respectively. The variables $c$ and $c_{\text{pivo}}$ denote the number of recursive calls of \BKB{} without and with pivoting, respectively. Empty cells represent that the time limit of one hour was exceeded. The instances were computed for $k=1,2,3$. 
  	Non-listed data sets signal that the instance could not be solved with the connectedness criterion within the time limit of one hour. }  
  	\centering
  \begin{filecontents*}{results_d3.csv}
name, k ,runtime_classic,num_calls_classic,num_kplexes_classic,max_kplex_size_classic,avg_kplex_size_classic,max_kplex_length_classic,avg_kplex_length_classic,runtime_classic_pivo,num_calls_classic_pivo,num_kplexes_classic_pivo,max_kplex_size_classic_pivo,avg_kplex_size_classic_pivo,max_kplex_length_classic_pivo,avg_kplex_length_classic_pivo,runtime_conn,num_calls_conn,num_kplexes_conn,max_kplex_size_conn,avg_kplex_size_conn,max_kplex_length_conn,avg_kplex_length_conn,runtime_conn_pivo,num_calls_conn_pivo,num_kplexes_conn_pivo,max_kplex_size_conn_pivo,avg_kplex_size_conn_pivo,max_kplex_length_conn_pivo,avg_kplex_length_conn_pivo
\midrule $\boldsymbol{k=1}$,,,,,,,,,,,,,,,,,,,,,,,,,,,,,
london,1,0.987957954406738,602,28500,3,2.036,1439,23.7933333333333,1.154620885849,602,28500,3,2.036,1439,23.7933333333333,0.850924015045166,122,1296,3,3,526,6.3070987654321,1.05120515823364,122,1296,3,3,526,6.3070987654321
paris,1,1.5463600159,671,50702,3,2.0211825963,1439,14.5620093882,1.8481550217,671,50702,3,2.0211825963,1439,14.5620093882,1.3275740146637,120,1376,3,3,378,3.14752906976744,1.50421595573425,120,1376,3,3,378,3.14752906976744
ny,1,3.3815159798,991,115718,3,2.0320952661,1439,8.2934893448,3.9642798901,991,115718,3,2.0320952661,1439,8.2934893448,2.62375617027283,221,4131,3,3,328,2.51246671508109,3.19930386543274,221,4131,3,3,328,2.51246671508109
highschool-2011,1,1.0881419182,4619,7389,7,2.4291514413,272330,5194.6293138449,1.3569729328,4196,7389,7,2.4291514413,272330,5194.6293138449,1.01570081710815,3672,2373,7,3.38938053097345,1276,423.316477033291,1.26203298568726,3296,2373,7,3.38938053097345,1276,423.316477033291
highschool-2012,1,1.2032730579,4108,9501,6,2.2039785286,729500,14784.7207662351,1.5418128967,3897,9501,6,2.2039785286,729500,14784.7207662351,1.150141954422,2607,1799,6,3.17732073374097,2008,765.587548638132,1.43614888191223,2440,1799,6,3.17732073374097,2008,765.587548638132
highschool-2013,1,6.4543859959,10772,57121,6,2.1605539119,363560,2227.5051207087,8.4801259041,10398,57121,6,2.1605539119,363560,2227.5051207087,6.12934994697571,6949,8442,6,3.12508884150675,1476,97.1210613598673,7.88178300857544,6667,8442,6,3.12508884150675,1476,97.1210613598673
primaryschool,1,14.3173310756683,17434,80648,5,2.20046374367622,116900,406.826232516615,19.447881937027,16939,80648,5,2.20046374367622,116900,406.826232516615,13.7847089767456,12272,14988,5,3.09480918067788,726,39.6854149986656,18.947114944458,11865,14988,5,3.09480918067788,726,39.6854149986656
hospital-ward,1,1.38535308837891,3739,8693,7,2.51880823651214,347500,3625.42919590475,1.72427105903626,3545,8693,7,2.51880823651214,347500,3625.42919590475,1.55375409126282,3075,3504,7,3.30850456621005,2034,503.69549086758,1.65914297103882,2897,3504,7,3.30850456621005,2034,503.69549086758
infectious,1,754.379296064377,170587,134787,9,2.99591948778443,6946340,566287.237062922,858.196667909622,149545,134787,9,2.99591948778443,6946340,566287.237062922,705.679383039475,143961,82125,9,3.76814611872146,5814,758.240401826484,821.392778873444,126520,82125,9,3.76814611872146,5814,758.240401826484
hypertext,1,1.29437208175659,3523,6344,7,2.1968789407314,212340,4362.49527112232,1.5932080745697,3273,6344,7,2.1968789407314,212340,4362.49527112232,1.19041299819946,1841,1166,7,3.16809605488851,2668,485.295025728988,1.47540092468262,1631,1166,7,3.16809605488851,2668,485.295025728988
dnc,1,25.273638010025,20442,24519,9,2.82715445164974,84869755,6563498.36918308,31.4676148891449,17977,24519,9,2.82715445164974,84869755,6563498.36918308,25.2900409698486,15913,11715,9,3.89048228766539,472458,102623.592744345,30.3505928516388,13557,11715,9,3.89048228766539,472458,102623.592744345
karlsruhe,1,58.8537931442261,25826,297293,7,2.07441480290488,123837267,792825.543497492,81.2019050121307,25603,297293,7,2.07441480290488,123837267,792825.543497492,57.986093044281,10419,21311,7,3.12585049974192,25139,10976.3512270658,78.2543909549713,10205,21311,7,3.12585049974192,25139,10976.3512270658
facebook-like,1,27.2274379730225,16108,33876,4,1.95563230605739,16736181,952040.42265911,32.6472518444061,16102,33876,4,1.95563230605739,16736181,952040.42265911,25.6492819786072,1765,393,4,3.00763358778626,17760,11176.9083969466,30.4217200279236,1765,393,4,3.00763358778626,17760,11176.9083969466
as-733,1,3130.89629292488,205821,4245340,12,2.51045452189931,67737600,121223.81029364,, , , , , , ,3318.01907801628,186918,1290050,12,3.68569590325956,296,295.693227394287,, , , , , , 
\midrule $\boldsymbol{k=2}$,,,,,,,,,,,,,,,,,,,,,,,,,,,,,
london,2,117.675229072571,37141,81164,4,2.56393228524962,1439,649.297582672121,116.51580286026,37141,81164,4,2.56393228524962,1439,649.297582672121,2.2125129699707,341,0,0,0,0,0,2.39055991172791,341,0,0,0,0,0
paris,2,200.9043660164,46428,144413,4,2.697859611,1439,456.2263785116,205.2696530819,46428,144413,4,2.697859611,1439,456.2263785116,3.9685230255127,390,0,0,0,0,0,4.17670512199402,390,0,0,0,0,0
ny,2,583.3912379742,88327,317195,4,2.7405255442,1439,395.2820000315,594.3624989986,88327,317195,4,2.7405255442,1439,395.2820000315,9.65313196182251,579,0,0,0,0,0,9.96362400054932,579,0,0,0,0,0
highschool-2011,2,17.0829899311,25953,27072,8,3.0345744681,272330,79515.4337322695,17.7492539883,25301,27072,8,3.0345744681,272330,79515.4337322695,6.9481189250946,6785,2109,8,5.26315789473684,756,338.105263157895,7.07563495635986,6531,2109,8,5.26315789473684,756,338.105263157895
highschool-2012,2,34.7958538532,27829,32399,7,2.5702336492,729500,363100.818543782,35.8042051792,27661,32399,7,2.5702336492,729500,363100.818543782,7.48498201370239,2179,185,7,5.0972972972973,1588,641.513513513514,7.52777695655823,2151,185,7,5.0972972972973,1588,641.513513513514
highschool-2013,2,269.6017990112,79651,111760,6,2.5752952756,363560,173440.526342162,270.1872179508,79481,111760,6,2.5752952756,363560,173440.526342162,42.3021740913391,4338,359,6,5.01114206128134,196,71.041782729805,42.0250210762024,4343,359,6,5.01114206128134,196,71.041782729805
primaryschool,2,240.670147895813,79053,142281,6,2.86926574876477,116900,23989.911808323,264.938748121262,78769,142281,6,2.86926574876477,116900,23989.911808323,97.8004808425903,6909,435,6,5.01149425287356,86,29.5402298850575,95.1880059242249,6914,435,6,5.01149425287356,86,29.5402298850575
hospital-ward,2,10.3721370697021,13603,23999,7,3.14146422767615,347500,40602.9737905746,10.4716579914093,13447,23999,7,3.14146422767615,347500,40602.9737905746,7.46811699867249,3050,960,7,5.171875,814,391,7.28465604782105,3009,960,7,5.171875,814,391
hypertext,2,13.9916679859161,17426,18910,9,2.73569539925965,212340,71349.965203596,14.1665599346161,17022,18910,9,2.73569539925965,212340,71349.965203596,6.8086941242218,2045,104,9,5.43269230769231,568,433.769230769231,7.03512001037598,1861,104,9,5.43269230769231,568,433.769230769231
dnc,2,, , , , , , ,, , , , , , ,197.152273893356,60609,27584,11,5.56104988399072,299794,92586.1546911253,203.944195985794,55439,27584,11,5.56104988399072,299794,92586.1546911253
karlsruhe,2,, , , , , , ,, , , , , , ,559.952960014343,14762,4438,8,5.10432627309599,19749,9739.27264533574,568.992789030075,14670,4438,8,5.10432627309599,19749,9739.27264533574
facebook-like,2,, , , , , , ,, , , , , , ,229.685123920441,8983,2,5,5,10357,9541.5,241.552314996719,8983,2,5,5,10357,9541.5
\midrule $\boldsymbol{k=3}$,,,,,,,,,,,,,,,,,,,,,,,,,,,,,
london,3,, , , , , , ,, , , , , , ,5.48647093772888,693,0,0,0,0,0,5.68560004234314,693,0,0,0,0,0
paris,3,, , , , , , ,, , , , , , ,10.8641550540924,799,0,0,0,0,0,11.2240259647369,799,0,0,0,0,0
ny,3,, , , , , , ,, , , , , , ,36.624125957489,1318,0,0,0,0,0,37.141979932785,1318,0,0,0,0,0
highschool-2011,3,765.8150551319,534032,577447,9,3.5060083436,272330,153688.322678964,783.3533349037,532758,577447,9,3.5060083436,272330,153688.322678964,128.569464921951,23421,2494,9,7.14194065757819,496,300.121892542101,130.523743867874,23305,2494,9,7.14194065757819,496,300.121892542101
highschool-2012,3,2428.4441468716,1215394,1293685,7,3.2645350298,729500,539186.561504539,2427.5814189911,1215271,1293685,7,3.2645350298,729500,539186.561504539,167.057809114456,13082,8,7,7,648,550.5,162.758733034134,13082,8,7,7,648,550.5
highschool-2013,3,, , , , , , ,, , , , , , ,1452.70768594742,41286,0,0,0,0,0,1436.97881913185,41286,0,0,0,0,0
hospital-ward,3,244.910546064377,121892,169876,8,3.67616967670536,347500,138411.276248558,248.785387039185,121816,169876,8,3.67616967670536,347500,138411.276248558,143.175406932831,6368,127,8,7.00787401574803,634,345.811023622047,140.798375844955,6373,127,8,7.00787401574803,634,345.811023622047
hypertext,3,603.744918107986,372289,397058,9,3.41723627278634,212340,125389.802668628,602.712821006775,371931,397058,9,3.41723627278634,212340,125389.802668628,229.61323094368,13437,48,9,7.27083333333333,548,487.166666666667,230.988229990005,13403,48,9,7.27083333333333,548,487.166666666667
\end{filecontents*}
\pgfplotstabletypeset[
   col sep=comma,
   columns={name,num_kplexes_conn,max_kplex_size_conn,runtime_conn,num_calls_conn,runtime_conn_pivo,num_calls_conn_pivo},
   columns/name/.style={column type=l,string type,column name={Instance}},
   columns/num_kplexes_conn/.style={column name={$\# R$},column type=r,precision=1},
   columns/max_kplex_size_conn/.style={column name={$|C|_{\max}$},column type=r,precision=1},
   columns/runtime_conn/.style={column name={$t$},column type=r,precision=0, relative*={0}},
   columns/num_calls_conn/.style={column name={$c$},column type=r,precision=0},
   columns/runtime_conn_pivo/.style={column name={$t_{\text{pivo}}$},precision=0, relative*={0}},
   columns/num_calls_conn_pivo/.style={column name={$c_{\text{pivo}}$},column type=r,precision=0},   
   every head row/.style={ before row=\toprule & \multicolumn{6}{c}{Connectedness Criterion for $\Delta \sim 5^3$ } \\ \cmidrule(r){2-7} ,after row=\midrule},
   every last row/.style={after row=\bottomrule},
   after row={},
   ]{results_d3.csv}
   \label{data_delta32}
 \end{table*}
\begin{table*}[!hp]%
\footnotesize
  \caption{$\Delta$-$k$-plex statistics and running times for~$\Delta \sim 5^5$ with the connectedness criterion: $\# R$ denotes the number of maximal $\Delta$-$k$-plexes, $|C|_{\max}$ denotes the maximum size of a $\Delta$-$k$-plex with respect to the number of vertices, $t$ and~$t_{\text{pivo}}$ denote the running times (in seconds) of \BKB~  without and with pivoting, respectively. The variables $c$ and $c_{\text{pivo}}$ denote the number of recursive calls of \BKB{} without and with pivoting, respectively. Empty cells represent that the time limit of one hour was exceeded. The instances were computed for $k=1,2,3$. 
  	Non-listed data sets signal that the instance could not be solved with the connectedness criterion within the time limit of one hour. }  
  	\centering
  \begin{filecontents*}{results_d5.csv}
name, k ,runtime_classic,num_calls_classic,num_kplexes_classic,max_kplex_size_classic,avg_kplex_size_classic,max_kplex_length_classic,avg_kplex_length_classic,runtime_classic_pivo,num_calls_classic_pivo,num_kplexes_classic_pivo,max_kplex_size_classic_pivo,avg_kplex_size_classic_pivo,max_kplex_length_classic_pivo,avg_kplex_length_classic_pivo,runtime_conn,num_calls_conn,num_kplexes_conn,max_kplex_size_conn,avg_kplex_size_conn,max_kplex_length_conn,avg_kplex_length_conn,runtime_conn_pivo,num_calls_conn_pivo,num_kplexes_conn_pivo,max_kplex_size_conn_pivo,avg_kplex_size_conn_pivo,max_kplex_length_conn_pivo,avg_kplex_length_conn_pivo
\midrule $\boldsymbol{k=1}$,,,,,,,,,,,,,,,,,,,,,,,,,,,,,
london,1,0.666768074035645,602,10631,3,2.0467500705484,1439,68.8998212773963,0.782767057418823,602,10631,3,2.0467500705484,1439,68.8998212773963,0.629827976226807,122,767,3,3,554,14.6166883963494,0.71788215637207,122,767,3,3,554,14.6166883963494
paris,1,0.8793518543,671,13087,3,2.0046611141,1439,64.1807136853,1.0103099346,671,13087,3,2.0046611141,1439,64.1807136853,0.762784004211426,120,363,3,3,540,21.3911845730028,0.871726989746094,120,363,3,3,540,21.3911845730028
ny,1,2.2506468296,992,58540,3,2.03559959,1439,20.3413905022,2.9403460026,991,58540,3,2.03559959,1439,20.3413905022,1.91850996017456,223,2501,3,3,522,8.32467013194722,2.1763379573822,223,2501,3,3,522,8.32467013194722
highschool-2011,1,2.4610850811,15200,6978,10,3.3059615936,272330,16094.0725136142,2.5045781136,10903,6978,10,3.3059615936,272330,16094.0725136142,2.33603096008301,14513,4339,10,4.12929246370131,21738,10611.3422447569,2.32069301605225,10276,4339,10,4.12929246370131,21738,10611.3422447569
highschool-2012,1,1.2479259968,6074,7356,7,2.5730016313,729500,37605.3463839043,1.5333800316,5586,7356,7,2.5730016313,729500,37605.3463839043,1.17177200317383,4909,3345,7,3.31390134529148,45402,17599.1662182362,1.39259505271912,4499,3345,7,3.31390134529148,45402,17599.1662182362
highschool-2013,1,5.8596410751,23573,29483,8,2.7114608418,363560,6667.9152392904,7.3013279438,20751,29483,8,2.7114608418,363560,6667.9152392904,5.56113505363464,20880,14422,8,3.4771182914991,9930,2208.65829981972,6.65949988365173,18193,14422,8,3.4771182914991,9930,2208.65829981972
primaryschool,1,22.778491973877,84409,83157,9,3.24821722765372,116900,1428.09480861503,27.3395898342133,74584,83157,9,3.24821722765372,116900,1428.09480861503,22.3514239788055,81993,60021,9,3.73339331234068,5580,950.111127771947,26.8596129417419,72305,60021,9,3.73339331234068,5580,950.111127771947
hospital-ward,1,8.04682803153992,39545,10560,12,4.55956439393939,347500,14686.9875,8.47243618965149,29852,10560,12,4.55956439393939,347500,14686.9875,7.38489890098572,39163,9180,12,4.95250544662309,39356,11593.7867102397,8.02195000648499,29505,9180,12,4.95250544662309,39356,11593.7867102397
infectious,1,1145.8453540802,1815707,162220,16,4.55570213290593,6946340,489008.698631488,1067.01967716217,833076,162220,16,4.55570213290593,6946340,489008.698631488,1048.11416101456,1793901,125093,16,5.40193296187636,29516,20327.6443765838,1014.92827510834,815927,125093,16,5.40193296187636,29516,20327.6443765838
hypertext,1,1.74976587295532,7783,7313,7,2.81197866812526,212340,14867.3179269794,2.01671600341797,6996,7313,7,2.81197866812526,212340,14867.3179269794,1.59912300109863,6727,4243,7,3.42611359886873,25108,10307.4060806033,1.87058997154236,5990,4243,7,3.42611359886873,25108,10307.4060806033
dnc,1,106.262460947037,370853,12811,16,4.49270158457576,84869755,14514633.580907,48.0642821788788,80337,12811,16,4.49270158457576,84869755,14514633.580907,105.007340908051,366891,8304,16,6.07032755298651,4276955,2417924.18340559,49.2588050365448,76746,8304,16,6.07032755298651,4276955,2417924.18340559
karlsruhe,1,58.0550000667572,46653,235614,9,2.39735754242108,123837267,1326490.97180134,78.1328570842743,44624,235614,9,2.39735754242108,123837267,1326490.97180134,56.7113828659058,34031,73300,9,3.30276944065484,1306977,285928.605457026,73.9032800197601,32068,73300,9,3.30276944065484,1306977,285928.605457026
facebook-like,1,26.8673341274261,18184,23247,5,2.03148793392696,16736181,1709217.58463458,32.1139361858368,18104,23247,5,2.03148793392696,16736181,1709217.58463458,25.8974671363831,5111,2488,5,3.05747588424437,561961,277973.730305466,30.1288900375366,5066,2488,5,3.05747588424437,561961,277973.730305466
as-733,1,3103.52206683159,205821,4245340,12,2.51045452189931,67737600,128343.354023942,, , , , , , ,3193.38904500008,186918,1290050,12,3.68569590325956,7436,7428.2933746754,, , , , , , 
\midrule $\boldsymbol{k=2}$,,,,,,,,,,,,,,,,,,,,,,,,,,,,,
london,2,77.7697401046753,37141,55653,4,2.35509316658581,1439,948.794584299139,81.7953860759735,37141,55653,4,2.35509316658581,1439,948.794584299139,1.58037090301514,341,0,0,0,0,0,1.67599010467529,341,0,0,0,0,0
paris,2,117.2602350712,46428,76169,4,2.4127138337,1439,868.1294489884,116.2135088444,46428,76169,4,2.4127138337,1439,868.1294489884,2.0717499256134,390,0,0,0,0,0,2.24799108505249,390,0,0,0,0,0
ny,2,405.1012890339,88333,216582,4,2.6142431042,1439,581.2829828887,415.607929945,88333,216582,4,2.6142431042,1439,581.2829828887,6.08727502822876,579,0,0,0,0,0,6.64371085166931,579,0,0,0,0,0
highschool-2011,2,31.939437151,103916,48691,11,3.8604670268,272330,52871.072785525,31.232711792,87938,48691,11,3.8604670268,272330,52871.072785525,22.1002559661865,73577,12989,11,6.05535453075679,14618,9722.73385172069,21.7870590686798,59818,12989,11,6.05535453075679,14618,9722.73385172069
highschool-2012,2,33.3653450012,46342,51012,7,2.9306633733,729500,242292.701815259,35.2973759174,45696,51012,7,2.9306633733,729500,242292.701815259,10.4763629436493,7769,1869,7,5.11503477795613,27862,14737.8592830391,10.4806609153748,7633,1869,7,5.11503477795613,27862,14737.8592830391
highschool-2013,2,199.4550311565,177197,184404,10,3.0937181406,363560,106635.361543134,204.653826952,171866,184404,10,3.0937181406,363560,106635.361543134,49.7233240604401,52936,15637,10,5.38453667583296,5870,1825.08089787044,51.507150888443,49762,15637,10,5.38453667583296,5870,1825.08089787044
primaryschool,2,320.340541124344,581238,581212,11,3.81313186926629,116900,6755.73646105036,320.307415008545,561507,581212,11,3.81313186926629,116900,6755.73646105036,234.386074066162,302582,122076,11,5.44562403748485,3960,779.718372161604,234.407918930054,288722,122076,11,5.44562403748485,3960,779.718372161604
hospital-ward,2,55.8627550601959,194253,71854,13,4.91575973501823,347500,24500.520806079,57.7325620651245,173849,71854,13,4.91575973501823,347500,24500.520806079,52.4996798038483,168798,38529,13,6.23613382127748,36636,10251.9730073451,52.6835367679596,150078,38529,13,6.23613382127748,36636,10251.9730073451
hypertext,2,27.4098010063171,70856,66139,9,3.43703412510017,212340,29687.9488955079,27.7218241691589,68791,66139,9,3.43703412510017,212340,29687.9488955079,20.6437408924103,21228,6726,9,5.19134701159679,16548,8637.94617900684,21.3303859233856,20367,6726,9,5.19134701159679,16548,8637.94617900684
dnc,2,, , , , , , ,, , , , , , ,1016.51148295403,2100186,88045,18,6.91421432222159,4111834,2241829.56519961,730.060783863068,848903,88045,18,6.91421432222159,4111834,2241829.56519961
karlsruhe,2,, , , , , , ,, , , , , , ,665.11843585968,62642,31918,10,5.33222632997055,692711,264971.015351839,661.246778964996,60575,31918,10,5.33222632997055,692711,264971.015351839
facebook-like,2,, , , , , , ,, , , , , , ,297.504162073135,18726,425,6,5.02823529411765,349699,235921.447058824,310.664970874786,18722,425,6,5.02823529411765,349699,235921.447058824
\midrule $\boldsymbol{k=3}$,,,,,,,,,,,,,,,,,,,,,,,,,,,,,
london,3,, , , , , , ,, , , , , , ,3.5472891330719,693,0,0,0,0,0,3.68244910240173,693,0,0,0,0,0
paris,3,, , , , , , ,, , , , , , ,5.03234195709229,799,0,0,0,0,0,5.18124198913574,799,0,0,0,0,0
ny,3,, , , , , , ,, , , , , , ,21.0405480861664,1318,0,0,0,0,0,22.0299990177155,1318,0,0,0,0,0
highschool-2011,3,836.8415000439,1288366,1059245,13,4.0599917866,272330,90780.1064003134,850.4650108814,1233562,1059245,13,4.0599917866,272330,90780.1064003134,287.100867986679,355247,54717,13,7.7561087047901,13878,9149.67956942084,289.949823856354,321818,54717,13,7.7561087047901,13878,9149.67956942084
highschool-2012,3,2149.5472280979,1722247,1950811,8,3.5414112387,729500,366041.390903065,2168.4767839909,1721351,1950811,8,3.5414112387,729500,366041.390903065,191.913043022156,36195,430,8,7.02558139534884,20042,13332.2395348837,190.394711971283,36204,430,8,7.02558139534884,20042,13332.2395348837
highschool-2013,3,, , , , , , ,, , , , , , ,1215.41944980621,214144,21660,11,7.31560480147738,2670,1638.45775623269,1224.81869602203,210588,21660,11,7.31560480147738,2670,1638.45775623269
hospital-ward,3,450.933631896973,1023768,619778,14,5.24908434955742,347500,47623.328646709,478.005537986755,976275,619778,14,5.24908434955742,347500,47623.328646709,364.789851903915,547333,128059,14,7.75387907136554,28856,9327.85000663757,373.403419017792,516450,128059,14,7.75387907136554,28856,9327.85000663757
hypertext,3,805.626487016678,1211479,1315313,9,4.06902159409966,212340,46090.4464549503,842.882575035095,1205841,1315313,9,4.06902159409966,212340,46090.4464549503,488.321811914444,97808,5891,9,7.1120353080971,12628,8034.61245968426,481.175151109695,97711,5891,9,7.1120353080971,12628,8034.61245968426
\end{filecontents*}
\pgfplotstabletypeset[
   col sep=comma,
   columns={name,num_kplexes_conn,max_kplex_size_conn,runtime_conn,num_calls_conn,runtime_conn_pivo,num_calls_conn_pivo},
   columns/name/.style={column type=l,string type,column name={Instance}},
   columns/num_kplexes_conn/.style={column name={$\# R$},column type=r,precision=1},
   columns/max_kplex_size_conn/.style={column name={$|C|_{\max}$},column type=r,precision=1},
   columns/runtime_conn/.style={column name={$t$},column type=r,precision=0, relative*={0}},
   columns/num_calls_conn/.style={column name={$c$},column type=r,precision=0},
   columns/runtime_conn_pivo/.style={column name={$t_{\text{pivo}}$},precision=0, relative*={0}},
   columns/num_calls_conn_pivo/.style={column name={$c_{\text{pivo}}$},column type=r,precision=0},   
   every head row/.style={ before row=\toprule & \multicolumn{6}{c}{Connectedness Criterion for $\Delta \sim 5^5$ } \\ \cmidrule(r){2-7} ,after row=\midrule},
   every last row/.style={after row=\bottomrule},
   after row={},
   ]{results_d5.csv}
   \label{data_delta52}
 \end{table*}
 }
We present the results of \BKB{} and compare them to the findings of Himmel et al.~\cite{himmel2017adapting}.
We start with some expected findings and then continue with explaining some more interesting observations.

In \cref{data_delta01,data_delta31,data_delta51,data_delta02,data_delta32,data_delta52} (Appendix), one can find the number of $\Delta$-$k$-plexes, the running times, and the number of recursive calls of \BKB{} with and without pivoting and with and without the connectedness criterion on all our considered data sets for different~$\Delta$-values.
We found that increasing the~$\Delta$-value increases both the maximum size and the maximum lifetime of~$\Delta$-$k$-plexes for a fixed~$k$.
This is to be expected as each~$\Delta$-$k$-plex is also a~$(\Delta+1)$-$k$-plex.
Similarly to Himmel et al.~\cite{himmel2017adapting}, we found that for fixed~$k$ increasing the~$\Delta$-value first decreases and then increases the number of maximal~$\Delta$-$k$-plexes.

Since each~$k$-plex is also a~$(k+1)$-plex, it seems plausible that the number of~$k$-plexes increases with larger values of~$k$.
Note, however, that it might happen that two maximal~$k$-plexes merge into one maximal~${(k+1)}$-plex and so this number can actually decrease.
Ultimately, the number of new~$(k+1)$-plexes that are not~$k$-plexes outweighed the number of these merges in our experiments.
More interestingly, the number of connected~$\Delta$-$k$-plexes for a fixed~$k$ and large~$\Delta$ is significantly closer to the total number of~$\Delta$-$k$-plexes.
In contrast, the number of connected~$\Delta$-$k$-plexes for small~$\Delta$ is significantly smaller than the total number of~$\Delta$-$k$-plexes (by orders of magnitude).
We conjecture the reason to be that the higher the value of~$\Delta$ is, the more likely it becomes that there is an edge between two sets of vertices within a~$\Delta$-frame.


\section{Conclusion}
\label{sec:conclusion}
We introduced the \BKB~algorithm for enumerating all maximal~\dkps~in a temporal graph and studied its running time.
In experiments on real-world networks, we showed that our algorithm is faster when enumerating maximal $\Delta$-cliques, which are the same as~$\Delta$-$1$-plexes, than the algorithm by Himmel et al. \cite{himmel2017adapting} by an average factor between~$4$ and~$50$ depending on the~$\Delta$-value (typically greater speedups for smaller values of $\Delta$). 
We further showed that our algorithm performs better than the state-of-the-art algorithm by~\citet{ViardML18} on most instances but is also heavily out-performed on one of our instances. 
Explaining this behavior in terms of identifying properties of the instance that cause it is an interesting question for future research.

Our experiments also suggest that the number of trivial solutions for increasing $k$ greatly limits the scalability of any algorithm enumerating all maximal \dkps. 
Thus, we instead proposed to enumerate all maximal \emph{connected} $k$-plexes of minimum order~${2k+1}$. 
For this setting, we developed a heuristic for our algorithm which greatly improved its performance. 
We believe that there is still room for further heuristics to improve the scalability for larger $k$-values.

From a modelling perspective, it would be interesting to extend the definition of $\Delta$-cliques and \dkps\ in a way that allows to quantify how much each one should be \emph{isolated} from the remaining graph. This is a common requirement in community detection and has been formalized in the static case for cliques and $k$-plexes~\cite{HuffnerKMN09,MaxIsolatedCliques,Ito2009}. A natural future research direction is to lift the concept of \emph{isolation} to the temporal setting.

\section*{Acknowledgments} A-SH was supported by DFG, projects DAMM (NI 369/13) and FPTinP (NI 369/16), and HM was supported by DFG, project MATE (NI 369/17).

\bibliographystyle{abbrvnat}
\bibliography{references}

\clearpage
\section{Appendix: Further Tables and Diagrams}
\label{sec:appendix}
\appendixProofText

\end{document}